\documentclass[reprint, superscriptaddress, prx,
noeprint]{revtex4-2}

\usepackage[margin=0.7in]{geometry}

\usepackage[caption=false]{subfig}
\usepackage{graphicx} 
\usepackage{epstopdf}
\usepackage{array}
\usepackage{verbatim}
\usepackage{amsmath,amsfonts,amssymb,amscd}
\usepackage{amsthm}
\usepackage{tabularx}
\usepackage{enumerate}
\usepackage{enumitem}
\usepackage{wasysym}
\usepackage[dvipsnames]{xcolor}
\usepackage[normalem]{ulem}

\usepackage{comment}
\usepackage{color}
\usepackage{stmaryrd}
\usepackage{braket}
\usepackage{multirow}
\usepackage{bbm,xcolor}
\usepackage{algorithm}
\usepackage{algpseudocode}
\usepackage{tikz}
\usetikzlibrary{decorations}
\usetikzlibrary{decorations.pathreplacing,decorations.pathmorphing,calligraphy}
\usepackage{blkarray}

\newcommand{\rank}[1]{\mathrm{rank}(#1)}

\theoremstyle{plain}
\newtheorem{theorem}{Theorem}
\newtheorem{corollary}[theorem]{Corollary}
\newtheorem{lemma}[theorem]{Lemma}

\makeatletter
\newtheorem*{rep@theorem}{\rep@title}
\newcommand{\newreptheorem}[2]{%
\newenvironment{rep#1}[1]{%
 \def\rep@title{#2 \ref{##1}}%
 \begin{rep@theorem}}%
 {\end{rep@theorem}}}
\makeatother

\newreptheorem{theorem}{Theorem}
\newreptheorem{lemma}{Lemma}

\theoremstyle{definition}
\newtheorem{definition}[theorem]{Definition}
\newtheorem{remark}[theorem]{Remark}

\usepackage{hyperref}
\hypersetup{
    bookmarksnumbered=true, 
    unicode=false, 
    pdfstartview={FitH}, 
    pdftitle={}, 
    pdfauthor={}, 
    pdfsubject={}, 
    pdfcreator={}, 
    pdfproducer={}, 
    pdfkeywords={}, 
    pdfnewwindow=true, 
    colorlinks=true, 
    linkcolor=blue, 
    citecolor=blue, 
    filecolor=blue, 
    urlcolor=blue 
}

\newcommand{\Gate}[1]{\mathsf{#1}}

\newcommand{\cnotgate}{\Gate{CNOT}}

\newcommand{\lx}{\mathcal{L}^{X}}
\newcommand{\lz}{\mathcal{L}^{Z}}
\newcommand{\vx}{\mathcal{V}^{X}}
\newcommand{\vz}{\mathcal{V}^{Z}}
\newcommand{\qx}{\mathcal{Q}^{X}}
\newcommand{\qz}{\mathcal{Q}^{Z}}
\newcommand{\cx}{\mathcal{C}^{X}}
\newcommand{\cz}{\mathcal{C}^{Z}}
\newcommand{\qcal}{\mathcal{Q}}
\newcommand{\ccal}{\mathcal{C}}

\DeclareMathOperator{\supp}{supp}

\newcommand{\inlinemod}[1]{(\mathrm{mod}\text{\space}#1)}

\begin{document}
\title{Universal adapters between quantum LDPC codes
}

\author{Esha Swaroop}
\email{eswaroop@uwaterloo.ca}
\affiliation{\small Institute for Quantum Computing and Department of Physics and Astronomy, University of Waterloo, Waterloo, ON N2L 3G1, Canada}
\affiliation{\small Perimeter Institute for Theoretical Physics, Waterloo, ON N2L 2Y5, Canada}
\author{Tomas Jochym-O'Connor}
\email{tjoc@ibm.com}
\affiliation{\small IBM T. J. Watson Research Center, Yorktown Heights, NY, 10598, United States}
\affiliation{\small IBM Quantum, Almaden Research Center, San Jose, CA, United States}
\author{Theodore J Yoder}
\email{ted.yoder@ibm.com}
\affiliation{\small IBM T. J. Watson Research Center, Yorktown Heights, NY, 10598, United States}

\begin{abstract}
We propose the repetition code adapter as a way to perform joint logical Pauli measurements within a quantum low-density parity check (LDPC) codeblock or between separate such codeblocks, thus providing a flexible tool for fault-tolerant computation with quantum LDPC codes. This adapter is universal in the sense that it works regardless of the LDPC codes involved and the logical Paulis being measured. The construction achieves joint logical Pauli measurement of $t$ weight $O(d)$ operators using $O(d)$ time and $\tilde O(td)$ additional qubits and checks, up to a factor polylogarithmic in $d$.
As a special case, for some geometrically-local codes in fixed $D\ge2$ dimensions, only $O(td)$ additional qubits and checks are required instead. By extending the adapter in the case $t=2$, we also construct a toric code adapter that uses $O(d^2)$ additional qubits and checks to perform addressable logical CNOT gates on arbitrary LDPC codes via Dehn twists. To obtain these results, we develop a novel weaker form of graph edge expansion and the $\mathsf{SkipTree}$ algorithm, which ensures a sparse transformation between different weight-2 check bases for the classical repetition code.
\end{abstract}

\maketitle

\section{Introduction} \label{sec:intro}

The long-term promise of quantum computing and quantum algorithms will rely on the backbone of quantum error correction~(QEC) and fault tolerance. While early experiments focused on demonstrating the building blocks of QEC~\cite{andersen2020repeated,google2021exponential,ryan2021realization,marques2022logical,sundaresan2023demonstrating}, recently there has been increased focus on scaling up the \textit{distance}, the ability to protect against errors, and \textit{encoding rate}, the relative number of encoded logical qubits, of QEC codes~\cite{acharya2024quantum,bluvstein2024logical}. This is coupled with a line of theoretical research aimed towards increasing QEC~code parameters~\cite{tillich2013quantum,kovalev2013,leverrier2015quantum}, culminating with the recent discovery of \emph{good} quantum low-density parity check~(LDPC) codes~\cite{breuckmann2021balanced,panteleev2022asymptotically,leverrier2022quantum}. Quantum LDPC codes are also of experimental interest, in the hope of potentially simplifying requirements of physical systems to realize fault-tolerant quantum computation, since LDPC codes by design require low-weight measurements and limited qubit connectivity. 

Although much of the progress in quantum LDPC codes has centered on improving code parameters, it is equally important to establish a model for doing logical computation in these codes. One of the leading approaches for logical gates is that of lattice surgery in the surface code~\cite{horsman2012surface,lingling2018surfacecode}, where an additional surface code patch can be prepared and fused to the original code, allowing for the joint-measurement of their respective logical operators. Given the success of this approach~\cite{litinski2019game} with the surface code which itself is an LDPC code, it was natural to ask if a similar type of approach can be developed for more general LDPC codes. Inspired from ideas in weight reduction of quantum codes \cite{hastings2016weight,hastings2021quantum}, this line of research led to recent schemes for quantum LDPC surgery~\cite{cohen2022,cowtan2024css, cross2024linear,cowtan2024ssipautomatedsurgery, zhang2024time,williamson2024gauging,ide2024faulttolerant} to implement logical Pauli measurements. Logical Pauli measurements are sufficient to implement the Clifford group~\cite{smithsmolin2015pbc} and offer a route to universal quantum computation provided magic states are available.

The key idea of quantum code surgery schemes is to deform the original code into a larger intermediate (`deformed') code by introducing additional qubits and checks, such that the deformed code contains all the logical qubits of the code except the logical qubit corresponding to the operator being measured, as well as newly introduced gauge degrees of freedom if the number of additional qubits exceeds the number of independent new checks. The stabilizer group of the deformed code (or gauge group in case of a deformed subsystem code) contains the high-weight logical operator to be measured. By measuring low-weight stabilizers (or gauge operators) of the deformed code, the measurement outcome of a high-weight logical operator of the original code can be inferred. 
In Ref.~\cite{williamson2024gauging}, this general scheme is referred to as the \textit{gauging logical measurement} framework, within which the authors introduced an optimized logical measurement scheme through use of expander graph-based ancillary systems, which we briefly summarize now. The goal is to measure an arbitrary logical Pauli operator on an arbitrary LDPC code. The main idea is to design an appropriate auxiliary graph~$\mathcal{G}$, where ancilla qubits reside on edges and new stabilizer checks are associated to both vertices and cycles in this graph. This auxiliary graph is merged with the original code by making a subset of the vertex stabilizers interact with the qubits in the support of the logical operator being measured, see Fig~\ref{fig:summary}{\small a}. As a result, the stabilizers of the original code are deformed, and one needs to make careful choices to avoid deforming into a non-LDPC code or suffering a deformed code with reduced code distance. To preserve the LDPC nature of stabilizers as well as code distance will necessitate several nontrivial properties of the auxiliary graph $\mathcal{G}$,
including the existence of certain short perfect matchings, the existence of a suitably sparse cycle basis, and edge expansion, which we will review afresh in detail. Due to the overloading of the word ``gauge", we will henceforth refer to gauging logical measurement as \textit{auxiliary graph surgery}.

\begin{figure}[t]
\centering
\subfloat[Gauging logical measurement via an auxiliary graph]{
\hspace{0.5in}
\includegraphics[width=0.37\linewidth]{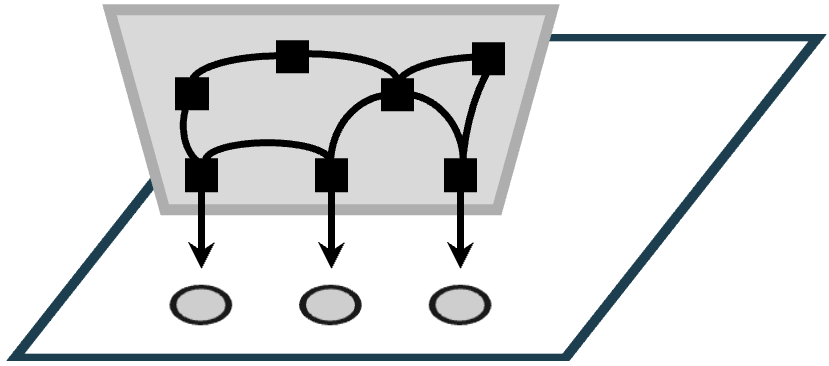}
}
\newline
\subfloat[Repetition code adapters for joint logical measurements]{
    \centering
\includegraphics[width=\linewidth]{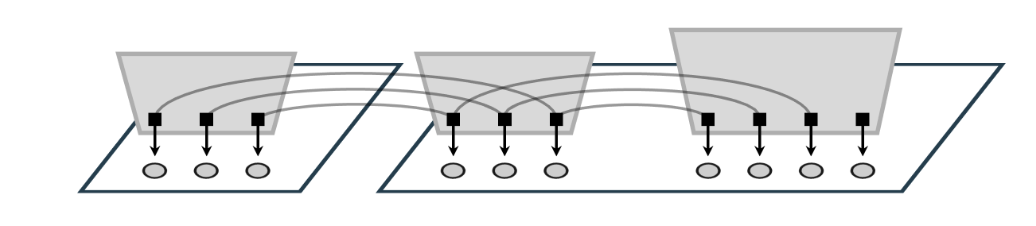}
         }
\newline
\subfloat[Toric code adapter for logical gates]{
    \includegraphics[width=0.8\linewidth]{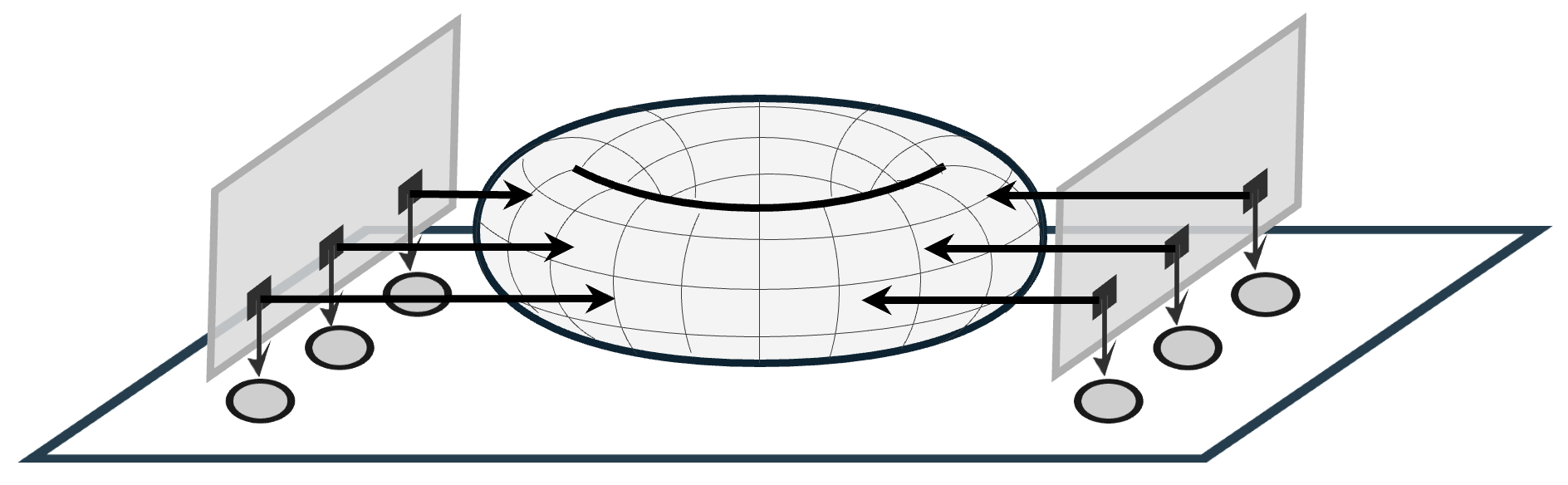} 
    }
     \caption{(a) For an arbitrary LDPC code (drawn as a rectangular patch), gauging logical measurement \cite{williamson2024gauging} of a logical Pauli operator works by attaching a stabilizer state defined on an appropriate auxiliary graph (gray area) to the qubit support of the logical operator (gray circles) creating a deformed code in which the logical operator becomes a stabilizer. Only some graph vertices, which represent checks, connect 1-to-1 to the logical support. This set of vertices is referred to as a \emph{port}. In this work, we construct additional tools: (b) The ports of several auxiliary graphs can be connected together with carefully chosen \emph{adapters} (curved edges) while keeping the deformed code LDPC. This measures the product of logical Pauli operators without measuring any individually. Adapters can connect operators in the same codeblock or separate blocks and regardless of their structure. Mathematically, our adapter construction relies on a sparse basis transformation for the classical repetition code, which we call the $\mathsf{SkipTree}$ algorithm. (c) Similarly, an arbitrary LDPC code can be merged with a toric code (or other codes) to perform logical gates.}
     \label{fig:summary}
\end{figure}


Our main result addresses a practical problem in auxiliary graph surgery to perform Pauli product measurements. Because computing with quantum LDPC surgery naturally lends itself to Pauli-based computation \cite{Bravyi_2016} (analogously to surface code lattice surgery \cite{litinski2019game}), one naively expects the need to measure potentially any logical Pauli operator and thus that exponentially many auxiliary graphs with suitable properties will need to be constructed. This task can be greatly simplified if suitably constructed auxiliary graphs to measure individual logical operators, say those on single logical qubits, could simply be connected in some way to measure the product of those operators instead while maintaining the necessary graph properties. For logical operators with isomorphic auxiliary subgraphs, this connection can be done directly \cite{cowtan2024ssipautomatedsurgery}, but it is a priori unclear how to do these joint measurements more generally.

Here, we solve this joint measurement problem in very general setting. Our solution is universal in the sense that it works regardless of the structure of the codes or the logical operators involved, provided only that a constant number of these logical operators act on any one code qubit and they do so with the same type of Pauli operator. We are able to adapt the structure of any auxiliary graph to any other, connecting them via a set of \emph{adapter} edges into one large graph, see Fig.~\ref{fig:summary}b. This builds on the idea of a bridge system from Ref.~\cite{cross2024linear}, but with new tools to guarantee the resulting deformed code for the product measurement is LDPC and suffers no loss of code distance. We refer to the construction in this work as a \textit{universal adapter}.

Auxiliary graph surgery and adapters necessitate additional qubits and thus increase the space overhead for fault-tolerant computation. Gauging logical measurement (henceforth also referred to as \textit{auxiliary graph surgery}) of a single operator of weight $d$ uses $O(d\log^3d)$ additional qubits in general, and connecting these systems with our adapters to measure the product of $t$ such operators takes $O(td\log^3d)$ additional qubits. Being quasilinear in $d$ like Ref.~\cite{williamson2024gauging}, this method offers a quadratic improvement in space overhead over the previous method for logical measurements~\cite{cohen2022} on arbitrary quantum LDPC codes with overhead $O(d^2)$. Since an overhead of $O(d^2)$ is comparable to fault-tolerant quantum computing schemes that rely on the surface code, a method for logic with quasilinear qubit overhead brings to fruition the advantages of higher-rate quantum codes. This all comes with the same time-overhead as surface code computation lattice surgery: auxiliary graph surgery and our adapter variants all use $O(d)$ time to sufficiently deal with measurement errors.

We remark that the $\log^3d$ factor in qubit overhead arises due to an extra \textit{decongestion} step in the construction to ensure that the auxiliary graph has a suitably sparse cycle basis. It turns out that for the special case of geometrically-local LDPC codes,  one can construct an auxiliary graph that does not require decongestion. This auxiliary graph can be realized by choosing a suitable triangulation of the set of points representing a logical operator, namely the Delaunay triangulation~\cite{delaunay}, thereby improving the overhead of auxiliary graph surgery to $O(d)$ in this special case. 

Our adapters can be further modified to couple to other ancillary systems with other desirable properties. To illustrate, we present a method to merge an arbitrary quantum LDPC code with the toric code along the supports of two disjoint logical operators, Fig.~\ref{fig:summary}c. From there, a unitary circuit suffices to implement a logical $\cnotgate$ using a method inspired by the Dehn twist $\cnotgate$ from Refs.~\cite{koenig2010turaev,breuckmann2017hyperbolic,lavasani2019dehn}. Although the toric code adapter is asymptotically inefficient in that it uses $O(d^2)$ additional qubits, it has potential advantages at finite size: (1) these interface systems are alone smaller than existing interfaces between LDPC and surface codes \cite{xu2023constantreconfig,viszlai2024matchingneutral}, $O(d\log^3d)$ qubits compared with $O(d^2)$, and (2) it directly performs a logical $\cnotgate$ gate rather than building the gate through logical measurements. This is an explicit example of how adapters can be used as a tool for mapping between codes with different properties, such as those with high rate for space-efficient memory and those with additional symmetries enabling fault-tolerant logical computation. Although prior work has explored multi-code architectures in topological codes \cite{poulsen2017topologicalinterface}, the universal adapter facilitates multi-code architectures for generic LDPC codes, which are a wider class of codes and require non-trivial deformation strategies.

In Section~\ref{sec:prelim} we review the basics of stabilizer codes and their representation as Tanner graphs. We also review auxiliary graph surgery \cite{williamson2024gauging} and LDPC surgery concepts from Ref.~\cite{cross2024linear}, but through the lens of a novel definition of \emph{relative expansion} which unifies and simplifies some ideas from these prior works. In Section~\ref{sec:skiptree}, we present an efficient classical algorithm for transforming between different check bases for the classical repetition code, which will be of use in creating our adapters. In Section~\ref{sec:rep_adapter}, we present the repetition code adapter as a \textit{universal adapter} for measuring joint logical operators in general quantum LDPC codes that works by connecting individual graphs with relative expansion into one large graph guaranteed to have relative expansion. Section~\ref{sec:leveraging_code_properties} re-examines these general results for special scenarios where one can leverage code properties to simplify the scheme -- in Section~\ref{subsec:replace_desideratum_4}, we exhibit cases where the expansion requirement in the auxiliary graph is unnecessary and, in Section~\ref{subsec:geo_local}, we show the decongestion step is unnecessary for geometrically-local codes by providing a method for constructing auxiliary graphs that are naturally sparse. Next, Section~\ref{sec:multicode} broadens the scope of the universal adapter as a novel primitive to connect different quantum LDPC codes in multi-code architectures. Finally, in Section~\ref{sec:examples} we provide some concrete examples of the universal adapter construction applied to two specific quantum LDPC codes, a 200-qubit lifted-product code and a 98-qubit bivariate bicycle code.

\subsection*{Related Work}
This work builds on both Ref.~\cite{cross2024linear} and Ref.~\cite{williamson2024gauging}. In Ref.~\cite{cross2024linear}, the authors proposed \textit{gauge-fixing} the merged subsystem code introduced in Ref.~\cite{cohen2022}, to implement logical measurements on quantum LDPC codes with sufficient expansion. Ref.~\cite{williamson2024gauging} proposed moving this expansion requirement to auxiliary graphs that can be connected to any quantum LDPC code, thereby broadening the scope to perform logical Pauli measurements on arbitrary quantum LDPC codes. Ref.~\cite{cross2024linear} in particular introduced a `bridge' as a way to connect different ancilla systems to measure joint logical operators, which inspired our work. Here, we
guarantee that constructing this adapter in a particular way will not ruin the LDPC property of the
code, thereby assuring practicality of the scheme. To do so, the $\mathsf{SkipTree}$ algorithm relies on the perspective shared by Refs.~\cite{williamson2024gauging} that the ancilla qubits
and checks reside on a graph instead of a more general hypergraph. Finally, Ref.~\cite{cross2024linear,williamson2024gauging,ide2024faulttolerant}
all remark that graph or hypergraph edge expansion is sufficient for their surgery techniques to preserve the
code and fault distances. Here, we generalize that to a weaker requirement by introducing an appropriate notion of relative expansion
instead or, in certain cases, eschew this expansion requirement entirely.

\tableofcontents

\section{Preliminaries} \label{sec:prelim}

\subsection{Stabilizer codes and Tanner graphs}\label{subsec:notation_tanner_graph} 

An $n$-qubit hermitian Pauli operator can be written $i^{u\cdot v}X(u)Z(v)$, where $u,v\in\mathbb{F}_2^n$, or simply as a symplectic vector $[u|v]$. Vectors in this paper are always by default row vectors and must be transposed, e.g.~$v^\top$, to obtain a column vector. If $v=0$ (resp.~$u=0$), we say the Pauli is $X$-type (resp.~$Z$-type). Several Pauli operators written in symplectic notation can be gathered together as the rows of a symplectic check matrix
\begin{equation}\label{eq:stabilizer_checks}
H=[H_X|H_Z]\in\mathbb{F}_2^{\;r\times2n},\;\text{such that}\;H\left(\begin{array}{cc}0&I\\I&0\end{array}\right)H^\top=0,
\end{equation}
where we use $I$ to denote the identity matrix of context-dependent size (here $n$), and the condition guarantees that all Pauli operators in the set commute with one another. Thus, $H$ describes the checks of a stabilizer code \cite{gottesman1997Stabilizer}. 

Logical operators of a stabilizer code are Pauli operators that commute with all the checks. All the checks are logical operators also. Logical operators that are not a product of checks are called nontrivial. The code distance is the minimum Pauli weight of any nontrivial logical operator.

Pauli Logical operators are denoted as $\bar{X}$ or $\bar{Z}$. A superscript, such as $\bar{Z}^{(i)}$ indicates the logical operator corresponding to the logical qubit indexed $i$. A subscript, such as $\bar{Z}_{l}$ indicates which code (block) the operator belongs to. For instance, $l$ or $r$ could be used to indicate left or right code blocks.

We can illustrate a stabilizer code by drawing a Tanner graph. This is a bipartite graph containing a vertex for each qubit and each check (a row of $H$). A qubit is connected to each check in which it participates with an edge labeled $[1|0]$, $[0|1]$, or $[1|1]$ depending on whether the check acts on the qubit as $X$, $Z$, or $Y$.

However, in cases with where multiple qubits and checks share identical check matrices, drawing a Tanner graph with a single node for every individual qubit and check could be a needlessly verbose approach. We abstract away the detail by drawing Tanner graphs with qubits gathered into named sets, say $\mathcal{Q}_0$, $\mathcal{Q}_1$, $\mathcal{Q}_2$,\dots, and checks gathered into named sets, say $\mathcal{C}_0$, $\mathcal{C}_1$, $\mathcal{C}_2$, \dots. An edge is drawn between $\mathcal{C}_i$ and $\mathcal{Q}_j$ if any check from $\mathcal{C}_i$ acts on any qubit in $\mathcal{Q}_j$. Label the edge with a symplectic matrix $[C_X|C_Z]\in\mathbb{F}_2^{|\mathcal{Q}_j|\times2|\mathcal{C}_i|}$ indicating which checks act on which qubits and with what type of Pauli operator. For instance, the entire stabilizer code from Eq.~\eqref{eq:stabilizer_checks} is drawn as in Fig.~\ref{fig:Tanner_graph_examples}a.

\begin{figure*}[t]
        \begin{tabular}{lll}
(a) &  & (b) \\
 \begin{tikzpicture}[]
\pgfmathsetmacro{\ewid}{0.7pt}
\pgfmathsetmacro{\drwid}{0.7pt}
\pgfmathsetmacro{\sqrsz}{0.65cm}
\pgfmathsetmacro{\sqrdf}{0.12}
\pgfmathsetmacro{\crcsz}{0.36}
\pgfmathsetmacro{\crcdf}{0.08}

        \filldraw[color=white] (-0.6, -0.4) rectangle (3.6, 0.8);
        \draw[draw=gray,line width=0.6pt] (0,0) to (3,0);
        
        \node[draw, line width=\drwid, fill=white, minimum size=\sqrsz] at (0-\sqrdf,0+\sqrdf) {};
            \node[draw, line width=\drwid, fill=white,minimum size=\sqrsz] at (0, 0) {};
       \draw[fill=white,line width=\drwid] (3-\crcdf,0+\crcdf) circle (\crcsz);
            \draw[fill=white,line width=\drwid] (3,0) circle (\crcsz) node {};

        \node[] at (1.5, 0.3) {$[H_X|H_Z]$};

\end{tikzpicture}
&
\hspace{1in}
&
\begin{tikzpicture}[]
\pgfmathsetmacro{\ewid}{0.7pt}
\pgfmathsetmacro{\drwid}{0.7pt}
\pgfmathsetmacro{\sqrsz}{0.65cm}
\pgfmathsetmacro{\sqrdf}{0.12}
\pgfmathsetmacro{\crcsz}{0.36}
\pgfmathsetmacro{\crcdf}{0.08}

        \filldraw[color=white] (-0.6, -0.4) rectangle (5, 0.8);
        \draw[draw=gray,line width=0.6pt] (0,0) to (4,0);
        
        \node[draw, line width=\drwid, fill=white, minimum size=\sqrsz] at (0-\sqrdf,0+\sqrdf) {};
            \node[draw, line width=\drwid, fill=white,minimum size=\sqrsz] at (0, 0) {$X$};
        \draw[fill=white,line width=\drwid] (2.2-\crcdf,0+\crcdf) circle (\crcsz);
            \draw[fill=white,line width=\drwid] (2.2,0) circle (\crcsz) node {};
        \node[draw, line width=\drwid, fill=white, minimum size=\sqrsz] at (4.4-\sqrdf,0+\sqrdf) {};
            \node[draw, line width=\drwid, fill=white,minimum size=\sqrsz] at (4.4, 0) {$Z$};

        \node[] at (1.1, 0.3) {$H_X$};
        \node[] at (3.3, 0.3) {$H_Z$};

\end{tikzpicture}
\end{tabular}

\begin{tikzpicture}[]
\pgfmathsetmacro{\ewid}{0.7pt}
\pgfmathsetmacro{\drwid}{0.7pt}
\pgfmathsetmacro{\sqrsz}{0.75cm}
\pgfmathsetmacro{\sqrdf}{0.12}
\pgfmathsetmacro{\crcsz}{0.4}
\pgfmathsetmacro{\crcdf}{0.1}
\pgfmathsetmacro{\diff}{2.2cm}
\pgfmathsetmacro{\hgt}{2.1}
\pgfmathsetmacro{\hdf}{0.8}
\pgfmathsetmacro{\hdn}{0.7}
\pgfmathsetmacro{\brnd}{4}
\definecolor{bl}{rgb}{0.63, 0.79, 0.95}

\filldraw[color=white] (-1, 0) rectangle (10, 2.8);

        \draw (0, \hgt+0.5) arc[start angle=110, end angle=68, x radius=4*\brnd, y radius=4.15*\brnd];
        \draw (0.1, -0.41) arc[start angle=-110, end angle=-68, x radius=4*\brnd, y radius=4.1*\brnd];
        \node[] at (1, 0.3) {$I$};
        \node[] at (0.8, \hgt+1.1) {$I$};
        \node[] at (10.5, -0.9) {$I$};
        
        \draw[draw=gray,line width=\ewid] (0,0) to (7.4,0);
        \draw[draw=gray,line width=\ewid] (0,\hgt) to (7.4,\hgt);
        \draw[draw=gray,line width=\ewid] (8.3,0) to (11,0);
        \draw[draw=gray,line width=\ewid] (8.3,\hgt) to (11,\hgt);

        \begin{scope}
        
            \draw[draw=gray,line width=\ewid] (0,0) to (0,2); 
            \node[draw, line width=\drwid, fill=white, minimum size=\sqrsz] at (0-\sqrdf,\hgt+\sqrdf) {};
            \node[draw, line width=\drwid, fill=white,minimum size=\sqrsz] at (0, \hgt) {$X$};
            \draw[fill=white,line width=\drwid] (0-\crcdf,0+\crcdf) circle (\crcsz);
            \draw[fill=white,line width=\drwid] (0,0) circle (\crcsz) node {};
            \node[] at (-0.3, 1) {$H_C$};
            \node[] at (0, -\hdn) {$\mathcal{Q}^Z_0$};
            \node[] at (0,\hgt+\hdf) { $\mathcal{C}^X_0$};
            \node[] at (1, 0.3+\hgt) {$I$};
            \node[] at (1, 0.3) {$I$};
        \end{scope}
       
        \begin{scope}[xshift=\diff]
            \draw[draw=gray,line width=\ewid] (0,0) to (0,2);
            \draw[fill=white,line width=\drwid] (0-\crcdf,\hgt+\crcdf) circle (\crcsz);
            \draw[fill=white,line width=\drwid] (0,\hgt) circle (\crcsz) node {};
            \node[draw, line width=\drwid, fill=white, minimum size=\sqrsz] at (0-\sqrdf, 0+\sqrdf) {};
            \node[draw, line width=\drwid, fill=white,minimum size=\sqrsz] at (0, 0) {$Z$};
            \node[] at (-0.3, 1) {$H_C^{\top}$};
            \node[] at (0, -\hdn) {$\mathcal{C}^Z_0$};
            \node[] at (0,\hgt+\hdf) {$\mathcal{Q}^X_0$}; 
            \node[] at (1, 0.3+\hgt) {$I$};
            \node[] at (1, 0.3) {$I$};
        \end{scope}

        \begin{scope}[xshift=2*\diff]
            \draw[draw=gray,line width=\ewid] (0,0) to (0,2); 
            \node[draw, line width=\drwid, fill=white, minimum size=\sqrsz] at (0-\sqrdf,\hgt+\sqrdf) {};
            \node[draw, line width=\drwid, fill=white,minimum size=\sqrsz] at (0, \hgt) {$X$};
            \draw[fill=white,line width=\drwid] (0-\crcdf,0+\crcdf) circle (\crcsz);
            \draw[fill=white,line width=\drwid] (0,0) circle (\crcsz) node {};
            \node[] at (-0.3, 1) {$H_C$};
            \node[] at (0, -\hdn) {$\mathcal{Q}^Z_1$};
            \node[] at (0,\hgt+\hdf) {$\mathcal{C}^X_1$}; 
            \node[] at (1, 0.3+\hgt) {$I$};
            \node[] at (1, 0.3) {$I$};
        \end{scope}

        \begin{scope}[xshift=3*\diff]
            \draw[draw=gray,line width=\ewid] (0,0) to (0,2);
            \draw[fill=white,line width=\drwid] (0-\crcdf,\hgt+\crcdf) circle (\crcsz);
            \draw[fill=white,line width=\drwid] (0,\hgt) circle (\crcsz) node {};
            \node[draw, line width=\drwid, fill=white, minimum size=\sqrsz] at (0-\sqrdf, 0+\sqrdf) {};
            \node[draw, line width=\drwid, fill=white,minimum size=\sqrsz] at (0, 0) {$Z$};
            \node[] at (-0.3, 1) {$H_C^{\top}$};
            \node[] at (0, -\hdn) {$\mathcal{C}^Z_1$};
            \node[] at (0,\hgt+\hdf) {$\mathcal{Q}^X_1$};  
        \end{scope}

        \begin{scope}[xshift=3.5*\diff+3]
        \node[] at (0, 0) {{ $...$}};
        \node[] at (0, 2) {{ $...$}};
        \end{scope}

        \begin{scope}[xshift=4*\diff+10]
            \draw[draw=gray,line width=\ewid] (0,0) to (0,2); 
            \node[draw, line width=\drwid, fill=white, minimum size=\sqrsz] at (0-\sqrdf,\hgt+\sqrdf) {};
            \node[draw, line width=\drwid, fill=white,minimum size=\sqrsz] at (0, \hgt) {$X$};
            \draw[fill=white,line width=\drwid] (0-\crcdf,0+\crcdf) circle (\crcsz);
            \draw[fill=white,line width=\drwid] (0,0) circle (\crcsz) node {};
            \node[] at (-0.3, 1) {$H_C$};
            \node[] at (0, -\hdn) {$\mathcal{Q}^Z_{d-1}$};
            \node[] at (0,\hgt+\hdf) {$\mathcal{C}^X_{d-1}$}; 
            \node[] at (1, 0.3+\hgt) {$I$};
            \node[] at (1, 0.3) {$I$};
        \end{scope}

        \begin{scope}[xshift=5*\diff+10]
            \draw[draw=gray,line width=\ewid] (0,0) to (0,2);
            \draw[fill=white,line width=\drwid] (0-\crcdf,\hgt+\crcdf) circle (\crcsz);
            \draw[fill=white,line width=\drwid] (0,\hgt) circle (\crcsz) node {};
            \node[draw, line width=\drwid, fill=white, minimum size=\sqrsz] at (0-\sqrdf, 0+\sqrdf) {};
            \node[draw, line width=\drwid, fill=white,minimum size=\sqrsz] at (0, 0) {$Z$};
            \node[] at (-0.3, 1) {$H_C^{\top}$};
            \node[] at (0, -\hdn) {$\mathcal{C}^Z_{d-1}$};
            \node[] at (0,\hgt+\hdf) {$\mathcal{Q}^X_{d-1}$};  
        \end{scope}

\end{tikzpicture}

\caption{Example Tanner graphs with qubit sets drawn as circles and check sets drawn as squares. (a) A generic stabilizer code, (b) a CSS code, and (c) the distance-$d$ toric code. In the toric code, each set of qubits and checks contains $d$ objects.}
     \label{fig:Tanner_graph_examples}
\end{figure*}

For CSS codes \cite{calderbank1996good,steane1996error}, there is a basis of checks in which each check is either $X$-type or $Z$-type. In the Tanner graph of such a code, if a set of checks $\mathcal{C}_i$ is entirely $X$-type, instead of labeling each edge with a symplectic matrix like $[C_X|0]$, we simply label the check set with $X$ and each edge with $C_X$ only. We handle sets of $Z$-type checks analogously. The Tanner graph of a CSS code is shown in Fig.~\ref{fig:Tanner_graph_examples}b.

We can alternatively view sets $\mathcal{C}_i$ and $\mathcal{Q}_j$ as vector spaces instead and use the notation to denote both the sets as well as associated vector spaces with these symbols. For example, a vector $u\in\mathbb{F}_2^{|\mathcal{C}_i|}$ indicates a subset of checks from $\mathcal{C}_i$ by its nonzero elements. We write $\mathcal{H}(u\in\mathcal{C}_i)$ to denote the product of this subset of checks, i.e.~it is an $n$-qubit Pauli operator. If all checks in the set are $X$-type or $Z$-type, we write $\mathcal{H}_X(u\in\mathcal{C}_i)$ or $\mathcal{H}_Z(u\in\mathcal{C}_i)$ instead. Likewise, a Pauli operator on qubit set $\mathcal{Q}_j$ is denoted by $X(v_x\in\mathcal{Q}_j)Z(v_z\in\mathcal{Q}_j)$ for appropriate vectors $v_x,v_z\in\mathbb{F}_2^{|\mathcal{Q}_j|}$. 

We can perform calculations in matrix-vector multiplication over $\mathbb{F}_2$ with this notation. For instance, if the check set $\mathcal{C}_i$ is connected to only a single qubit set $\mathcal{Q}_j$ with edge labeled $[C_X|C_Z]$, then
\begin{equation}
\mathcal{H}(u\in\mathcal{C}_i)=X(u C_X\in\mathcal{Q}_j)Z(u C_Z\in\mathcal{Q}_j).
\end{equation}

If a family of stabilizer codes with growing code size $n$ is $(\alpha,\beta)$ low-density parity-check (LDPC), it has a basis of checks in which each check acts on at most $\alpha=O(1)$ qubits and each qubit is acted upon by at most $\beta=O(1)$ checks. This translates to the sparsity of the code's parity check matrix $H=[H_X|H_Z]$. We say a matrix over $\mathbb{F}_2$ is $(r,c)$-sparse if the maximum row weight is at most $r$ and the maximum column weight is at most $c$. If both $H_X$ and $H_Z$ are $(r,c)$-sparse then the code is $(2r,2c)$ LDPC.

One outstanding family of quantum LDPC codes is the toric code family \cite{kitaev2003fault}, depicted as a Tanner graph in Fig.~\ref{fig:Tanner_graph_examples}c. There, we use $H_C$ to denote the canonical cyclic parity check matrix of the classical repetition code, i.e.~
\begin{equation}
H_C=\left(\begin{array}{ccccc}
1&1&&\dots&\\
&1&1&&\\
&\vdots&&\ddots&\\
1&&&&1
\end{array}\right).
\end{equation}
Denote by $e_i$ the length-$n$ vector with a 1 in only the $i\;(\text{mod}\;n)$ position. We can define a cyclic shift matrix $C\in\mathbb{F}_2^{\;n\times n}$ that acts as $e_iC=e_{i+1}$ for all $i$. Then, $H_C=I+C$. The only vector in the nullspace of $H_C$ is the vector of all 1s, denoted $\vec1$, so $\vec1H_C=H_C\vec1^{\;\top}=0$. In general, the sizes of $e_i$, $\vec1$, $I$, $C$, and $H_C$ are context dependent.

\subsection{Graphs and expansion}
Let $\mathcal{G}=(\mathcal{V},\mathcal{E})$ signify a graph with vertex set $\mathcal{V}$ of size $n$ and an edge set $\mathcal{E}$ of size $m$. An alternative description is provided by the incidence matrix $G=\mathbb{F}_2^{\;m\times n}$, a matrix in which each row represents an edge, each column represents a vertex, and $G_{ij}=1$ if and only if edge $i$ contains vertex $j$. If the maximum vertex degree of $\mathcal{G}$ is $w$, then $G$ is $(2,w)$-sparse.

If $\mathcal{G}$ has $p$ connected components, a complete cycle basis can be specified by a matrix $N$ over $\mathbb{F}_2$ satisfying $NG=0$ and $\rank{N}=m-n+p$, known as the cyclomatic number of the graph \cite{berge2001theory,graphs05tb}. We say the cycle basis is $(r,c)$-sparse if $N$ is an $(r,c)$-sparse matrix. In an $(r,c)$-sparse cycle basis, each basis cycle is no longer than $r$ edges and each edge is in no more than $c$ basis cycles.

The edge expansion of a graph is characterized via its Cheeger constant. Intuitively, this is a measure of how bottlenecked a graph is -- if a graph contains a large set of vertices with very few outgoing edges, the Cheeger constant of that graph is small. 
\begin{definition}[Cheeger constant]\label{def:expansion}
Let $\mathcal{G}=(\mathcal{V},\mathcal{E})$ be a graph on $n=|\mathcal{V}|$ vertices and $m=|\mathcal{E}|$ edges. Let $G\in\mathbb{F}_2^{\;m\times n}$ be the incidence matrix of this graph. The \textit{expansion} (also known as the Cheeger constant or isoperimetric number \cite{mohar1989isoperimetric}) $\beta(\mathcal{G})$ of this graph is the largest real number such that, for all $v\in\mathbb{F}_2^n$ (i.e.~all subsets of vertices),
\begin{equation}
|vG^\top|\ge\beta(\mathcal{G})\min(|v|,n-|v|).
\end{equation}
\end{definition}
More generally, we also propose a novel \textit{relative} version of edge expansion. 

\begin{definition}[Relative Expansion]\label{def:rel_expansion} The expansion \textit{relative} to a vertex subset $\mathcal{U}\subseteq\mathcal{V}$ and parameterized by integer $t>0$, denoted $\beta_t(\mathcal{G},\mathcal{U})$, is the largest real number such that, for all $v\in\mathbb{F}_2^n$,
\begin{equation}
|vG^\top|\ge\beta_t(\mathcal{G},\mathcal{U})\min(t,|u|,|\mathcal{U}|-|u|),
\end{equation}
where we use $u$ to denote the restriction of $v$ to vertices $\mathcal{U}$. Note that $\beta(\mathcal{G})=\beta_{|\mathcal{V}|}(\mathcal{G},\mathcal{V})$. We sometimes refer to $\beta(\mathcal{G})$ as the global expansion to distinguish it from the relative expansion. 
\end{definition}

Any graph always has relative expansion at least as large as its global expansion, i.e.~$\beta_t(\mathcal{G},\mathcal{U})\ge\beta(\mathcal{G})$ for all $\mathcal{U}\subseteq\mathcal{V}$ and $t>0$. This is important for our purposes because we will develop some techniques to guarantee a graph has sufficient relative expansion (to ensure a good quantum code distance) but not necessarily large global expansion. The relation between expansions is a corollary of a simple lemma.
\begin{lemma}\label{lem:relative_expansion_relation}
Suppose $\mathcal{G}=(\mathcal{V},\mathcal{E})$ is a graph and $\mathcal{U},\mathcal{U}'\subseteq\mathcal{V}$ are subsets of vertices. If $\mathcal{U}'\subseteq\mathcal{U}$ and $0< t'\le t$, then $\beta_{t'}(\mathcal{G},\mathcal{U}')\ge\beta_t(\mathcal{G},\mathcal{U})$.
\end{lemma}
\begin{proof}
Consider any $u\in\mathbb{F}_2^{|\mathcal{U}|}$, and let $u'$ be the restriction of $u$ to $\mathcal{U}'\subseteq\mathcal{U}$. It follows $|u'|\le|u|$, $|\mathcal{U}'|-|u'|\le|\mathcal{U}|-|u|$, and $\min(t,|u'|,|\mathcal{U}'|-|u'|)\le\min(t,|u|,|\mathcal{U}|-|u|)$. Also, because $t'\le t$, 
\begin{align}
\min(t',|u'|,|\mathcal{U}'|-|u'|) &\le \min(t,|u'|,|\mathcal{U}'|-|u'|) \nonumber \\
&\le \min(t,|u|,|\mathcal{U}|-|u|).
\end{align}
The result now follows by applying definition \ref{def:rel_expansion}.
\end{proof}

The relative expansion can diverge significantly from the global expansion. For instance, we can show (Appendix~\ref{app:relative_exp_lemma}, Lemma~\ref{lem:expansion_lemma}) that the relative expansion of a graph can be increased to at least $1$ by taking a Cartesian graph product with a path graph, while the global expansion cannot be increased this way. 

This Cartesian product, or \emph{thickening} will be needed later, so we formally define it here.
\begin{definition}[Thickening]\label{def:thickening}
Suppose $\mathcal{G}_0=(\mathcal{V}_0,\mathcal{E}_0)$ and $\mathcal{G}_1=(\mathcal{V}_1,\mathcal{E}_1)$ are two graphs. The Cartesian product \cite{sabidussi1959graph,vizing1963cartesian} is another graph $\mathcal{G}=\mathcal{G}_0\square\mathcal{G}_1=(\mathcal{V}_0\times\mathcal{V}_1,\mathcal{E})$ where $((u_0,u_1),(v_0,v_1))\in\mathcal{E}$ if and only if either (1) $u_0=v_0$ and $(u_1,v_1)\in\mathcal{E}_1$, or (2) $u_1=v_1$ and $(u_0,v_0)\in\mathcal{E}_0$. 

Let $\mathcal{P}_L$ be the path graph with $L$ vertices. We say the graph $\mathcal{G}_0^{(L)}=\mathcal{G}_0\square\mathcal{P}_L$ is $\mathcal{G}_0$ \textit{thickened} $L$ times.
\end{definition}
\noindent Intuitively, the thickened graph $\mathcal{G}_0^{(L)}$ is $L$ copies of $\mathcal{G}_0$ stacked on top of one another and connected ``transversally", i.e.~each vertex is connected to its copy above and below in the stack (or connected to just one copy of itself if it is at the bottom or top of the stack). One can skip ahead to Fig.~\ref{fig:graph_thickening}b to see an example.

\subsection{Auxiliary graph LDPC surgery}
\label{sec-gauging-meas}

Auxiliary graph surgery \cite{williamson2024gauging} is a flexible recipe for measuring a logical operator of a quantum code. Because several of our results are best viewed in the context of auxiliary graph surgery, we give a review of it here. A novelty of our presentation is to weaken the expansion requirement to a subset of vertices in the auxiliary graph used to construct the deformed code, which was only implicit in previous works \cite{williamson2024gauging,cross2024linear}.

Suppose $\mathcal{L}$ is the set of qubits supporting the logical operator $\overline{Z}$ to be measured. Here we assume $\overline{Z}$ is a $Z$-type Pauli operator, which is without loss of generality if we choose the appropriate local basis for each qubit of $\mathcal{L}$ and allow the code to be non-CSS. 

Introduce an ``auxiliary" graph $\mathcal{G}=(\mathcal{V},\mathcal{E})$ and an injective map $f:\mathcal{L}\rightarrow\mathcal{V}$. The function indicates a subset of vertices $\mathrm{im}f=f(\mathcal{L})$, the ``port", at which we attach the original code and the graph. We use $\mathcal{G}$ and $f$ to define a deformed code. 

\begin{itemize}
    \item Each edge $e \in \mathcal{E}$ of this graph is associated to a single additional qubit. $X(e)$ and $Z(e)$ denote physical Pauli operators on this ``edge qubit".

    \item Each vertex $v$ is associated to a $Z$ check, $A_v$. Each ``vertex check" is supported on edge qubits incident to vertex $v$, and in cases also a single qubit $q$ of $\mathcal{L}$ as
\begin{equation}\label{eq:vertex_checks}
A_v=\bigg\{
\begin{array}{ll}
Z(q)\prod_{e\ni v}Z(e),& \exists q\in\mathcal{L}, f(q)=v\\
\prod_{e\ni v}Z(e),& \mathrm{otherwise}
\end{array}
\end{equation}

\item and each cycle $c\subseteq\mathcal{E}$ (in the cycle basis of $\mathcal{G}$) is associated to an $X$ check, $B_c$. Each ``cycle check" is supported on qubits on edges in the cycle basis in the graph as
\begin{equation}
B_c=\prod_{e\in c}X(e).
\end{equation}

\end{itemize}

If we initialize all edge qubits in $\ket{+}$ and measure all checks $A_v$, $B_c$, certain checks of the original code must pick up $X$-type support on the edge qubits to commute with all $A_v$. These deformed checks are exactly the checks of the original code with $X$-type support (either acting as Pauli $X$ or $Y$) on qubits in $\mathcal{L}$. Suppose one such check $s$ from the original code has $X$-type support on qubits $\mathcal{L}_{s}\subseteq\mathcal{L}$. Notice this $X$-type overlap $|\mathcal{L}_{s}|$ is even because check $s$ commutes with logical operator $\overline{Z}$. Since $f$ is injective, $|f(\mathcal{L}_{s})| =|\mathcal{L}_{s}|$ is also even. So far, the $X$-type support of original check $s$ and new vertex check $A_v$ overlap on exactly one qubit $q \in \mathcal{L}$ (where $f(q)=v$), and hence anti-commute. This means the original code check $s$ would anti-commute with an \textit{even} number of new vertex checks, exactly $|f(\mathcal{L}_{s})|$. We can pair up all vertices $f(\mathcal{L}_{s})$ and add additional qubit support to each original check $s$ exactly specified by a path of edges in  $\mathcal{G}$ whose endpoints are a pair of vertices from $f(\mathcal{L}_{s})$. Such a path always exists because $\mathcal{G}$ is connected. This set of paths between paired vertices is also known as a \textit{perfect matching} $\mu(\mathcal{L}_{s})\subseteq\mathcal{E}$ in $\mathcal{G}$ of vertices $f(\mathcal{L}_{s})$. Since each path touches its endpoint vertices once, it adds exactly one edge qubit for each $s$ to the vertex checks and fixes the anti-commutation between $A_v$ and $s$. Therefore after measurement of $A_v$, the original code check $s$ deforms to gain additional support on edge qubits as
\begin{equation}\label{eq:deformed_checks}
s\rightarrow s\prod_{e\in\mu(\mathcal{L}_{s})}X(e),
\end{equation}
where $\mu(\mathcal{L}_{s})\subseteq\mathcal{E}$ is a perfect matching in $\mathcal{G}$ of vertices $f(\mathcal{L}_{s})$. All these components of auxiliary graph surgery are depicted as a Tanner graph in Fig.~\ref{fig:gauging_measurement}.

\begin{figure}[h]
    \centering
    \begin{tikzpicture}[]
\pgfmathsetmacro{\ewid}{1pt}
\pgfmathsetmacro{\drwid}{0.8pt}
\pgfmathsetmacro{\sqrsz}{0.7cm}
\pgfmathsetmacro{\sqrdf}{0.12}
\pgfmathsetmacro{\crcsz}{0.4}
\pgfmathsetmacro{\crcdf}{0.1}
\pgfmathsetmacro{\basecodewid}{2.3}
\pgfmathsetmacro{\brnd}{0.028cm}
\definecolor{edgcol}{rgb}{0.5, 0.5, 0.5}
\definecolor{bl}{rgb}{0.63, 0.79, 0.95}

 \fill[bl, opacity=0.3]
        (-3.1, -0.1) arc[start angle=180, end angle=270, x radius=\brnd, y radius=\brnd] -- ++(\basecodewid, 0) arc[start angle=270, end angle=360, x radius=\brnd, y radius=\brnd] -- ++(0, 2.6) arc[start angle=0, end angle=90, x radius=\brnd, y radius=\brnd] -- ++(-\basecodewid, 0) arc[start angle=90, end angle=180, x radius=\brnd, y radius=\brnd]-- cycle;

        \draw[draw=edgcol,line width=\ewid] (-2.5,0) to (2.7,0);
        \draw[draw=edgcol,line width=\ewid] (-2.5,2) to (2.7,2);
        \draw[draw=edgcol,line width=\ewid] (0,0) to (0,2);
        
        \begin{scope}[shift={(-2.2,0)}]
            \draw[draw=edgcol,line width=\ewid] (-0,0) to (0,2);
           \draw[fill=white,line width=\drwid] (0-\crcdf,2+\crcdf) circle (\crcsz);
            \draw[fill=white,line width=\drwid] (0,2) circle (\crcsz) node {};
            \node[draw, line width=\drwid, fill=white, minimum size=\sqrsz] at (0-\sqrdf, 0+\sqrdf) {};
            \node[draw, line width=\drwid, fill=white,minimum size=\sqrsz] at (0, 0) {}; 
            \node[text=black!90] at (1, 2.9) {{\scriptsize original code}};
        \end{scope}
        
        \draw[fill=white,line width=\drwid] (0-\crcdf,0+\crcdf) circle (\crcsz);
        \draw[fill=white,line width=\drwid] (0,0) circle (\crcsz) node {$\bar{Z}$};
        \node[draw, line width=\drwid, fill=white, minimum size=\sqrsz] at (0-\sqrdf, 2+\sqrdf) {};
        \node[draw, line width=\drwid, fill=white,minimum size=\sqrsz] at (0, 2) {};

        \node[] at (-0.65, 1) {{\footnotesize $[S_X | S_Z]$}};
        
        \node[] at (1.2, 0.2) {{\footnotesize $F$}};
        \node[] at (1.35, 2.2) {{\footnotesize $[M|0]$}};

        \node[] at (0.35, -0.4) {{ $\mathcal{L}$}};
        \node[] at (0.6, 1.65) {{$\mathcal{S}$}};

        \begin{scope}[shift={(2.4,0)}]
        \draw[draw=edgcol,line width=\ewid] (0,0) to (0,4);
        \draw[fill=white,line width=\drwid] (0-\crcdf, 2+\crcdf) circle (\crcsz);
        \draw[fill=white,line width=\drwid] (0,2) circle (\crcsz) node {};
         \node[draw, line width=\drwid,fill=white, minimum size=\sqrsz] at (0-\sqrdf,0+\sqrdf) {};
        \node[draw, line width=\drwid,fill=white,minimum size=\sqrsz] at (0, 0) {$Z$};
        \node[draw, line width=\drwid,fill=white, minimum size=\sqrsz] at (0-\sqrdf,4+\sqrdf) {};
        \node[draw, line width=\drwid,fill=white,minimum size=\sqrsz] at (0, 4) {$X$};
        \node[] at (-0.25, 1) {{\footnotesize $G^{\top}$}};
        \node[] at (-0.25, 3) {{\footnotesize $N$}};
        \node[] at (0.5, 1.7) {{ $\mathcal{E}$}};
        \node[] at (0.55, -0.4) {{ $\mathcal{V}$}};
        \node[] at (0.55, 3.6) {{ $\mathcal{U}$}};
        \end{scope}

\end{tikzpicture}
    \caption{The deformed code created during auxiliary graph surgery of logical operator $\overline{Z}$ supported on qubits $\mathcal{L}$. Edge qubits $\mathcal{E}$ are introduced and vertex checks $\mathcal{V}=\{A_v\}_v$ and cycle checks $\mathcal{U}=\{B_c\}_c$ are measured. Here $G$ is the incidence matrix of the graph $\mathcal{G}$ defining these checks, $N$ is a cycle basis satisfying $NG=0$, and $M$ encodes the support gained by some of the original stabilizers $\mathcal{S}$, specifically, those that have $X$-type support on $\mathcal{L}$, see Eq.~\eqref{eq:deformed_checks}. The matrix $F$ has elements $F_{qv}$ for all $q\in\mathcal{L}$ and $v\in\mathcal{V}$ and $F_{qv}=1$ if and only if $f(q)=v$. The rest of the qubits and checks of the original code are drawn on the left but their Tanner graph connectivity does not change during the code deformation.}
    \label{fig:gauging_measurement}
\end{figure}

There are additional properties of the auxiliary graph $\mathcal{G}$ and the port function $f:\mathcal{L}\rightarrow\mathcal{V}$ that ensure that auxiliary graph surgery results in a suitable deformed code, as formalized in the next theorem. 
\begin{theorem}[Graph Desiderata]\cite{williamson2024gauging}\label{thm:graph_desiderata}

\noindent To ensure the deformed code has exactly one less logical qubit than the original code and measures the target logical operator $\overline{Z}=Z(\mathcal{L})$, it is sufficient that
\begin{enumerate}
\setcounter{enumi}{-1}
\item $\mathcal{G}$ is connected.
\end{enumerate}

\noindent To ensure the deformed code is LDPC, it is necessary and sufficient that
\begin{enumerate}
\setcounter{enumi}{0}
\item $\mathcal{G}$ has $O(1)$ vertex degree.
\item For all stabilizers $s$ of the original code, (a) $|\mu(\mathcal{L}_{s})|=O(1)$, i.e.~each stabilizer has a short perfect matching in $\mathcal{G}$, and (b) each edge is in $O(1)$ matchings $\mu(\mathcal{L}_{s})$.
\item There is a cycle basis of $\mathcal{G}$ in which (a) each cycle is length $O(1)$ and (b) each edge is in $O(1)$ cycles.
\end{enumerate}
To ensure the deformed code has code distance at least the distance $d$ of the original code, it is sufficient that
\begin{enumerate}
\setcounter{enumi}{3}
\item $\mathcal{G}$ has sufficient expansion relative to the image of port function $f$, namely, $\beta_d(\mathcal{G},\mathrm{im}f)\ge1$.
\end{enumerate}
\end{theorem}
\begin{proof}

That desiderata $1,2,3$ are necessary and sufficient for guaranteeing the deformed code is LDPC is evident from the construction, since these specify that row and column weights for check matrices $G$, $M$, $N$, respectively, which label the Tanner graph edges shown in Fig.~\ref{fig:gauging_measurement}, are constants. 

For completeness, we prove the sufficiency of desiderata 0 and 4 for their respective purposes in Appendix~\ref{app:desiderata_proofs}. It is worth noting that a stronger form of desideratum 4 is assumed to prove the deformed code distance in prior work \cite{williamson2024gauging}, namely that $\beta(\mathcal{G})\ge1$. Lemma~\ref{lem:relative_expansion_relation} implies this is sufficient for a graph to satisfy our desideratum 4 but is not necessary. We note the port function only needs to be injective for the proof of 4 and not the others.

Appendix~\ref{app:set-valued_port_function} generalizes this theorem slightly to include the notion of a set-valued port function, so that qubits can be connected to multiple vertex checks. This idea is useful for the most general version of our joint measurement scheme in Section~\ref{sec:rep_adapter}, but is not essential to understand the main ideas of this paper.
\end{proof}


We remark that it is easy to create a graph satisfying desiderata 0, 1 and 2 provided the original code is LDPC. This can be done by pairing up the vertices in $f(\mathcal{L}_{s})$ for all stabilizers $s$ and drawing edges between paired vertices. Desideratum 2 is satisfied by construction. The resulting graph must be constant degree because only $O(1)$ sets $\mathcal{L}_{s}$ can contain any given qubit for an LDPC code, and so we have desideratum 1. If the graph is not connected, add additional edges to connect it while preserving constant vertex degree, and so satisfy desideratum 0. So far, this recipe leaves only desiderata 3 and 4 unsatisfied.

We can satisfy desideratum 4 by adding more edges to increase the relative expansion. One approach is to construct a constant degree graph with constant global expansion (using e.g.~\cite{reingold2000entropy,alon2008elementary}) and add its edges to the existing graph. This constant expansion, if initially insufficient, can be boosted to relative expansion at least $1$ using thickening, as we explain in Appendix~\ref{app:relative_exp_lemma}. Alternatively, in a perhaps more practical approach, one can iteratively add edges as in Ref.~\cite{ide2024faulttolerant} until sufficient global or relative expansion is achieved. Either method does not increase the number of edges or degree by more than constant factors.

Thus, we now have a graph $\mathcal{G}_0=(\mathcal{V}_0,\mathcal{E}_0)$ satisfying desiderata 0,1,2, and 4. The next steps will create from $\mathcal{G}_0$ a new graph that satisfies all the desiderata while having only a factor of $O(\log^3|\mathcal{V}_0|)$ more vertices and edges than $\mathcal{G}_0$. This process, illustrated in Fig.~\ref{fig:graph_thickening}, is made possible by the decongestion lemma from Freedman and Hastings. Here we quote the relevant parts of that lemma.
\begin{lemma}[Decongestion Lemma]\cite{freedman2021buildingmanifoldsquantumcodes}\label{lem:decongestion_lemma}
If $\mathcal{G}=(\mathcal{V},\mathcal{E})$ is a graph with vertex degree $O(1)$, then there exists a cycle basis in which each edge appears in at most $O(\log^2|\mathcal{V}|)$ cycles of the basis. Moreover, the basis cycles can be ordered so that each cycle intersects at most polylogarithmically many cycles later in the basis. This basis can be constructed by an efficient randomized algorithm.
\end{lemma}
\noindent It was pointed out to us \cite{sunny} after the first version of this paper appeared that 
by inspecting of the proof of Freedman and Hastings, one can deduce that each cycle actually overlaps $O(\log^3|\mathcal{V}|)$ cycles later in the basis. As a corollary \cite{sunny}, for any graph $\mathcal{G}_0=(\mathcal{V}_0,\mathcal{E}_0)$, the $L=O(\log^3|\mathcal{V}_0|)$ times thickened graph $\mathcal{G}_0^{(L)}$ (see Definition~\ref{def:thickening}) has a cycle basis in which each edge appears in at most $O(1)$ cycles, satisfying desideratum 3b. This corrects an erroneous claim in the previous version that a $O(\log^2|\mathcal{V}|)$ times thickened graph was sufficient. The construction of this thickened graph and its sparse cycle basis is depicted in Fig.~\ref{fig:graph_thickening}c. 

Because $\mathcal{G}_0$ satisfies desideratum 4 with respect to some port set of vertices $f(\mathcal{L})\subseteq\mathcal{V}_0$, also $\mathcal{G}_0^{(L)}$ satisfies it with respect to any copy of that port in the thickened graph, i.e.~$f(\mathcal{L})\times\{l\}$ for any $l=0,1,\dots,L-1$. See Appendix~\ref{app:relative_exp_lemma} for the proof. 

\begin{figure*}
\includegraphics[width=\textwidth]{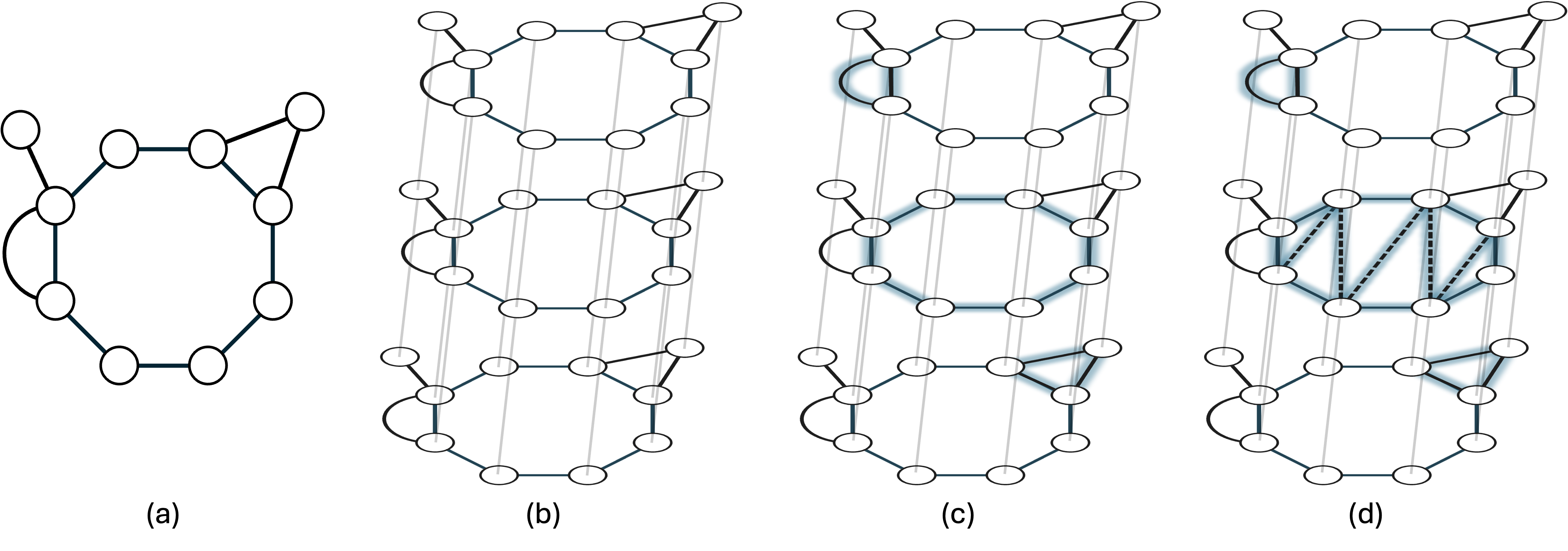}
\caption{An example of the steps going into constructing a graph $\mathcal{G}$ satisfying all the desiderata of Theorem~\ref{thm:graph_desiderata} starting from a graph $\mathcal{G}_0$ satisfying just desiderata 0, 1, and 2. (a) The graph $\mathcal{G}_0$. (b) The thickened graph $\mathcal{G}^{(L)}_0$, here with $L=3$. Note the edges (drawn lighter gray) connecting the corresponding vertices in adjacent layers. (c) A cycle basis for $\mathcal{G}^{(L)}_0$ includes every length four cycle constructed from edge $e=(i,j)$ in layer $l$, the copy of edge $e$ in adjacent layer $l+1$, and the lighter gray edges connecting the copies of $i$ and $j$. The cycle basis of $\mathcal{G}^{(L)}_0$ also includes (highlighted) one cycle for each cycle in a basis of $\mathcal{G}_0$, and each such cycle can be put into any one of the layers independently. A highlighted cycle is equivalent to its copies in other layers by adding to it the length four cycles between layers. (d) Long cycles in the basis can now be cellulated by adding edges (dashed) and including the resulting triangles in the cycle basis instead.}
\label{fig:graph_thickening}
\end{figure*}

The final step is to satisfy desiderata 3a. This necessitates that $\mathcal{G}^{(L)}_0$ not have any long cycles in its cycle basis. However, this can be done by adding enough edges to cellulate each long cycle in the basis into constant length cycles, as depicted in Fig.~\ref{fig:graph_thickening}d. The constant length cycles can be made very short, even length three if desired. Because the graph $\mathcal{G}^{(L)}_0$ already has a basis satisfying 3b, this cellulation, which only modifies cycles in the basis, cannot increase the degree of any vertex by more than $O(1)$. Moreover, cellulation can only increase relative expansion. Thus, the cellulated version of $\mathcal{G}^{(L)}_0$ is our final graph $\mathcal{G}$ that, together with the port set $f(\mathcal{L})\times\{l\}$ described above, satisfies all desiderata in Theorem~\ref{thm:graph_desiderata}.


Finally, we briefly summarize the entire auxiliary graph surgery protocol in spacetime from Ref.~\cite{williamson2024gauging}. There are four steps (1) initialize all edge qubits $\mathcal{E}$ in the $\ket{+}$ state, (2) repeat the measurement of the deformed code's checks at least $d$ times, (3) measure out all the edge qubits in the $X$ basis, and (4) apply a Pauli correction to the original code qubits to return to the original codespace. To describe the correction in the last step, fix an arbitrary vertex $v_0 \in \mathcal{V}$. For any other vertex $v$, let $\mu_v$ be an arbitrary path from $v_0$ to $v$ in graph $\mathcal{G}$, and let $m_e$ be the $\pm1$ measurement result from measuring the edge qubit $e$. If $\prod_{e \in \mu_v} m_e=-1$ apply a correction of $Z_q$ for the qubit $q=f^{-1}(v)$ connected to the check on $v$. 

We note that auxiliary graph surgery implemented this way has phenomenological fault distance \cite{beverland2024fault} equal to the code distance $d$ of the original code. The phenomenological fault distance is the minimum number of qubit or measurement errors that can cause an undetected logical error (note, other noise in the circuits for measuring checks is not included). We refer to \cite{williamson2024gauging, cross2024linear}. for the proof.

\section{The $\mathsf{SkipTree}$ basis transformation} \label{sec:skiptree}

Consider a \emph{classical} code defined on a connected graph $\mathcal{G}=(\mathcal{V},\mathcal{E})$ with vertices representing $n=|\mathcal{V}|$ bits and edges representing $m=|\mathcal{E}|$ parity checks. It is clear the code is equivalent to a repetition code of length $n$ with an unconventional basis of parity checks. Indeed, the incidence matrix $G\in\mathbb{F}_2^{m\times n}$ of the graph is the parity check matrix of the code.

In this section, we ask if we can instead always use a canonical basis of repetition code checks $H_C$ such that each of these canonical checks is a product of a constant number (independent of $n$) of the old checks of $G$. We allow bits in the canonical basis to have different indices than they had in the old basis. Thus, our question is equivalent to asking whether there is always a sparse transformation matrix $T\in\mathbb{F}_2^{n\times m}$ and permutation $P$ (representing the aforementioned bit relabeling) such that $TGP=H_C$.

Our main result is the following.
\begin{theorem}\label{thm:skiptree}
For any connected graph with $n$ vertices, $m$ edges, and $m\times n$ incidence matrix $G$, there exists an $n\times m$ $(3,2)$-sparse matrix $T$ and $n\times n$ permutation matrix $P$ such that $TGP=H_C$. There is also an algorithm to find $T$ and $P$ that takes $O(n+m)$ time (returning $T$ and $P$ as sparse matrices).
\end{theorem}
\begin{proof}
We analyze Algorithm~\ref{alg:skiptree} that returns $T,P$ given $G$. We assume $n>1$, $m>0$ in the following since otherwise the theorem is trivial. Notably, the number of edges need not be bounded and our algorithm works for high-degree graphs.

The algorithm works by first finding a spanning tree $S$ of the graph $G$. 
We may choose an arbitrary node $r$ to be the root of $S$, which also uniquely defines parent/child relationships for the whole tree. Finding a spanning tree can be done in $O(n+m)$ time \cite{cormen2022introduction} and can be stored in a data structure in which it is constant time to find the children or parent of a given node, e.g.~by having each node store the indices of its neighbors in the tree.

Now we proceed recursively with two functions $\mathsf{LabelFirst}$ and $\mathsf{LabelLast}$, both of which take a node in the tree as an argument. In words, $\mathsf{LabelFirst}(i)$ will first label node $i$ with the next unused integer from $\mathbb{Z}_n=\{0,1,\dots,n-1\}$ (the next unused integer is tracked globally), and then call $\mathsf{LabelLast}(j)$ on each child $j$ of $i$. Similarly, $\mathsf{LabelLast}(i)$ will call $\mathsf{LabelFirst}(j)$ on each child $j$ of $i$ and, only after all those function calls have finished, will label node $i$ with the next unused integer. Each node is either the argument of a $\mathsf{LabelFirst}$ call or a $\mathsf{LabelLast}$ call (not both), and so is labeled exactly once. Moreover, the ``first" and ``last" type nodes make a two-coloring of the vertices of the spanning tree (though not necessarily a two-coloring of graph $G$) in the sense that no two first-type nodes are adjacent and no two last-type nodes are adjacent in the spanning tree. First-type nodes are labeled before every other node in their sub-tree, while last-type nodes are labeled after every other node in their sub-tree.

This recursion structure means that very often (though not always), nodes labeled $i$ and $i+1\inlinemod{n}$ are not adjacent in the spanning tree but instead have one or two nodes separating them, thus explaining $\mathsf{SkipTree}$ as the algorithm's name. See Fig.~\ref{fig:skiptree_example} for an example. This is the desired behavior, because (as we show later) we want to ensure that for all nodes $i$, the shortest path from $i$ to $i+1\inlinemod{n}$ in the tree is constant length. Intuitively, this involves traveling up and down the tree, labeling nodes as we go, but when traveling away from the root we should leave some nodes unlabeled so they can be labeled later on the way back to the root.

\begin{figure*}[t]
    \centering
    \includegraphics[width=0.75\linewidth]{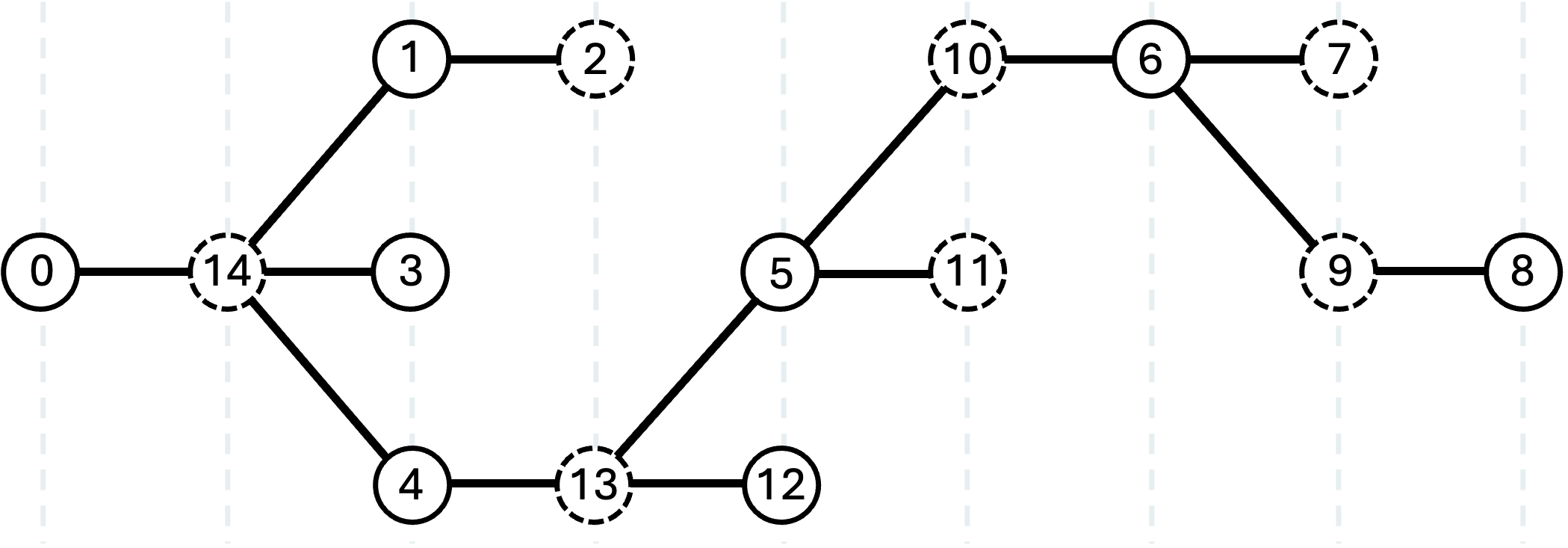}
    \caption{An example spanning tree with nodes labeled according to the $\mathsf{SkipTree}$ algorithm, Algorithm~\ref{alg:skiptree}. The root node is labeled $0$ on the far left, and first-type nodes are represented with solid circles, while last-type nodes are dashed. Path $i$ is the unique shortest path in the spanning tree from node labeled $i$ to the node labeled $i+1\inlinemod{15}$. The proof of Theorem~\ref{thm:skiptree} involves arguing that all such paths are length three or less. Examples of all cases encountered in that proof are included in this figure -- path 12 for case (1), path 3 for case (2), path 0 for case (3), path 6 for case (4), path 14 for case (a), path 13 for case (b), path 11 for case (c), path 10 for case (d), and path 7 for case (e).}
    \label{fig:skiptree_example}
\end{figure*}

Although it can be arbitrary, there is an order in which functions $\mathsf{LabelFirst}$ and $\mathsf{LabelLast}$ are called on the children of a node $i$. If $j$ is the first child on which the function is called, we call $j$ the oldest child of $i$. If $j$ is the last on which the function is called, we call $j$ the youngest child of $i$. We introduce the ``oldest" and ``youngest" terminology to avoid using the words ``first" and ``last" again. We can also use terms like ``next oldest" and ``next youngest" to refer to this order of children.

We construct matrix $T$ so that its row $i$ indicates the edges of $G$ that are in the unique shortest path between the node labeled $i$ and the node labeled $i+1\inlinemod{n}$ in the spanning tree. Note that only edges in the spanning tree are used in these paths. This is perhaps less than optimal for some graphs but is sufficient to achieve a $T$ of constant sparsity without worrying about the graph structure beyond the spanning tree. We refer to the shortest path between $i$ and $i+1\inlinemod{n}$ as path $i$.

Each edge $e=(f,l)$ of the spanning tree is adjacent to a first-type node $f$ and a last-type node $l$, and the label of the last-type node must be larger than the label of the first-type node. If $f$ is the parent of $l$, then path $f$ and path $l$ include the edge $e$. If instead $l$ is the parent of $f$, then path $f-1\inlinemod{n}$ and path $l-1\inlinemod{n}$ include the edge $e$. In either case, once the edge is used the second time, the entire sub-tree below and including the edge has been labeled, so the edge is not included in any other paths. This shows the column weight of $T$ is exactly two (for all columns corresponding to edges in the spanning tree).

The weight of the $i^\text{th}$ row of $T$ is the number of edges in path $i$. We claim the longest a path can be is length three. We show this to complete the proof. An example including all cases encountered in this part of the proof is provided in Fig.~\ref{fig:skiptree_example}.

Consider path $i$ where $i$ is a first-type node. Note that the node labeled $n-1$ is always last-type and so $i\neq n-1$. We first assume node $i$ has no children, which also implies $i$ is not the root and it has a parent. There are two sub-cases -- (1) if node $i$ is the youngest child of its parent, then its parent will be labeled $i+1$, in which case the path length is one, (2) if node $i$ is not the youngest child of its parent, then its next youngest sibling is labeled $i+1$, in which case the path length is two. Next, assume node $i$ has at least one child. There are again two sub-cases -- (3) if the oldest child of $i$ has a child, then node $i+1$ will be a grandchild of $i$, in which case the path length is two, (4) if the oldest child of $i$ does not have a child, then node $i+1$ will be the oldest child of $i$, in which case the path length is one.

Now consider path $i$ where $i$ is a last-type node. Because node $i$ has been labeled, its entire sub-tree has already been labeled with integers $<i$. Also, node $i$ always has a parent because the only parentless node, the root, is first-type.
We first assume node $i$ is the youngest child of its parent $p$. There are three sub-cases -- (a) Node $p$ has no parent and so is the root, which implies $i=n-1$, $i+1\inlinemod{n}=0$, and the path length is one, (b) Node $p$ is the youngest child of its parent, which implies $i+1$ is the grandparent of $i$ and the path length is two, (c) Node $p$ is not the youngest child of its parent, which implies the next youngest child of the parent of $p$ is $i+1$ (i.e.~an ``uncle" of node $i$) and the path length is three. Now, assume node $i$ is not the youngest child of its parent $p$. There is a next youngest child $y$ and it is a last-type node. There are two sub-cases -- (d) $y$ has no children, which implies $y$ is labeled $i+1$ and a path length of two, (e) $y$ has a child, which implies its oldest child is labeled $i+1$ (i.e.~a ``nephew" of $i$) and the path length is three.
\end{proof}

\begin{algorithm}[H]
\caption{Given a connected graph $G\in\mathbb{F}_2^{\;m\times n}$, find $T\in\mathbb{F}_2^{\;n\times m}$ such that $TGP=H_C$. Both $T$ and $P$ have $O(n)$ nonzero entries and can be constructed and returned as sparse matrices.}
\label{alg:skiptree}
\begin{algorithmic}[1]
    \State\protect\Comment{Note we say ``vertex $v$" if it corresponds to the $v^\text{th}$ column of $G$. One goal of this algorithm is to label vertices uniquely with integers from $\mathbb{Z}_n$ in such a way to guarantee $T$ is sparse. We say ``vertex $v$ is labeled $l$" if it acquires label $l\in\mathbb{Z}_n$. Of course, $v$ need not equal $l$.}
    \Procedure{$\mathsf{SkipTree}$}{$G$}
        \State $S\leftarrow$ a spanning tree of $G$ \protect\Comment{has incidence matrix $S_I\in\mathbb{F}_2^{\;n-1\times n}$ that we do not need to store}
        \State Index $\leftarrow0$ 
        \State $\mathrm{Label}\leftarrow$ empty list of length $n$
        \Procedure{$\mathsf{LabelFirst}$}{$v$}
            \State $\mathrm{Label}[\mathrm{Index}]\leftarrow v$
            \State Index $\leftarrow$ Index + 1
            \For{each child of vertex $v$ in $S$} 
                \State $\mathsf{LabelLast}(\mathrm{child})$
            \EndFor
        \EndProcedure
        
        \Procedure{$\mathsf{LabelLast}$}{$v$}
            \For{each child of vertex $v$ in $S$} 
                \State $\mathsf{LabelFirst}(\mathrm{child})$
            \EndFor
            \State $\mathrm{Label}[\mathrm{Index}]\leftarrow v$
            \State Index $\leftarrow$ Index + 1
        \EndProcedure

        \State $\mathsf{LabelFirst}(0)$ \protect\Comment{we choose the root to be 0. After this line, $\mathrm{Label}[l]=v$ means vertex $v$ is labeled $l$.}
        \State $P \leftarrow n\times n$ matrix with $P_{vl}=1$ iff $\mathrm{Label}[l]=v$.
        \State $\tilde T\leftarrow$ matrix with $n$ rows and $n-1$ columns.
        \State $\tilde T_{le}=1$ iff edge $e$ is part of the shortest path in $S$ from $\mathrm{Label}[l]$ to $\mathrm{Label}[l+1]$. \protect\Comment{now $\tilde TS_I=H_CP^\top$}
        \State Add zero columns to $\tilde T$, obtaining $T$ so that $TG=\tilde TS_I$.
        \State Return $T$, $P$.

    \EndProcedure
\end{algorithmic}
\end{algorithm}

We remark on modifications to the $\mathsf{SkipTree}$ algorithm and special cases. First, we observe that the $\mathsf{SkipTree}$ algorithm can also be used to find $T',P'$ such that $T'GP'=H_R$ where $H_R\in\mathbb{F}_2^{(n-1)\times n}$ is the canonical \emph{full-rank} parity check matrix of the repetition code. That is, $H_R$ is $H_C$ with the last row removed. Thus, it is clear the $\mathsf{SkipTree}$ algorithm also solves this problem by just removing the last row of $T$ to get $T'$ and setting $P'=P$. However, we also present a slightly modified version of the $\mathsf{SkipTree}$ algorithm in Appendix~\ref{app:skip_fullrank} to improve the sparsity of $T'$ for some graphs.

If there exists a Hamiltonian cycle in $G$, i.e.~a cycle that visits every vertex and uses every edge at most once, then the $\mathsf{SkipTree}$ algorithm is not necessary to solve $TGP=H_C$. Instead, $T$ can be chosen to be the $(1,1)$-sparse matrix that selects just those edges from $G$ that are in the Hamiltonian cycle. Likewise, $T'GP'=H_R$ can be solved with a  $(1,1)$-sparse matrix $T'$ if there exists a Hamiltonian path in $G$. Unfortunately, these observations are not generally useful as it is well-known that finding a Hamiltonian cycle or path (or even determining the existence of one) are NP-complete problems.

For graph $\mathcal{G}$, the power graph $\mathcal{G}^p$ is defined to have the same vertex set as $\mathcal{G}$ and edges connecting any two vertices that are at most distance $p$ apart in $\mathcal{G}$. The $\mathsf{SkipTree}$ algorithm provides an alternative proof of the fact that, for any connected graph $\mathcal{G}$, the graph $\mathcal{G}^3$ has a Hamiltonian cycle \cite{karaganis1968cube,chartrand1969cube}. Similarly, some properties of $\mathcal{G}$ are known to be sufficient so that $\mathcal{G}^2$ has a Hamiltonian cycle \cite{fleischner1974square,chartrand1974square,radoszewski2011hamiltonian}. This body of literature in graph theory could likely be used to furnish more variants of the $\mathsf{SkipTree}$ basis transformation with, for instance, reduced sparsity of matrix $T$ in certain cases, but we leave rigorous exploration of this direction to future work.

To conclude this section, we provide some intuition for why Theorem~\ref{thm:skiptree} is useful in quantum LDPC code deformations. Consider some CSS (for simplicity, not necessity) quantum LDPC code with check matrices $H_X,H_Z$ and a logical $Z$-type operator $L$ of that code. Let $S_X$ be the sub-matrix of $H_X$ that is supported on $\overline{Z}$. It is known that when $L$ is irreducible (i.e.~there is no other $Z$-type logical or stabilizer supported entirely within the support of $\overline{Z}$), then $S_X$ is a parity check matrix of a repetition code (see Lemma~8 of \cite{cross2024linear} or Appendix~\ref{app:supportlemma}). While $S_X$ is perhaps not the incidence matrix of a graph (because its rows may have weight larger than two), auxiliary graph surgery \cite{williamson2024gauging} decomposes each such higher weight row into a sum of weight two rows, and thus we obtain a graph incidence matrix $G$ that is also a parity check matrix for the repetition code. It is this parity check matrix that can undergo a basis change into canonical form $H_C$. 

For instance, see the Tanner graph in Fig.~\ref{fig:ldpc_deformation_example}, in which we claim (1) all the checks commute, (2) the code is LDPC, and (3) $Z(\mathcal{L})$ and $Z(\mathcal{Q})$ are equivalent logical operators. The idea encapsulated by this example is the same used to connect codes with repetition code and toric code adapters in Sections~\ref{sec:rep_adapter} and \ref{sec:toric_adapter}. For the adapter constructions, we delve deeper into the details of the deformed codes.

\begin{figure}[t]
\centering
\begin{tikzpicture}[]
\pgfmathsetmacro{\ewid}{1pt}
\pgfmathsetmacro{\drwid}{0.8pt}
\pgfmathsetmacro{\sqrsz}{0.7cm}
\pgfmathsetmacro{\sqrdf}{0.12}
\pgfmathsetmacro{\crcsz}{0.4}
\pgfmathsetmacro{\crcdf}{0.1}
\pgfmathsetmacro{\basecodewid}{2.3}
\pgfmathsetmacro{\brnd}{0.028cm}
\definecolor{edgcol}{rgb}{0.5, 0.5, 0.5}
\definecolor{bl}{rgb}{0.63, 0.79, 0.95}

 \fill[bl, opacity=0.3]
        (-3.1, -0.1) arc[start angle=180, end angle=270, x radius=\brnd, y radius=\brnd] -- ++(\basecodewid, 0) arc[start angle=270, end angle=360, x radius=\brnd, y radius=\brnd] -- ++(0, 2.6) arc[start angle=0, end angle=90, x radius=\brnd, y radius=\brnd] -- ++(-\basecodewid, 0) arc[start angle=90, end angle=180, x radius=\brnd, y radius=\brnd]-- cycle;

        \draw[draw=edgcol,line width=\ewid] (-2.5,0) to (4.8,0);
        \draw[draw=edgcol,line width=\ewid] (-2.5,2) to (4.8,2);
        \draw[draw=edgcol,line width=\ewid] (0,0) to (0,2);
        
        \begin{scope}[shift={(-2.2,0)}]
            \draw[draw=edgcol,line width=\ewid] (-0,0) to (0,2);
           \draw[fill=white,line width=\drwid] (0-\crcdf,2+\crcdf) circle (\crcsz);
            \draw[fill=white,line width=\drwid] (0,2) circle (\crcsz) node {};
            \node[draw, line width=\drwid, fill=white, minimum size=\sqrsz] at (0-\sqrdf, 0+\sqrdf) {};
            \node[draw, line width=\drwid, fill=white,minimum size=\sqrsz] at (0, 0) {}; 
            \node[text=black!90] at (1, 2.9) {{\scriptsize original code}};
        \end{scope}
        
        \draw[fill=white,line width=\drwid] (0-\crcdf,0+\crcdf) circle (\crcsz);
        \draw[fill=white,line width=\drwid] (0,0) circle (\crcsz) node {$\bar{Z}$};
        \node[draw, line width=\drwid, fill=white, minimum size=\sqrsz] at (0-\sqrdf, 2+\sqrdf) {};
        \node[draw, line width=\drwid, fill=white,minimum size=\sqrsz] at (0, 2) {};

        \node[] at (-0.65, 1) {{\footnotesize $[S_X | S_Z]$}};
        
        \node[] at (1.2, 0.2) {{\footnotesize $F$}};
        \node[] at (1.35, 2.2) {{\footnotesize $[M|0]$}};

        \node[] at (0.4, -0.4) {{ $\mathcal{L}$}};
        \node[] at (0.6, 1.65) {{$\mathcal{S}$}};

        \begin{scope}[shift={(2.4,0)}]
        \draw[draw=edgcol,line width=\ewid] (0,0) to (0,4);

        \draw[draw=edgcol,dashed,line width=\ewid-0.1] (1,-0.7) to (1,3.4);
        \draw[fill=white,line width=\drwid] (0-\crcdf, 2+\crcdf) circle (\crcsz);
        \draw[fill=white,line width=\drwid] (0,2) circle (\crcsz) node {};
         \node[draw, line width=\drwid,fill=white, minimum size=\sqrsz] at (0-\sqrdf,0+\sqrdf) {};
        \node[draw, line width=\drwid,fill=white,minimum size=\sqrsz] at (0, 0) {$Z$};
        \node[draw, line width=\drwid,fill=white, minimum size=\sqrsz] at (0-\sqrdf,4+\sqrdf) {};
        \node[draw, line width=\drwid,fill=white,minimum size=\sqrsz] at (0, 4) {$X$};
        \node[] at (-0.25, 1) {{\footnotesize $G^{\top}$}};
        \node[] at (1.4, 2.2) {{\footnotesize $T$}};
        \node[] at (1.4, 0.2) {{\footnotesize $P$}};
        \node[] at (-0.25, 3) {{\footnotesize $N$}};
        \node[] at (0.5, 1.7) {{ $\mathcal{E}$}};
        \node[] at (0.55, -0.4) {{ $\mathcal{V}$}};
        \node[] at (0.55, 3.6) {{ $\mathcal{U}$}};
        \end{scope}

        \begin{scope}[shift={(4.8,0)}]
        \draw[draw=edgcol,line width=\ewid] (0,0) to (0,2);
        \node[draw, line width=\drwid,fill=white, minimum size=\sqrsz] at (0-\sqrdf, 2+\sqrdf) {};
        \node[draw, line width=\drwid,fill=white,minimum size=\sqrsz] at (0, 2) {$X$};
        \draw[fill=white,line width=\drwid] (0-\crcdf,0+\crcdf) circle (\crcsz);
        \draw[fill=white,line width=\drwid] (0,0) circle (\crcsz) node {};
        \node[] at (-0.25, 1) {{\footnotesize $H_C$}};
        \node[] at (0.5, 1.7) {{ $\mathcal{C}$}};
        \node[] at (0.45, -0.4) {{ $\mathcal{Q}$}};
        \end{scope}

\end{tikzpicture}
\caption{An LDPC code deformation which makes use of Theorem~\ref{thm:skiptree} to ensure the code pictured is LDPC and its checks commute. Building on top of the construction in Fig.~\ref{fig:gauging_measurement} (left of the dashed line), edges $T$ and $P$ are introduced, which are outputs of the $\mathsf{SkipTree}$ algorithm, that connect to new qubits $\qcal$ and checks $\ccal$ having repetition code structure. }\label{fig:ldpc_deformation_example}
\end{figure}

\section{The repetition code adapter for joint logical measurements} \label{sec:rep_adapter}

Of practical concern in fault-tolerant quantum computing is the connection of bespoke systems designed to measure two operators $\overline{Z}_l$ and $\overline{Z}_r$ into one system to measure the product $\overline{Z}_l\overline{Z}_r$ without measuring either operator individually. The operators $\overline{Z}_l$ and $\overline{Z}_r$ may be in different codeblocks or the same codeblock. While these logicals may overlap, we assume that they do not anti-commute locally on any single qubit so that they can both be made $Z$-type simultaneously via single-qubit Cliffords. If there are multiple codeblocks, we still refer to them together as the original code, which just happens to be separable.

\begin{figure*}[t]
    \centering
    \begin{tikzpicture}[]
\pgfmathsetmacro{\ewid}{1pt}
\pgfmathsetmacro{\drwid}{0.8pt}
\pgfmathsetmacro{\sqrsz}{0.7cm}
\pgfmathsetmacro{\sqrdf}{0.12}
\pgfmathsetmacro{\crcsz}{0.4}
\pgfmathsetmacro{\crcdf}{0.1}
\pgfmathsetmacro{\basecodewid}{1.2}
\pgfmathsetmacro{\brnd}{0.016cm}
\definecolor{edgcol}{rgb}{0.5, 0.5, 0.5}
\definecolor{bl}{rgb}{0.63, 0.79, 0.95}

 \fill[bl, opacity=0.3]
        (-1.3, -0.3) arc[start angle=180, end angle=270, x radius=\brnd, y radius=\brnd] -- ++(\basecodewid, 0) arc[start angle=270, end angle=360, x radius=\brnd, y radius=\brnd] -- ++(0, 2.8) arc[start angle=0, end angle=90, x radius=\brnd, y radius=\brnd] -- ++(-\basecodewid, 0) arc[start angle=90, end angle=180, x radius=\brnd, y radius=\brnd]-- cycle;

        \draw[draw=edgcol,line width=\ewid] (-1,0) to (10.4,0);
        \draw[draw=edgcol,line width=\ewid] (-1,2) to (10.4,2);
        
       \draw[draw=edgcol,line width=\ewid] (0,0) to (0,2);
        \draw[fill=white,line width=\drwid] (0-\crcdf,0+\crcdf) circle (\crcsz);
        \draw[fill=white,line width=\drwid] (0,0) circle (\crcsz) node {$\bar{Z}_l$};
        \node[draw, line width=\drwid, fill=white, minimum size=\sqrsz] at (0-\sqrdf, 2+\sqrdf) {};
        \node[draw, line width=\drwid, fill=white,minimum size=\sqrsz] at (0, 2) {};

        \node[] at (-0.7, 1) {{\footnotesize $[S^X_l | S^Z_l]$}};
        
        \node[] at (1.2, 0.2) {{\footnotesize $F_l$}};
        \node[] at (1.3, 2.2) {{\footnotesize $[M_l|0]$}};

        \node[] at (0.5, -0.4) {{ $\mathcal{L}_l$}};
        \node[] at (0.65, 1.65) {{$\mathcal{S}_l$}};

        \begin{scope}[shift={(2.4,0)}]
        \draw[draw=edgcol,line width=\ewid] (0,0) to (0,4);
        \draw[fill=white,line width=\drwid] (0-\crcdf, 2+\crcdf) circle (\crcsz);
        \draw[fill=white,line width=\drwid] (0,2) circle (\crcsz) node {};
         \node[draw, line width=\drwid,fill=white, minimum size=\sqrsz] at (0-\sqrdf,0+\sqrdf) {};
        \node[draw, line width=\drwid,fill=white,minimum size=\sqrsz] at (0, 0) {$Z$};
        \node[draw, line width=\drwid,fill=white, minimum size=\sqrsz] at (0-\sqrdf,4+\sqrdf) {};
        \node[draw, line width=\drwid,fill=white,minimum size=\sqrsz] at (0, 4) {$X$};
        \node[] at (-0.3, 1) {{\footnotesize $G^{\top}_l$}};
        \node[] at (1.2, 2.2) {{\footnotesize $T_l$}};
        \node[] at (1.2, 0.2) {{\footnotesize $P_l$}};
        \node[] at (-0.25, 3) {{\footnotesize $N_l$}};
        \node[] at (0.5, 1.7) {{ $\mathcal{E}_l$}};
        \node[] at (0.55, -0.4) {{ $\mathcal{V}_l$}};
        \node[] at (0.6, 3.6) {{ $\mathcal{U}_l$}};
        \end{scope}

        \begin{scope}[shift={(4.8,0)}]
        \draw[draw=edgcol,line width=\ewid] (0,0) to (0,2);
        \node[draw, line width=\drwid,fill=white, minimum size=\sqrsz] at (0-\sqrdf, 2+\sqrdf) {};
        \node[draw, line width=\drwid,fill=white,minimum size=\sqrsz] at (0, 2) {$X$};
        \draw[fill=white,line width=\drwid] (0-\crcdf,0+\crcdf) circle (\crcsz);
        \draw[fill=white,line width=\drwid] (0,0) circle (\crcsz) node {};
        \node[] at (-0.3, 1) {{\footnotesize $H_C$}};
        \node[] at (0.55, 1.7) {{ $\mathcal{C}$}};
        \node[] at (0.45, -0.4) {{ $\mathcal{A}$}};
        \end{scope}

         \begin{scope}[shift={(7.2,0)}]
         \draw[draw=edgcol,line width=\ewid] (0,0) to (0,4);
        \draw[fill=white,line width=\drwid] (0-\crcdf, 2+\crcdf) circle (\crcsz);
        \draw[fill=white,line width=\drwid] (0,2) circle (\crcsz) node {};
         \node[draw, line width=\drwid,fill=white, minimum size=\sqrsz] at (0-\sqrdf,0+\sqrdf) {};
        \node[draw, line width=\drwid,fill=white,minimum size=\sqrsz] at (0, 0) {$Z$};
        \node[draw, line width=\drwid,fill=white, minimum size=\sqrsz] at (0-\sqrdf,4+\sqrdf) {};
        \node[draw, line width=\drwid,fill=white,minimum size=\sqrsz] at (0, 4) {$X$};
        \node[] at (-0.3, 1) {{\footnotesize $G^{\top}_r$}};
        \node[] at (-1.2, 2.2) {{\footnotesize $T_r$}};
        \node[] at (-1.2, 0.2) {{\footnotesize $P_r$}};
        \node[] at (-0.25, 3) {{\footnotesize $N_r$}};
        \node[] at (1.2, 2.2) {{\footnotesize $[M_r|0]$}};
        \node[] at (1.2, 0.2) {{\footnotesize $F_r$}};
        \node[] at (0.5, 1.7) {{ $\mathcal{E}_r$}};
        \node[] at (0.55, -0.4) {{ $\mathcal{V}_r$}};
        \node[] at (0.6, 3.6) {{ $\mathcal{U}_r$}};
         \end{scope}

         \begin{scope}[shift={(9.6,0)}]
          \fill[bl, opacity=0.3]
        (-0.7, -0.3) arc[start angle=180, end angle=270, x radius=\brnd, y radius=\brnd] -- ++(\basecodewid, 0) arc[start angle=270, end angle=360, x radius=\brnd, y radius=\brnd] -- ++(0, 2.8) arc[start angle=0, end angle=90, x radius=\brnd, y radius=\brnd] -- ++(-\basecodewid, 0) arc[start angle=90, end angle=180, x radius=\brnd, y radius=\brnd]-- cycle;
           \draw[draw=edgcol,line width=\ewid] (0,0) to (0,2);
         
            \draw[fill=white,line width=\drwid] (0-\crcdf,0+\crcdf) circle (\crcsz);
        \draw[fill=white,line width=\drwid] (0,0) circle (\crcsz) node {$\bar{Z}_r$};
        \node[draw, line width=\drwid, fill=white, minimum size=\sqrsz] at (0-\sqrdf, 2+\sqrdf) {};
        \node[draw, line width=\drwid, fill=white,minimum size=\sqrsz] at (0, 2) {};

        \node[] at (0.7, 1) {{\footnotesize $[S^X_r | S^Z_r]$}};

        \node[] at (0.5, -0.4) {{ $\mathcal{L}_r$}};
        \node[] at (0.65, 1.65) {{$\mathcal{S}_r$}};
        \end{scope}
         
\end{tikzpicture}
    \caption{Measuring $\overline{Z}_l\overline{Z}_r$ without measuring either logical operator individually. The adapter edges joining auxiliary graphs $G_l$ and $G_r$ into one larger auxiliary graph are those hosting the qubits $\mathcal{A}$. The matrices $T_l,T_r,P_l,P_r$ are determined using the $\mathsf{SkipTree}$ algorithm, Theorem~\ref{thm:skiptree}. Differing from Figs.~\ref{fig:gauging_measurement} and \ref{fig:ldpc_deformation_example}, we do not include in this diagram qubits or checks outside the supports of the logicals $\overline{Z}_l$ and $\overline{Z}_r$, and instead just suggest their existence through dangling edges. This is left ambiguous because we allow $\overline{Z}_l$ and $\overline{Z}_r$ to be contained in different codeblocks or the same codeblock. In the latter case, it also may happen that the check sets $\mathcal{S}_l$ and $\mathcal{S}_r$ share checks, which then gain support on both auxiliary graphs during deformation.}
    \label{fig:joint_measurement}
\end{figure*}

In this section, we solve this joint measurement problem while guaranteeing that the resulting connected graph satisfies all the desiderata of Theorem~\ref{thm:graph_desiderata}, provided the individual graphs did. A Tanner graph illustration of our construction is shown in Fig.~\ref{fig:joint_measurement}. The connection between the individual graphs is done via a bridge of edges, as originally proposed in Ref.~\cite{cross2024linear}. Here, we instead call that set of edges an adapter, since we will specially choose them to join different graphs while preserving a sparse cycle basis. 

\begin{definition}
Provided two graphs $\mathcal{G}_l=(\mathcal{V}_l,\mathcal{E}_l)$ and $\mathcal{G}_r=(\mathcal{V}_r,\mathcal{E}_r)$ and vertex subsets $\mathcal{V}_l^*\subseteq\mathcal{V}_l$ and $\mathcal{V}_r^*\subseteq\mathcal{V}_r$ of equal size, an adapter is a set of edges $\mathcal{A}\subseteq\mathcal{V}_l^*\times\mathcal{V}_r^*$ defined by a bijective function $a:\mathcal{V}_l^*\rightarrow\mathcal{V}_r^*$ so that $(v_l,v_r)\in\mathcal{A}$ if and only if $a(v_l)=v_r$. We call the resulting graph 
$\mathcal{G}=\mathcal{G}_l\sim_{\mathcal{A}}\mathcal{G}_r=(\mathcal{V}_l\cup\mathcal{V}_r,\mathcal{E}_l\cup\mathcal{E}_r\cup\mathcal{A})$ the adapted graph.
\end{definition}

Given two graphs with sufficient relative expansion on ports $\mathcal{P}_l$ and $\mathcal{P}_r$, it is relatively straightforward to join them with an adapter between subsets of the ports. 

\begin{lemma}\label{lem:adapted_expansion}
If $\mathcal{G}_l=(\mathcal{V}_l,\mathcal{E}_l)$ has relative expansion $\beta_{t_l}(\mathcal{G}_l,\mathcal{P}_l)\ge1$ and $\mathcal{G}_r=(\mathcal{V}_r,\mathcal{E}_r)$ has relative expansion $\beta_{t_r}(\mathcal{G}_r,\mathcal{P}_r)\ge1$, then connecting them with an adapter $\mathcal{A}$ on any subsets $\mathcal{P}_l^*\subseteq\mathcal{P}_l$ and $\mathcal{P}_r^*\subseteq\mathcal{P}_r$ results in an adapted graph $\mathcal{G}=\mathcal{G}_l\sim_{\mathcal{A}}\mathcal{G}_r$ with relative expansion $\beta_t(\mathcal{G},\mathcal{P}_l\cup\mathcal{P}_r)\ge1$ for $t=\min(t_l,t_r,|\mathcal{A}|)$.
\end{lemma}
\begin{proof}

Let $G_l$ and $G_r$ be the incidence matrices of the graphs $\mathcal{G}_l$ and $\mathcal{G}_r$, respectively. We write the incidence matrix of $\mathcal{G}$ as
\begin{align}
    G \; &= \; \bordermatrix{ & \mathcal{V}_l & \mathcal{V}_r \cr
       \mathcal{E}_l & G_l & 0 \cr 
       \mathcal{A} & P_l^\top & P_r^\top \cr
       \mathcal{E}_r & 0 & G_r \cr
       } \qquad
\end{align}
where we have labeled rows and columns by the sets of edges and vertices they represent. The matrix $P_l^\top$ has exactly one 1 per row and one 1 in each column corresponding to $\mathcal{P}_l^*\subseteq\mathcal{V}_l$. Restricted to only those columns, $P_l^\top$ is a permutation matrix $\pi_l$. The remaining columns of $P_l^\top$ are all 0. The same structure holds for matrix $P_r^\top$ and $\mathcal{P}_r^*\subseteq\mathcal{V}_r$ with a permutation matrix $\pi_r$. 

We let $v=(v_l\;v_r)\in\mathbb{F}_2^{|\mathcal{V}_l|+|\mathcal{V}_r|}$ be a vector indicating an arbitrary subset of vertices with $v_l$ and $v_r$ its restriction to the left and right vertices, respectively. Likewise, $u_l$ and $u_r$ represent $v$ restricted to $\mathcal{P}_l$ and $\mathcal{P}_r$ and $u^*_l$ and $u^*_r$ its restriction to $\mathcal{P}_l^*$ and $\mathcal{P}_r^*$.

Making use of expansion,
\begin{align}\nonumber
|vG^\top|&=|v_lG_l^\top|+|v_rG_r^\top|+|P_l^\top v_l^\top + P_r^\top v_r^\top|\\\label{eq:adapter_first_bound}
&\ge\min(t_l,|u_l|,|\mathcal{P}_l|-|u_l|)+\min(t_r,|u_r|,|\mathcal{P}_r|-|u_r|) \nonumber \\
& \qquad + |P_l^\top v_l^\top + P_r^\top v_r^\top|.
\end{align}
Next, we notice that 
\begin{align}
|P_l^\top v_l^\top + P_r^\top v_r^\top| &= |\pi_lu_l^{*\top} +\pi_ru_r^{*\top}|\\
&\ge\max(|u^*_l|-|u^*_r|,|u^*_r|-|u^*_l|)
\end{align}
using the triangle inequality.

We now consider the different cases that result from evaluating the $\mathrm{min}$ functions in Eq.~\eqref{eq:adapter_first_bound}.
\begin{enumerate}
\item For vectors $v$ in which the first $\mathrm{min}$ function evaluates $t_l$ or the second evaluates to $t_r$, we have $|vG^\top|\ge\min(t_l,t_r)$.
\item If $|u_l|\le|\mathcal{P}_l|/2$ and $|u_r|\le|\mathcal{P}_r|/2$, we have $|vG^\top|\ge|u_l|+|u_r|\ge|u_{lr}|$ where by $u_{lr}$ we mean the restriction of $v$ to $\mathcal{P}_l\cup\mathcal{P}_r$.
\item If $|u_l|\ge|\mathcal{P}_l|/2$ and $|u_r|\ge|\mathcal{P}_r|/2$, we have $|vG^\top|\ge|\mathcal{P}_l|-|u_l|+|\mathcal{P}_r|-|u_r|\ge|\mathcal{P}_l\cup\mathcal{P}_r|-|u_{lr}|$.
\item If $|u_l|\le|\mathcal{P}_l|/2$ and $|u_r|\ge|\mathcal{P}_r|/2$, then 
\begin{align}
|vG^\top|&\ge|u_l|+|\mathcal{P}_r|-|u_r|+\max(|u^*_l|-|u^*_r|,|u^*_r|-|u^*_l|) \nonumber \\
&\ge|\mathcal{P}_r|+(|u_l|-|u_l^*|)-(|u_r|-|u_r^*|) \nonumber \\
&\ge|\mathcal{P}_r|-(|u_r|-|u_r^*|) \nonumber \\
&\ge|\mathcal{P}_r^*|=|\mathcal{A}|.
\end{align}
\item If $|u_l|\ge|\mathcal{P}_l|/2$ and $|u_r|\le|\mathcal{P}_r|/2$, then a similar argument yields $|vG^\top|\ge|\mathcal{A}|$.
\end{enumerate}
Combining these cases shows
\begin{align}
|vG^\top|\ge\min(t_l,t_r,|\mathcal{A}|,|u_{lr}|,|\mathcal{P}_l\cup\mathcal{P}_r|-|u_{lr}|)
\end{align}
which proves the relative expansion of the adapted code is as claimed.
\end{proof}

Note that Ref.~\cite{cross2024linear} contains a similar proof arguing for the code distance of their bridged systems. Here, we have abstracted out the relative expansion idea.

If we start with sufficient relative expansion on both initial graphs, i.e.~$\beta_d(\mathcal{G}_l,\mathcal{P}_l)\ge1$ and $\beta_d(\mathcal{G}_r,\mathcal{P}_r)\ge1$ so that the initial graphs satisfy desideratum 4 of Theorem~\ref{thm:graph_desiderata}, then Lemma~\ref{lem:adapted_expansion} says we must only choose an adapter of size $|\mathcal{A}|\ge d$ to ensure the adapted graph $\mathcal{G}$ is sufficiently expanding relative to $\mathcal{P}_l\cup\mathcal{P}_r$, i.e.~$\beta_d(\mathcal{G},\mathcal{P}_l\cup\mathcal{P}_r)\ge1$. This is a relatively mild constraint. The ports $\mathcal{P}_l$ and $\mathcal{P}_r$ are already of size at least $d$, because they were built to connect to logical operators $\overline{Z}_l$ and $\overline{Z}_r$, so it is certainly possible to create an adapter between subsets of the ports of sufficient size.

It is also clear that if the original graphs satisfy desiderata 0, 1, and 2, then the adapted graph will as well. The addition of adapter edges can also create new cycles in the adapted graph. It remains to ensure the resulting adapted graph satisfies desideratum 3, which demands it has a sparse cycle basis. Initially, this may seem complicated to guarantee. However, this is solved neatly by applying the $\mathsf{SkipTree}$ algorithm.

\begin{lemma}\label{lem:skip_adapter}
Consider two graphs $\mathcal{G}_l=(\mathcal{V}_l,\mathcal{E}_l)$ and $\mathcal{G}_r=(\mathcal{V}_r,\mathcal{E}_r)$ along with equal-sized vertex subsets $\mathcal{P}_l^*\subseteq\mathcal{V}_l$ and $\mathcal{P}_r^*\subseteq\mathcal{V}_r$ that induce connected subgraphs of their respective graphs. If the graphs $\mathcal{G}_l$ and $\mathcal{G}_r$ have $(\gamma,\delta)$-sparse cycle bases, there exists an adapter $\mathcal{A}$ between $\mathcal{P}_l^*$ and $\mathcal{P}_r^*$ such that the adapted graph $\mathcal{G}_l\sim_{\mathcal{A}}\mathcal{G}_r$ has a $(\gamma',\delta')$-sparse cycle basis with $\gamma'\le\max(\gamma,8)$ and $\delta'\le \delta+2$.
\end{lemma}
\begin{proof}
Let $A=|\mathcal{P}_l^*|=|\mathcal{P}_r^*|$ and $G_l$ and $G_r$ be the incidence matrices of the graphs $\mathcal{G}_l$ and $\mathcal{G}_r$. Also, denote by $G_l^*$ and $G_r^*$ the incidence matrices of the subgraphs induced by $\mathcal{P}_l^*$ and $\mathcal{P}_r^*$. We apply the $\mathsf{SkipTree}$ algorithm, Theorem~\ref{thm:skiptree}, to both of these induced subgraphs, obtaining $T_l^*,T_r^*,P_l^*,P_r^*$ such that
\begin{equation}
T_l^*G_l^*P_l^*=T_r^*G_r^*P_r^*=H_C\in\mathbb{F}_2^{A\times A}.
\end{equation}
By inserting 0 columns in the $T$ matrices and 0 rows in the $P$ matrices, we can also make matrices $T_l,T_r,P_l,P_r$ such that
\begin{equation}\label{eq:skip_tree_adapter}
T_lG_lP_l=T_rG_rP_r=H_C\in\mathbb{F}_2^{A\times A}.
\end{equation}

Adapter edges $\mathcal{A}$ are added as follows. A vertex in $\mathcal{P}_l^*$ that is labeled $i\in\{0,1,\dots,A-1\}$ by the $\mathsf{SkipTree}$ algorithm applied to $G_l^*$ is connected to the vertex in $\mathcal{P}_r^*$ that is also labeled $i$ by the $\mathsf{SkipTree}$ algorithm applied to $G_r^*$. We can express these connections in the incidence matrix $G$ of the adapted graph $\mathcal{G}=\mathcal{G}_l\sim_\mathcal{A}\mathcal{G}_r$ by writing
\begin{align}
    G \; &= \; \bordermatrix{ & \mathcal{V}_l & \mathcal{V}_r \cr
       \mathcal{E}_l & G_l & 0 \cr 
       \mathcal{A} & P^\top_l & P^\top_r \cr
       \mathcal{E}_r & 0 & G_r \cr
       },
\end{align}
with rows and columns labeled by edge and vertex sets.

We can also explicitly write out a cycle basis $N$ for the adapted graph, in terms of the cycle bases $N_l$ and $N_r$ of the original graphs.
\begin{align}
    N = \left(\begin{array}{ccc}N_l&0&0\\T_l&H_C&T_r\\0&0&N_r\end{array}\right).
\end{align}
One can check using Eq.~\eqref{eq:skip_tree_adapter} that $NG=0$. Moreover, we have added $|\mathcal{A}|=A$ edges to the graphs, and $A-1$ new independent cycles, i.e.~the $A$ cycles from the second block of rows of $N$ minus one since the sum of all those cycles is trivial. This is the expected cycle rank of the adapted graph, so $N$ is complete basis of cycles. We use the fact that $T_l$ and $T_r$ are $(3,2)$-sparse matrices from Theorem~\ref{thm:skiptree} to conclude that this cycle basis is $(\gamma',\delta')$-sparse with $\gamma'\le\max(\gamma,8)$ and $\delta'\le \delta+2$.
\end{proof}

Combining Lemmas~\ref{lem:adapted_expansion} and \ref{lem:skip_adapter} and the above discussion, we obtain the main theorem of this section. To state it in greater generality, we introduce the idea of a \textit{$c$-disjoint} set of logical operators $\{\overline{\Lambda}_i\}_i$ of a code \cite{jochym2018disjointness}, which means that for any qubit of the code, there are no more than $c$ logicals $\overline{\Lambda}_i$ with support on that qubit. If $c$ is constant (e.g.~independent of the code size and the size of the set of logicals), then we say the set of logicals is \textit{sparsely overlapping}.

\begin{theorem}\label{thm:joint_logical_measurement}
Consider a set of nontrivial, sparsely overlapping logical operators $\overline{Z}_0,\overline{Z}_1,\dots,\overline{Z}_{t-1}$ that can all be made $Z$-type simultaneously by applying single-qubit Cliffords. Provided $t$ auxiliary graphs that satisfy the graph desiderata of Theorem~\ref{thm:graph_desiderata} to measure these $t$ logical operators, there exists an auxiliary graph to measure the product $\overline{Z}_0\overline{Z}_1\dots\overline{Z}_{t-1}$ satisfying the desiderata of Theorem~\ref{thm:graph_desiderata}. In particular, if the individual deformed codes are LDPC with distance $d$, the deformed code for the joint measurement is LDPC with check weights and qubit degrees independent of $t$ and has distance $d$.
\end{theorem}
\begin{proof}
To exhibit a simpler proof here that still contains most of the main ideas, we assume the logical operators are all pairwise disjoint (equivalently, that they form a 1-disjoint set). To prove the general theorem for sparsely overlapping logicals requires the notion of a set-valued port function so that qubits can be connected to multiple vertex checks. In Appendix~\ref{app:set-valued_port_function}, we introduce that concept and complete the general proof of this theorem.

Denote the qubit supports of the individual logical operators by $\mathcal{L}_i$ for $i=0,1,\dots,t-1$. Denote the auxiliary graphs used to measure the individual logical operators by $\mathcal{G}_i=(\mathcal{V}_i,\mathcal{E}_i)$, the port functions by $f_i:\mathcal{L}_i\rightarrow\mathcal{V}_i$, and the ports by $\mathcal{P}_i=f_i(\mathcal{L}_i)$. 

Recall a $Z$-type logical operator $\overline{Z}$ is said to be irreducible if there is no other $Z$-type logical operator or $Z$-type stabilizer supported entirely on its support. For all $i$, there is some set of qubits $\mathcal{L}'_i\subseteq\mathcal{L}_i$ that supports a nontrivial, irreducible logical operator. Note that the logical operator supported on $\mathcal{L}'_i$ need not be logically equivalent to that on $\mathcal{L}_i$. We let $\mathcal{P}'_i=f_i(\mathcal{L}'_i)\subseteq\mathcal{P}_i$.

Because $Z(\mathcal{L}'_i)$ is irreducible, the $X$-type part of the check matrix of the original code restricted to $\mathcal{L}'_i$ is equivalent to the check matrix of the repetition code \cite{cross2024linear} (see also Appendix~\ref{app:supportlemma}). Because of desideratum 2 of Theorem~\ref{thm:graph_desiderata}, the graph $\mathcal{G}_i$ must have short perfect matchings for each of these checks. If two vertices $x,y\in\mathcal{V}_i$ are matched this way, we modify graph $\mathcal{G}_i$ by adding the edge $(x,y)$ if it does not already exist. This ensures the subgraph induced by $\mathcal{P}'_i=f_i(\mathcal{L}'_i)$ is connected, enabling the use of Lemma~\ref{lem:skip_adapter}. Moreover, the new graph necessarily satisfies all the desiderata because the original graph did as well. In particular, note that we can add a cycle to the cycle basis for each added edge $(x,y)$ which consists of the edge itself and the short path between $(x,y)$, and that this must preserve desideratum 3b (cycle sparsity) because the short paths originally satisfied desideratum 2b (matching sparsity). From now on, we just assume the graphs $\mathcal{G}_i$ have already been modified in this way so that the subgraphs induced by $\mathcal{P}'_i$ are connected.

Now we create adapters to join all the graphs together. In particular, for each $i=0,1,\dots,t-2$, we create an adapter $\mathcal{A}_i$ between $\mathcal{P}_i^*\subseteq\mathcal{P}'_i$ and $\mathcal{P}_{i+1}^*\subseteq\mathcal{P}'_{i+1}$. These subsets can always be chosen so that $|\mathcal{P}_i^*|=|\mathcal{P}_{i+1}^*|\ge d$ and both $\mathcal{P}_i^*$ and $\mathcal{P}_{i+1}^*$ induce connected subgraphs in their respective auxiliary graphs $\mathcal{G}_i$ and $\mathcal{G}_{i+1}$. Iterative application of Lemma~\ref{lem:skip_adapter} creates the adapted graph
\begin{equation}
\mathcal{G}=(\mathcal{V},\mathcal{E})=\mathcal{G}_0\sim_{\mathcal{A}_0}\mathcal{G}_1\sim_{\mathcal{A}_1}\mathcal{G}_2\sim_{\mathcal{A}_2}\dots\sim_{\mathcal{A}_{t-2}}\mathcal{G}_{t-1}.
\end{equation}

The port function for this adapted graph is defined as follows. Note, $\overline{Z}=\overline{Z}_0\overline{Z}_1\dots\overline{Z}_{t-1}$ is supported on exactly the qubits $\mathcal{L}=\mathcal{L}_0\cup\mathcal{L}_1\cup\dots\cup\mathcal{L}_{t-1}$. Let $f:\mathcal{L}\rightarrow\mathcal{V}$ be defined as $f(q)=f_i(q)$ when $q\in\mathcal{L}_i$. We let $\mathcal{P}=f(\mathcal{L})=\mathcal{P}_0\cup\mathcal{P}_1\cup\dots\cup\mathcal{P}_{t-1}$ be the port vertices in the graph $\mathcal{G}$.

The adapted graph $\mathcal{G}$ and port function $f$ satisfy the graph desiderata of Theorem~\ref{thm:graph_desiderata}. Desiderata 0, 1, 2 are inherited from the original graphs. We used Lemma~\ref{lem:skip_adapter} to construct adapters that guarantee desideratum 3 is satisfied. Lemma~\ref{lem:adapted_expansion} implies $\beta_d(\mathcal{G},f(\mathcal{L}))=\beta_d(\mathcal{G},\mathcal{P})\ge1$ and thus we have desideratum 4.
\end{proof}

One potential inconvenience of the adapter construction in this section is that the adapting edges are connected directly to the port that is also connected directly to the original code. In Ref.~\cite{cross2024linear} this is what is done for the construction used to measure logical $\overline{Y}=i\overline{X}\overline{Z}$. However, also in Ref.~\cite{cross2024linear} they connect these adapter edges (there called a bridge system) instead to higher levels of the thickened auxiliary graph when they measure a product of same-type logical operators like $\overline{Z}_l\overline{Z}_r$. It would be good to understand when this can be done in more generality. Perhaps this would require a notion of expansion relative to several ports simultaneously.

We also remark that to perform specific join logical measurements in practice, we believe it is likely that the construction in the proof of Theorem~\ref{thm:joint_logical_measurement} will typically be used without the explicit guarantees that the original auxiliary graphs have sufficient relative expansion. One should then simply verify directly that the deformed code defining the joint logical measurement has code distance $d$. We provide concrete examples in Section~\ref{sec:examples} that demonstrate this idea in the setting of finite-sized codes.

\section{Leveraging code properties to simplify auxiliary graph LDPC surgery} \label{sec:leveraging_code_properties}

The graph desiderata theorem \ref{thm:graph_desiderata} provided a checklist of desirable properties for an auxiliary graph to satisfy in order to measure a chosen logical operator while preserving code distance and sparsity of checks. 
Exploring the extent to which we can avoid some of these requirements will be the focus on this section.

\subsection{Expansion is unnecessary for measuring some joint logical operators} \label{subsec:replace_desideratum_4}

We noted previously that desideratum 4, that $\mathcal{G}$ has sufficient expansion relative to the image of port function $f$, is only sufficient for the deformed code to have the same code distance as the original code and not necessary. However, if we remove desideratum 4, we generally need to impose some other conditions on the codes and logical operators being measured to ensure we can build an auxiliary graph for which we can prove the deformed code preserves the code distance. We present one set of alternative conditions here for the case of measuring a joint logical operator $\overline{Z}_l\overline{Z}_r$.

\begin{enumerate}[label=(\alph*)]
    \item The operators $\overline{Z}_l$ and $\overline{Z}_r$ come from separate codeblocks (before deformation).
    \item The operator $\overline{Z}_l$ or $\overline{Z}_r$ with the larger weight is also a minimum weight logical operator of its code. If $|\overline{Z}_l|=|\overline{Z}_r|$, then only one of the two operators needs to be minimum weight in its code.
\end{enumerate}

To clarify condition (b), it may for instance happen that $|\overline{Z}_l|>|\overline{Z}_r|$. In this case, we demand that $|\overline{Z}_l|=d_l$, where $d_l$ is the distance of the left code.

\begin{figure}[t]
    \centering
\begin{tikzpicture}[]
\pgfmathsetmacro{\ewid}{0.8pt}
\pgfmathsetmacro{\drwid}{0.86pt}
\pgfmathsetmacro{\sqrsz}{0.8cm}
\pgfmathsetmacro{\sqrdf}{0.16}
\pgfmathsetmacro{\crcsz}{0.45}
\pgfmathsetmacro{\crcdf}{0.12}
\pgfmathsetmacro{\basecodewid}{1.2}
\pgfmathsetmacro{\brnd}{0.012cm}
\definecolor{edgcol}{rgb}{0, 0, 0}
\definecolor{bl}{rgb}{0.63, 0.79, 0.95}

 \fill[bl, opacity=0.3]
        (-1, -0.32) arc[start angle=180, end angle=270, x radius=\brnd, y radius=\brnd] -- ++(\basecodewid, 0) arc[start angle=270, end angle=360, x radius=\brnd, y radius=\brnd] -- ++(0, 3) arc[start angle=0, end angle=90, x radius=\brnd, y radius=\brnd] -- ++(-\basecodewid, 0) arc[start angle=90, end angle=180, x radius=\brnd, y radius=\brnd]-- cycle;

        \draw[draw=edgcol,line width=\ewid] (-0.9,0) to (5.5,0);
        \draw[draw=edgcol,line width=\ewid] (-0.9,2.2) to (5.5,2.2);
        \draw[draw=edgcol,line width=\ewid] (0,0) to (0,2.2);
        
        
        \draw[fill=white,line width=\drwid] (0-\crcdf,0+\crcdf) circle (\crcsz);
        \draw[fill=white,line width=\drwid] (0,0) circle (\crcsz) node {$\bar{Z}_l$};
        \node[draw, line width=\drwid, fill=white, minimum size=\sqrsz] at (0-\sqrdf, 2.2+\sqrdf) {};
        \node[draw, line width=\drwid, fill=white,minimum size=\sqrsz] at (0, 2.2) {$X$};

        \node[] at (-0.28, 1.2) {{\small $S_l$}};
        
        \node[] at (1.2, 0.2) {{\small $F_l$}};
        \node[] at (1.3, 2.5) {{\small $M_l$}};

        \node[] at (0.5, -0.4) {{ $\mathcal{L}_l$}};
        \node[] at (0.6, 1.65) {{$\mathcal{S}_l$}};

        \begin{scope}[shift={(2.4,0)}]
        \draw[draw=edgcol,line width=\ewid] (0,0) to (0,4.5);
        \draw[fill=white,line width=\drwid] (0-\crcdf, 2.2+\crcdf) circle (\crcsz);
        \draw[fill=white,line width=\drwid] (0,2.2) circle (\crcsz) node {};
         \node[draw, line width=\drwid,fill=white, minimum size=\sqrsz] at (0-\sqrdf,0+\sqrdf) {};
        \node[draw, line width=\drwid,fill=white,minimum size=\sqrsz] at (0, 0) {$Z$};
        \node[draw, line width=\drwid,fill=white, minimum size=\sqrsz] at (0-\sqrdf,4.4+\sqrdf) {};
        \node[draw, line width=\drwid,fill=white,minimum size=\sqrsz] at (0, 4.4) {$X$};
        \node[] at (-0.3, 1.2) {{\small $G^{\top}$}};
        \node[] at (1.2, 2.5) {{\small $M_r$}};
        \node[] at (1.2, 0.2) {{\small $F_r$}};
        \node[] at (-0.25, 3.4) {{\footnotesize $N$}};
        \node[] at (0.5, 1.8) {{ $\mathcal{E}$}};
        \node[] at (0.55, -0.4) {{ $\mathcal{V}$}};
        \node[] at (0.55, 3.8) {{ $\mathcal{U}$}};
        \end{scope}

        \begin{scope}[shift={(4.8,0)}]
         \fill[bl, opacity=0.3]
        (-0.8, -0.32) arc[start angle=180, end angle=270, x radius=\brnd, y radius=\brnd] -- ++(\basecodewid, 0) arc[start angle=270, end angle=360, x radius=\brnd, y radius=\brnd] -- ++(0, 3) arc[start angle=0, end angle=90, x radius=\brnd, y radius=\brnd] -- ++(-\basecodewid, 0) arc[start angle=90, end angle=180, x radius=\brnd, y radius=\brnd]-- cycle;

        \draw[draw=edgcol,line width=\ewid] (0,0) to (0,2);
        \node[draw, line width=\drwid,fill=white, minimum size=\sqrsz] at (0-\sqrdf, 2.2+\sqrdf) {};
        \node[draw, line width=\drwid,fill=white,minimum size=\sqrsz] at (0, 2.2) {$X$};
        \draw[fill=white,line width=\drwid] (0-\crcdf,0+\crcdf) circle (\crcsz);
        \draw[fill=white,line width=\drwid] (0,0) circle (\crcsz) node {{ $\bar{Z}_r$}};
        \node[] at (-0.25, 1.2) {{\small $S_r$}};
        \node[] at (0.53, 1.7) {{ $\mathcal{S}_r$}};
        \node[] at (0.51, -0.4) {{ $\mathcal{L}_r$}};
        \end{scope}

\end{tikzpicture}

\caption{Joint measurement of $\overline{Z}_l\overline{Z}_r$ using a bespoke graph $G$ instead of an adapter. We show that $G$ does not need to be expanding under certain conditions.}
\label{fig:expansionless_joint}
\end{figure}


\begin{theorem}
\label{thm:expansionless_joint}
We can construct an auxiliary graph such that the joint measurement between quantum LDPC codes depicted in Fig.~\ref{fig:expansionless_joint} measures only $\overline{Z}_l\overline{Z}_r$ and is LDPC. The size of the auxiliary graph is $O(D\log^3D)$ where $D=\max(|\overline{Z}_l|,|\overline{Z}_r|)$. If additionally condition (a) is satisfied, the deformed code has distance at least $d_l+d_r-D$, where we denote the code distances of the left and right codes as $d_l$ and $d_r$, respectively. If additionally condition (b) is satisfied, the deformed code has distance at least $d=\min(d_l,d_r)$.
\end{theorem}
\begin{proof}
As we have done before, we assume without loss of generality that both $\overline{Z}_l$ and $\overline{Z}_r$ are entirely $Z$-type, because we can always perform single-qubit Cliffords to make it so. We also use $\mathcal{L}_l$ and $\mathcal{L}_r$ to denote the qubit supports of operators $\overline{Z}_l=Z(\mathcal{L}_l)$ and $\overline{Z}_r=Z(\mathcal{L}_r)$. Furthermore, we make two more assumptions.
\begin{enumerate}[label=(\roman*)]
\item $|\overline{Z}_l|\ge|\overline{Z}_r|$ without loss of generality.
\item Choose a port function $f:\mathcal{L}_l\cup\mathcal{L}_r\rightarrow\mathcal{V}$ such that the port set of vertices for the left code $f(\mathcal{L}_l)$ is a superset of the port set of vertices for the right code $f(\mathcal{L}_r)$. Provided (i), this can always be done. Moreover, for all vectors $v\in\mathcal{V}$, it implies $|vF_l|\ge|vF_r|$ and $|vF_l|-|vF_r|\le|\overline{Z}_l|-|\overline{Z}_r|$, where $F_l,F_r$ are the matrices in Fig.~\ref{fig:expansionless_joint}.
\end{enumerate}

We note that (ii) in particular means that the port function is non-injective, i.e.~there are vertices connected to two qubits, one from the left code and one from the right. However, we noted in Theorem~\ref{thm:graph_desiderata} that desiderata 0-3 apply also to a non-injective port function. Thus, we can build a graph using the techniques of Section~\ref{sec-gauging-meas} that has just $D=\max(|\overline{Z}_l|,|\overline{Z}_r|)=|\overline{Z}_l|=|\text{im}f|$ port vertices, connected one-to-one to qubits in $\mathcal{L}_l$, and a subset of them connected one-to-one to $\mathcal{L}_r$ also. To ensure desideratum 2, we can add edges to the set of $D$ port vertices using the vertex pairing strategy described just below the proof of Theorem~\ref{thm:graph_desiderata}. This necessitates adding edges for both the left code checks overlapping $\overline{Z}_l$ and the right code checks overlapping $\overline{Z}_r$. Still, each port vertex has $O(1)$ degree after this process, satisfying desideratum 1. Additional edges can be added if necessary to ensure the graph is connected and satisfies desideratum 0. Finally, once the cycles are sparsified using thickening to satisfy desideratum 3, the final auxiliary graph contains a total of $O(D\log^3D)$ vertices and edges.

Since we are lacking desideratum 4, we must rely on the conditions (a) and (b) to prove lower bounds on the deformed code distance. We make use of just (a) to begin and only add (b) at the end of the proof.

Consider a nontrivial logical Pauli $\overline{\Lambda}$ from the original codes that commutes with $\overline{Z}_l\overline{Z}_r$ but is not equivalent modulo stabilizers to it. In other words, $\overline{\Lambda}$ represents any logical operator on logical qubits that are not measured. Write $\overline{\Lambda}=\overline{\Lambda}_l\overline{\Lambda}_r$ where $\overline{\Lambda}_l=\Lambda_l^X\Lambda_l^Z$ is supported entirely on the left code, $\overline{\Lambda}_r=\Lambda_r^X\Lambda_r^Z$ is supported entirely on the right, and we have denoted their Pauli $X$ and $Z$ components.

In the deformed code, $\overline{\Lambda}$ may pick up $X$-type support on the edge qubits of the auxiliary graph. Specifically, there will be some perfect matching $e\in\mathcal{E}$ on the auxiliary graph that ensures $\overline{\Lambda}X(e\in\mathcal{E})$ commutes with all vertex $Z$ checks. We ignore this edge support, and show that, for any choice of vertex checks $v$, $\overline{\Lambda}\mathcal{H}_Z(v\in\mathcal{V})$ has weight at least $d$ on just the original code qubits alone. Thus the distance of the deformed code is at least $d$.

\begin{figure*}[t]
    \centering
    \includegraphics[width=0.8\textwidth]{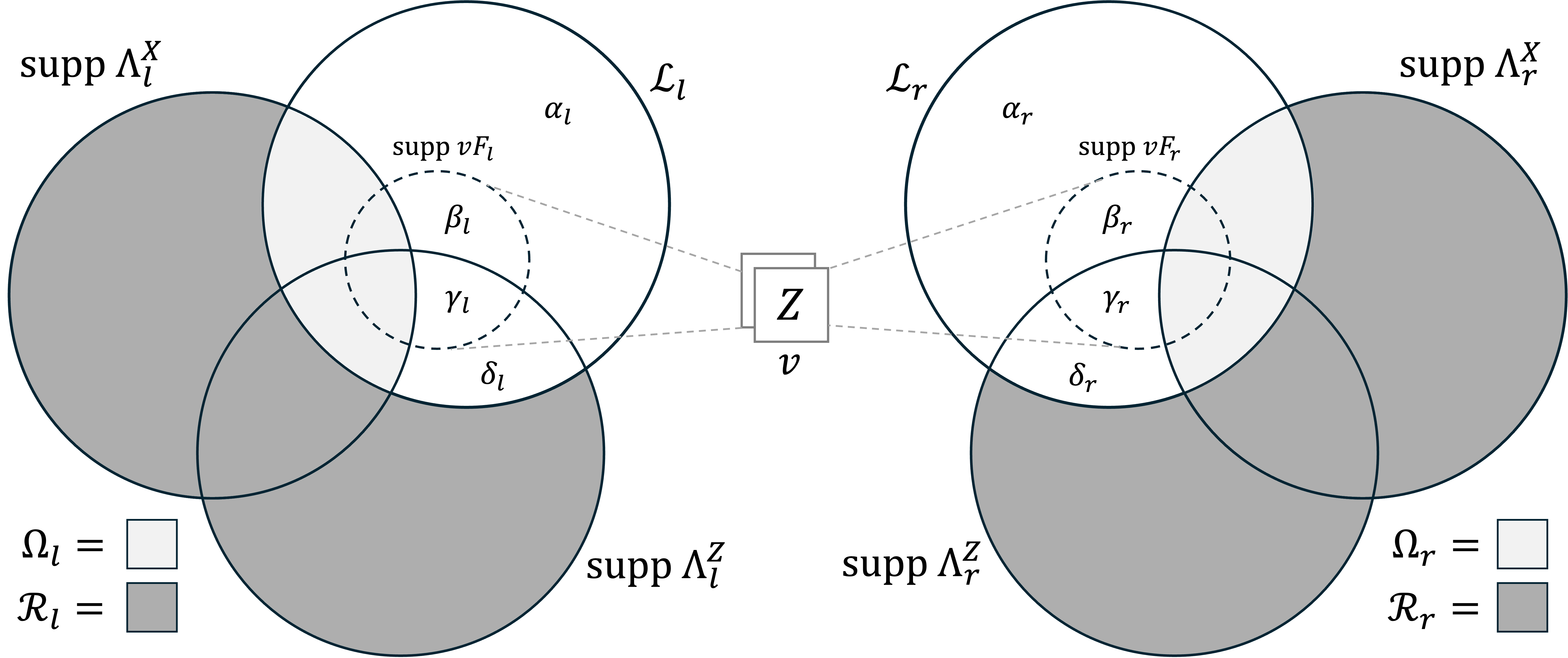}
    \caption{Notation for the proof of Theorem~\ref{thm:expansionless_joint}. On the left (right), all sets shown are subsets of $l$ ($r$) code qubits. Vertex checks $v$ are connected to sets of qubits $\supp vF_l$ and $\supp vF_r$ indicated by dashed circles.}
    \label{fig:expansionless_joint_venn}
\end{figure*}

We establish notation in Fig.~\ref{fig:expansionless_joint_venn}. In particular, we denote $\mathcal{R}_l=\supp \overline{\Lambda}_l\setminus\mathcal{L}_l$ and $\Omega_l=\supp\Lambda^X_l\cap\mathcal{L}_l$. We note the qubits connected to the checks indicated by $v$, namely $\supp vF_l$, are entirely contained within $\mathcal{L}_l$. The intersection of $\supp vF_l$ with $\Lambda_l^Z$ defines regions we denote by $\alpha_l,\beta_l,\gamma_l$, and $\delta_l$. We make the symmetric definitions for code $r$. In this notation, we wish to bound the weight of the following
\begin{equation}\label{eq:bound_this_reinterpreted}
|\overline{\Lambda}\mathcal{H}_Z(v\in\mathcal{V})|\ge|\beta_l|+|\delta_l|+|\Omega_l|+|\mathcal{R}_l|+|\beta_r|+|\delta_r|+|\Omega_r|+|\mathcal{R}_r|.
\end{equation}

Also, we can interpret condition (ii) in the same notation. Note that $|\beta_l|+|\gamma_l|\le|vF_l|\le|\beta_l|+|\gamma_l|+|\Omega_l|$ and likewise for the $r$ code. Therefore, $|vF_l|\ge|vF_r|$ and $|vF_l|-|vF_r|\le|\overline{Z}_l|-|\overline{Z}_r|$ from condition (ii) imply
\begin{align}\label{eq:e1_reinterpreted}
|\beta_l|+|\gamma_l|+|\Omega_l|&\ge|\beta_r|+|\gamma_r|,\\\label{eq:e2_reinterpreted}
|\beta_r|+|\gamma_r|+|\Omega_r|-|\overline{Z}_r|&\ge|\beta_l|+|\gamma_l|-|\overline{Z}_l|.
\end{align}

Since $\overline{\Lambda}$ is not equivalent to $\overline{Z}_l\overline{Z}_r$ modulo stabilizers, then $|\overline{\Lambda}_l\overline{Z}_l|\ge d_l$ or $|\overline{\Lambda}_r\overline{Z}_r|\ge d_r$ (possibly both). We assume that $|\overline{\Lambda}_l\overline{Z}_l|\ge d_l$. This assumption holds without loss of generality because if this is not the case, we can consider $\overline{\Lambda}\mathcal{H}_Z(\vec1\in\mathcal{V})=\overline{\Lambda}\;\overline{Z}_l\overline{Z}_r$ instead of $\overline{\Lambda}$, and the problem of bounding the weight $|\overline{\Lambda}\mathcal{H}_Z(v\in\mathcal{V})|$ for all $v$ remains unchanged. Translating to the notation of Fig.~\ref{fig:expansionless_joint_venn}, $d_l\le|\overline{\Lambda}_l\overline{Z}_l|=|\overline{Z}_l|-|\gamma_l|-|\delta_l|+|\mathcal{R}_l|$ and
\begin{equation}\label{eq:LlZ_reinterpreted}
|\mathcal{R}_l|+(|\overline{Z}_l|-d_l)\ge|\gamma_l|+|\delta_l|.
\end{equation}

Now, we proceed in two cases. Either $|\overline{\Lambda}_r|$ is a nontrivial logical operator of the $r$ code or it is trivial (i.e.~a stabilizer). In the first case, $d_r\le|\overline{\Lambda}_r|=|\gamma_r|+|\delta_r|+|\Omega_r|+|\mathcal{R}_r|$. Starting from Eq.~\eqref{eq:bound_this_reinterpreted}, we then have
\begin{align*}
|\overline{\Lambda}\mathcal{H}_Z(v\in\mathcal{V})|&\ge -|\gamma_l|+|\delta_l|+|\mathcal{R}_l|+2|\beta_r|+|\gamma_r| && \\ & \qquad +|\delta_r|  +|\Omega_r|+|\mathcal{R}_r|&&\text{[Eq.~\eqref{eq:e1_reinterpreted}]} \\
&\ge -|\gamma_l|+|\delta_l|+|\mathcal{R}_l|+2|\beta_r|+d_r&&\text{[$1^{\text{st}}$ case]}\\
&\ge d_r+2|\delta_l|+2|\beta_r|-(|\overline{Z}_l|-d_l)&&\text{[Eq.~\eqref{eq:LlZ_reinterpreted}]}\\
&\ge d_l+d_r-|\overline{Z}_l|.
\end{align*}
In the second case, $\overline{\Lambda}_r$ is trivial, but still the weight of $\overline{\Lambda}_r\overline{Z}_r$ cannot be less than $d_r$. Thus, $|\overline{\Lambda}_r\overline{Z}_r|=|\overline{Z}_r|-|\gamma_r|-|\delta_r|+|\mathcal{R}_r|\ge d_r$ and $|\mathcal{R}_r|+(|\overline{Z}_r|-d_r)\ge|\gamma_r|+|\delta_r|$. Also, if $\overline{\Lambda}_r$ is trivial then $\overline{\Lambda}_l$ must be nontrivial, so $|\overline{\Lambda}_l|=|\gamma_l|+|\delta_l|+|\Omega_l|+|\mathcal{R}_l|\ge d_l$. Therefore, we have
\begin{align*}
|\overline{\Lambda}\mathcal{H}_Z(v\in\mathcal{V})|&\ge 2|\beta_l|-|\overline{Z}_l|+|\overline{\Lambda}_l|-|\gamma_r|+|\delta_r| \\ & \qquad +|\mathcal{R}_r|+|\overline{Z}_r|&&\text{[Eq.~\eqref{eq:e2_reinterpreted}]}\\
&\ge d_r+2|\beta_l|+2|\delta_r|-(|\overline{Z}_l|-d_l)&&\text{[$2^{\text{nd}}$ case]}\\
&\ge d_l+d_r-|\overline{Z}_l|.
\end{align*}
We conclude from both cases that if condition (a) is satisfied, the deformed code distance is at least $d_l+d_r-|\overline{Z}_l|$ and this equals $d_l+d_r-\max(|\overline{Z}_l|,|\overline{Z}_r|)$ by applying (i).

Assuming condition (b), i.e.~$|\overline{Z}_l|=d_l$, this lower bound on the deformed code distance becomes $d_r$, which equals $d=\min(d_l,d_r)$ because $d_r\le|\overline{Z}_r|\le|\overline{Z}_l|=d_l$ using (i) and (b) together. 
\end{proof}

Some intuition for this result comes from considering special cases. If we have a logical operator $\overline{\Lambda}_l$ in the left code, it can lose weight in couple ways once it is multiplied by vertex checks. First, we could multiply it by vertex checks that are not even connected to the right code (because the left port vertices are merely a superset of the right port vertices). Then, without expansion in the auxiliary graph, it could simply lose that support. However, it cannot lose more than $|\overline{Z}_l|-|\overline{Z}_r|$ this way and so its weight remains at least $|\overline{Z}_r|\ge d_r$. A second way $\overline{\Lambda}_l$ can lose weight is if its weight is transferred to the right code through vertex checks that are connected to both codes, and then that weight can be reduced in weight by checks in the right code. This would, however, imply that $\overline{Z}_r$ can be reduced in weight by these checks as well, so the weight is not reduced by more than $|\overline{Z}_r|-d_r$. Therefore, the weight of $\overline{\Lambda}_l$ is still at least $d_l+d_r-|\overline{Z}_r|\ge d_l+d_r-|\overline{Z}_l|\ge d_r$. 

Notice that both these weight-reducing scenarios do not apply to a logical operator on the right code $\overline{\Lambda}_r$ instead of $\overline{\Lambda}_l$. The first does not apply because all vertex checks connected to the right code are also connected to the left. The second does not apply because there is no way to reduce the weight of $\overline{Z}_l$ by stabilizers -- it is already a minimum weight logical.

While the intuition covers only special cases, the proof of Theorem~\ref{thm:expansionless_joint} covers the most general case. Nevertheless, there is still some benefit to this intuition when assessing other scenarios. For instance, suppose the left code encodes just two qubits, 1 and 2, and we want to measure \emph{two} joint operators $\overline{Z}^{(1)}_l\overline{Z}_r$ and $\overline{X}^{(2)}_l\overline{X}_r$ simultaneously, where operators $\overline{Z}_r$ and $\overline{X}_r$ are from the separate right codeblock. We might consider introducing two separate auxiliary graphs, one for measuring $\overline{Z}^{(1)}_l\overline{Z}_r$ and one for $\overline{X}^{(2)}_l\overline{X}_r$, and proceed similarly to Theorem~\ref{thm:expansionless_joint}. Now, a logical operator $\overline{\Lambda}_r$ on the right might be multiplied by both sets of vertex checks together, and we must ensure that combined weight transferred to the left code cannot be reduced in weight by checks of the left code. This necessitates a more involved condition on the logical operators $\overline{Z}^{(1)}_l$ and $\overline{X}^{(2)}_l$, but one that it turns out is still satisfiable. For instance, the theorem below covering this scenario is applicable to the case in which the left code is the toric code.

\begin{theorem}\label{thm:expansionless_joint_toric}
Let $\overline{Z}_r$ and $\overline{X}_r$ be arbitrary non-overlapping logical operators in a distance $d_r$ quantum LDPC code, referred to as the right code. Consider another distance $d_l\ge\max(|\overline{Z}_r|,|\overline{X}_r|)$ quantum LDPC code, the left code, encoding just two logical qubits and possessing two non-overlapping, weight $d_l$ logical operators $\overline{Z}_l$ and $\overline{X}_l$. Suppose the weight of $\overline{Z}_l\overline{X}_l$ cannot be reduced to less than $2d_l$ by multiplying by stabilizers and logical operators of the left code other than $\overline{Z}_l$, $\overline{X}_l$, and $\overline{Z}_l\overline{X}_l$. The toric code is an example of such a left code. Then, we can construct two auxiliary graphs, each of size $O(d_l\log^3d_l)$ to measure $\overline{Z}_l\overline{Z}_r$ and $\overline{X}_l\overline{X}_r$, and only those logical operators, simultaneously. Moreover, the deformed code is LDPC and has distance at least $d_r$.
\end{theorem}
\begin{proof}
The proof is very similar to that of Theorem~\ref{thm:expansionless_joint} and we defer it to Appendix~\ref{app:proof_expansionless_joint_toric}.
That the toric code is a valid left code is left to Section~\ref{subsubsec:static-gmerge-param}, which focuses on adapting to the toric code.
\end{proof}

While we have studied just two scenarios in which desideratum 4 can be replaced in this section, there are very likely many other scenarios. As other works similarly surmise \cite{cross2024linear,williamson2024gauging,ide2024faulttolerant} about the requirement of graph expansion, we suspect desideratum 4 too is overkill in many practical settings.

\subsection{Decongestion is unnecessary for geometrically local codes} \label{subsec:geo_local} 

An inconvenience in the construction of the repetition code adapter is the necessity in general to thicken the deformed code to ensure a deformed code that is LDPC. Satisfying desideratum 3 of Theorem~\ref{thm:graph_desiderata} results in an additional $\log^3d$ factor in the space overhead. Here we note that this factor can be removed for the special case of geometrically local codes in constant dimensions $D\ge2$. In particular, for any logical operator of such a code we will show there exists a way to construct the graph $\mathcal{G}$ such that $\mathcal{G}$ has a cycle basis where each edge participates only in a constant number of cycle basis elements (here, the constant is independent of the code size, but dependent on the dimension $D$). This implies that thickening the deformed code is not necessary for our adapter constructions and for auxiliary graph surgery of a logical operator if the original code is geometrically local.

Consider a $D$-dimensional Euclidean space where qubits have a constant density, $\rho$. That is, for any $D$-dimensional hypersphere of radius $R$, the number of qubits within this ball is given by: $n_R = \rho R^D$. We define a family of \textit{geometrically-local stabilizer codes} to be one where we can define a set of stabilizer generators $\mathcal{S} = \langle S_i \rangle$, such that every generator $S_i$ is supported within a constant radius~$R_\mathcal{S}$, independent of system size.

As described in Section~\ref{sec-gauging-meas}, the task is, given a logical operator with support on a set of qubits~$\mathcal{L}$, to construct a graph $\mathcal{G} = (\mathcal{V},\mathcal{E})$ such that each vertex is associated with a qubit from~$\mathcal{L}$. Since the code is geometrically local, we can choose a physical layout of the vertices~$\mathcal{V}$ according to their position in the original geometric lattice. In addition, note that any irreducible logical operator will preserve some form of local structure in $D$~dimensions (or lower) in the following sense. If one considers a given vertex corresponding to one of the qubits within the support of~$\mathcal{L}$, then there must be a stabilizer in the original code that anti-commutes with the Pauli supported on this qubit of $\mathcal{L}$ (otherwise this single-qubit physical Pauli operator would also be a logical operator, which contradicts the irreducibility assumption). Secondly, there must be another qubit in the support of that stabilizer nearby, that is within a hypersphere of radius~$R_\mathcal{S}$ by definition of the geometrically-local code. For an irreducible logical operator~$\mathcal{L}$ we can recursively construct this mesh of qubits. By construction, any $D$-dimensional hypersphere of radius~$R_\mathcal{S}$ contains a number of qubits from~$\mathcal{L}$ that is lower bounded by~$cR_\mathcal{S}$ for some constant $c > 0$ independent of system size. In general, if the logical operator can be embedded in some lower-dimensional manifold, then the number of points from~$\mathcal{L}$ contained within a $D$-dimensional hypersphere of radius~$B$ will scale proportionately to $B^\alpha$, for $1\le \alpha \le D$, which will only strengthen the results that follow.
Note, this excludes the case of trying to measure the product of two disjoint logical operators. However, in that case the results in this section will hold for the individual graphs of each of the logical operators, which can then be connected using a repetition code adapter to avoid measuring each individually, as described in Section~\ref{sec:rep_adapter}.  

We now turn to the question of providing the edges of the graph~$\mathcal{G}$. We will choose the edges according to the Delaunay triangulation of points in $\mathbb{R}^D$~\cite{delaunay}. The Delaunay triangulation is defined as the partitioning of the system into simplices such that for any circum-hypersphere of points in a given simplex contains no other points. Such a triangulation is unique unless there exists a hypersphere containing $D+2$ or more points on its boundary without any points in the interior (referred to as \textit{special} systems~\cite{delaunay}), as then there may be multiple choices for the triangulation \footnote{For example, consider the four vertices at the corner of a square: they all reside on the same circle centered at the middle of the square with the appropriate radius. Their triangulation is not unique as one can choose either of the diagonals in the edge set.}. For simplicity, we will assume the uniqueness of the triangulation. In cases where there is a degeneracy in the choice due to the above condition, one can find a slight perturbation of the lattice to recover uniqueness and the remaining results will still hold~(proof can be found in Prop.~1 of Ref.~\cite{delaunay}). 

\begin{figure}
    \centering
    (a)\includegraphics[width=\linewidth,trim={5cm 5cm 5cm 5cm},clip]{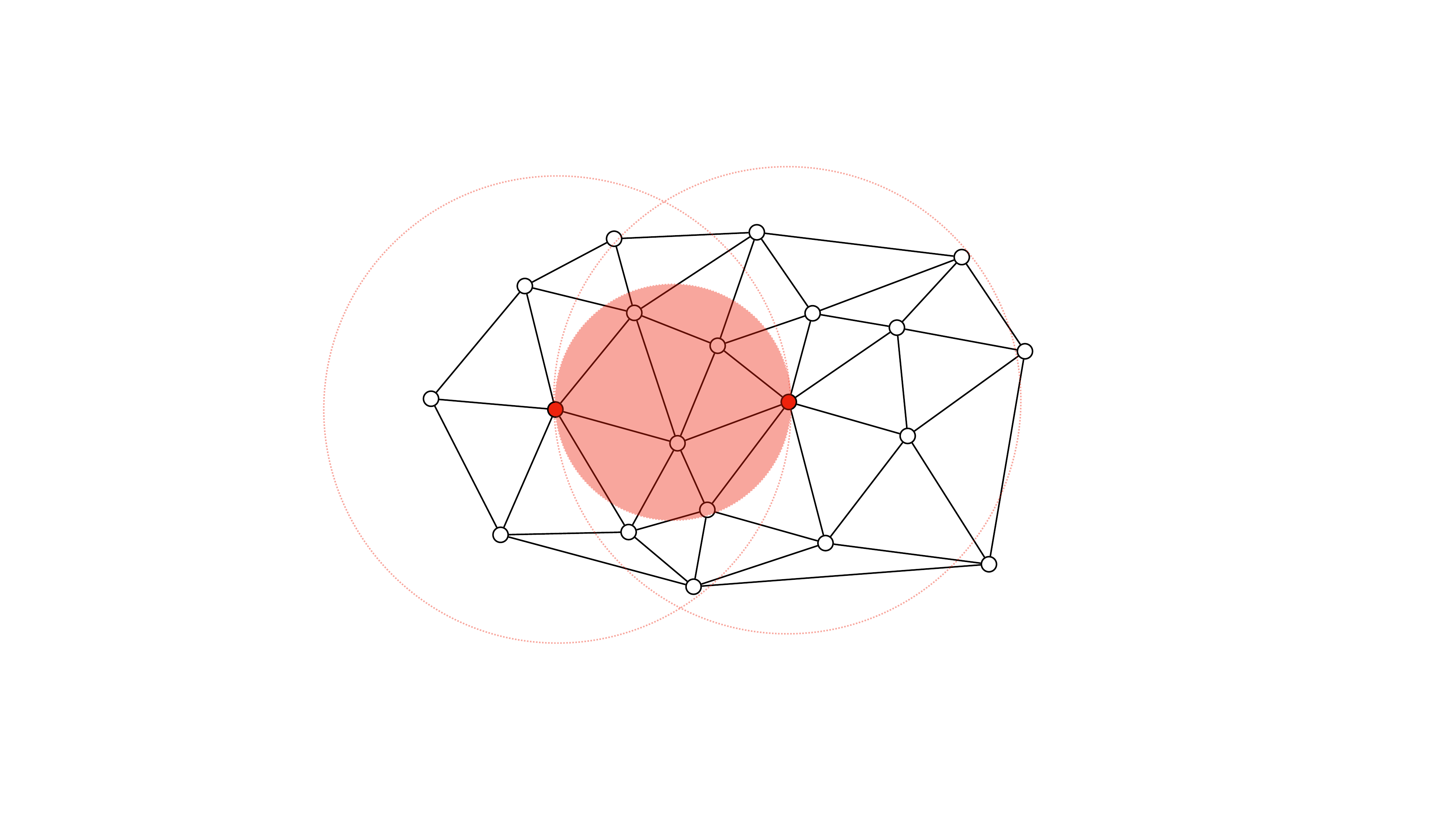}
    (b)\includegraphics[width=\linewidth,trim={5cm 5cm 5cm 5cm},clip]{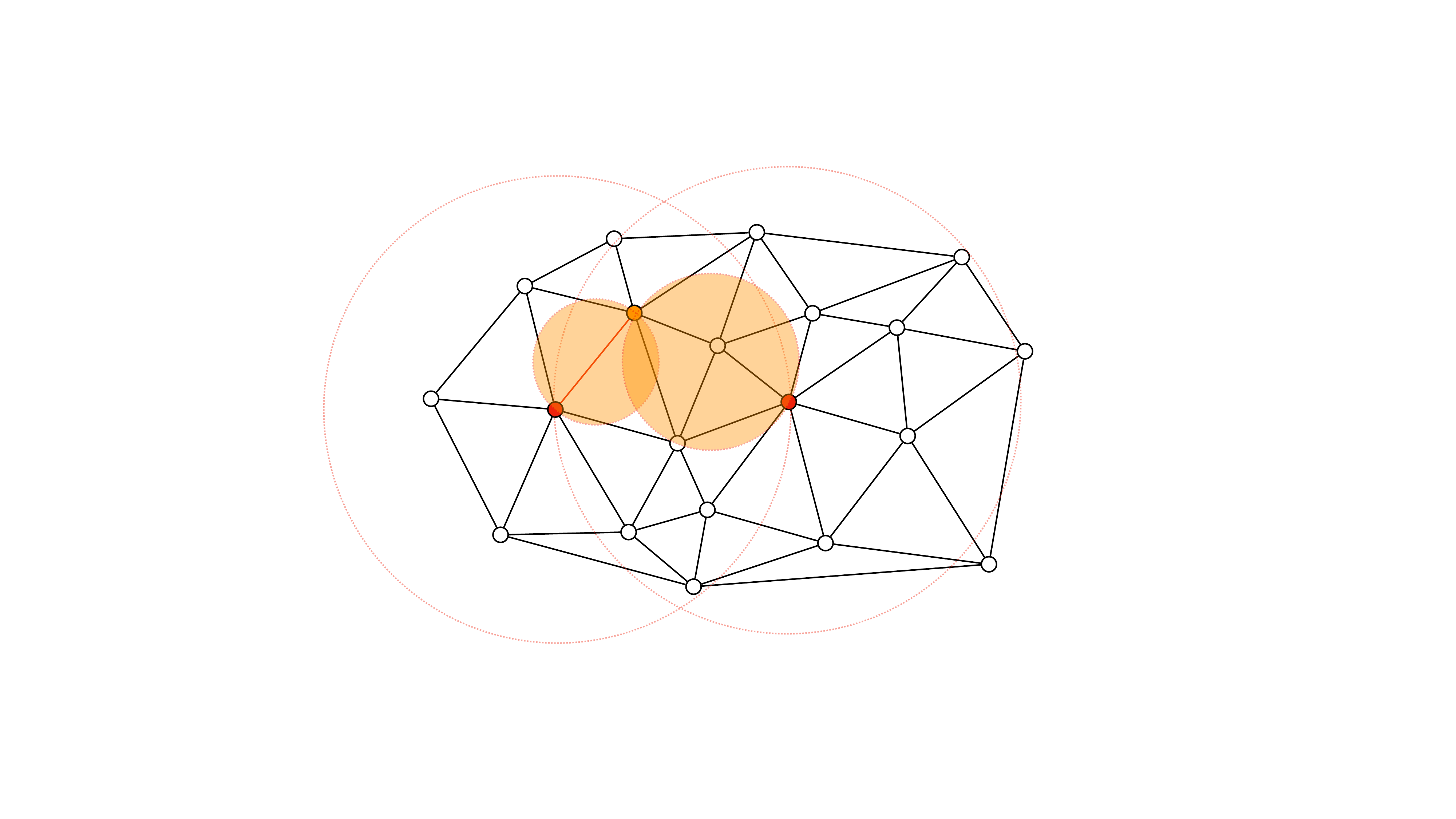}
    (c)\includegraphics[width=\linewidth,trim={5cm 5cm 5cm 5cm},clip]{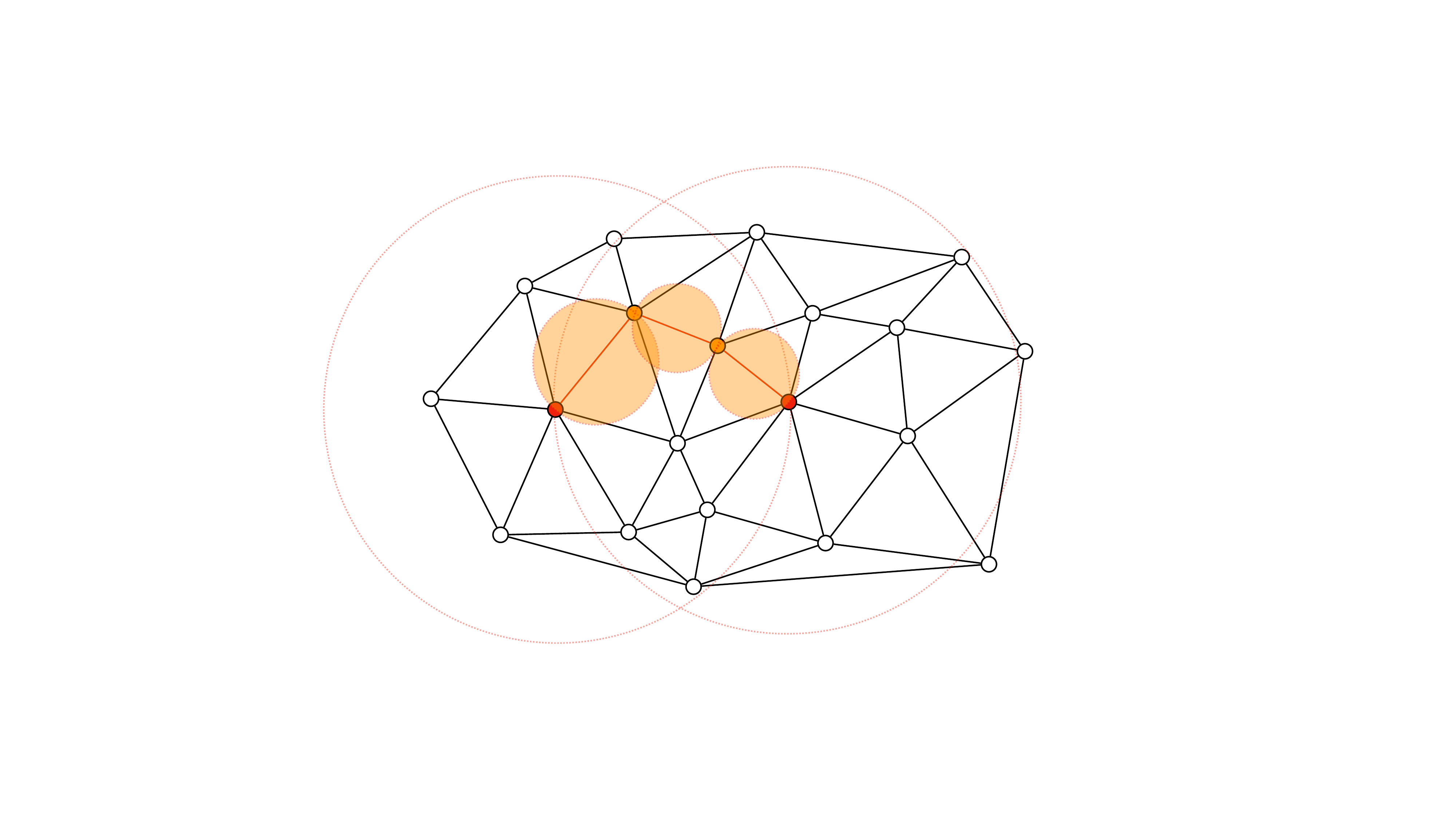}
    \caption{Finding a matching for two points given in red within a Delaunay triangulation. \newline (a)~Consider the circle centered at the midpoint between the points with diameter equal to their distance. Given the interior is non-empty, choose one of the interior points and iterate the process between the two original points and this midpoint. (b) The circle on the left contains no points, thus indicating the points must share an edge in the triangulation. The circle on the right is again non-empty, again break this problem into two and iterate. (c) Two further valid matching are found, completing the search for the overall path between the original vertices.}
    \label{fig:DelaunayMatching}
\end{figure}

Before we discuss the main result of this section, we prove a useful lemma for a non-special system of points (that is one where the Delaunay triangulation is unique).

\begin{lemma}\label{lem:delaunay}
    Given the graph $\mathcal{G} = (\mathcal{V},\mathcal{E})$ corresponding to the Delaunay triangulation of a non-special set of vertices~$\mathcal{V}$ in $D$-dimensional space. Consider two vertices $A$ and $B$ such that there exists a hypersphere whose boundary contain $A$ and $B$ and whose interior contains no other vertices. Then, there must exist an edge $e \in \mathcal{E}$ between $A$ and $B$.
\end{lemma}

\begin{proof}
    If the hypersphere that contains $A$ and $B$ on its boundary with empty interior additionally contains $D-1$ more points on its boundary, then these $D+1$ points form a simplex in the Delaunay triangulation and all share edges, completing the proof. If the boundary contains fewer than $D+1$ points, consider the continuous growth of all possible $D$-dimensional hyperpheres containing all points that are currently on the boundary. Such a growth will be constrained to hyperspheres centered on a restricted hyperplane and can be made continuously until a large enough hypersphere is found that has $D+1$ vertices on its boundary. Such a hypersphere is guaranteed to exist due to the non-special nature of the vertex set and thus we have found the corresponding simplex that contains both $A$ and $B$, and shown that they share an edge.
\end{proof}

Given the above Lemma, we are now ready to prove the main result of the section, which is a reformulation of the graph desiderata from Theorem~\ref{thm:graph_desiderata} in the case of geometrically-local codes.

\begin{theorem}\label{thm:delaunay}
    Given a (unique) Delaunay triangulation of an irreducible logical operator~$\mathcal{L}$ of a geometrically-local code with stabilizer radius bounded by~$R_\mathcal{S}$ the resulting graph~$\mathcal{G} = (\mathcal{V},\mathcal{E})$ will have the following properties:
    \begin{enumerate}
    \setcounter{enumi}{-1}
        \item $\mathcal{G}$ is connected.
        \item $\mathcal{G}$ has vertex degree that is bounded by $O\left(((D+1)R_\mathcal{S})^D\right)$, that is a constant independent of system-size (yet dependent on dimension~$D$).
        \item For all local stabilizer generators of the code~$S_i$ restricted to~$\mathcal{L}$, the resulting matching of associated vertices in~$\mathcal{G}$ is bounded, $|\mu(\mathcal{L}_{S_i})| = O(R_\mathcal{S}^D)$. Each edge can belong to at most~$O(R_\mathcal{S}^{2D})$ matchings.
        \item There is a cycle basis of $G$ such that (a)~each cycle is length $3$ and (b)~each edge is involved in at most $O\left(((D+1)R_\mathcal{S})^D\right)$~cycles.
    \end{enumerate}
\end{theorem}

\begin{proof}
    By definition the Delaunay triangulation is connected, and hence desideratum~0 is satisfied. 
    
    Given we have a logical operator that has to preserve some local structure, there exists a radius~$R_\mathcal{L}$ of constant size, such that any hypersphere whose radius is at least~$R_\mathcal{L}$ must contain $D+1$ points in its interior (or boundary) due to the constant density of points in the logical operator, by definition. Any circum-hypersphere of a simplex in Delaunay triangulation must have radius upper-bounded by~$R_\mathcal{L}=(D+1)R_\mathcal{S}$. A given vertex may only share an edge with another vertex in the graph if they belong to the same circum-hypersphere, and must be within distance~$2R_\mathcal{L}$ of one another. Therefore, the total number of neighbors a given vertex~$v \in \mathcal{V}$ can have is bounded by the number of vertices within the hypersphere of radius $2R_\mathcal{L}$, centered at~$v$, which is at most~$\rho (2R_\mathcal{L})^D$ by the definition of a geometrically-local code, proving desideratum~1. 
    
    We will return to the proof of desideratum~2 below. Showing that desideratum~3 is satisfied is relatively straightforward. Given the Delaunay triangulation forms a graph composed of simplices, the faces of the graph are all triangles. Therefore, a valid choice for the cycle basis of this graph is the set of all triangular faces, which are all of length $3$. Moreover, given each edge by definition is associated with a pair of points (its endpoints), in order for such an edge to be involved in a triangular face, the pair of endpoints must share a common neighbor in the graph. Since the number of neighbors for each individual endpoint is upper-bounded by~$\rho (2R_\mathcal{L})^D$, their common neighbors must be similarly upper-bounded, thus proving desideratum~3.

    Finally, we arrive at the proof of desideratum~2. Given a stabilizer~$S_i$ that anti-commutes with the given logical operator at an even number of qubits, this will correspond to an even number of vertices in the graph~$\mathcal{G}$. We are tasked with finding a matching for these vertices that is of bounded length. Given the stabilizer is local, the resulting set of vertices on the graph~$\mathcal{G}$ are local, that is they are all within a hypersphere of constant radius~$R_\mathcal{S}$. We can choose to match any two pairs of points, as we will show that we can find a path in~$\mathcal{G}$ of bounded length (not necessarily minimum-matching). Given two points $A$ and $B$, consider the hypersphere~$H_{AB}$ centered at the midpoint between these points, whose diameter is the distance between these points (thus $A$ and $B$ reside on the boundary of this hypersphere, and the diameter must be upper bounded by~$R_\mathcal{S}$). If there are no points of~$\mathcal{G}$ in the interior of~$H_{AB}$ then $A$ and $B$ must belong to the same simplex in the Delaunay triangulation (by Lemma~\ref{lem:delaunay}) and we have found a path. If there does exist any points in the interior of $H_{AB}$, choose any such point and denote it~$C$. Now, we can iterate the search for an edge with the new pairings: $AC$ and $BC$, thus breaking the original problem into two equivalent subproblems. Given $C$ is in the interior of $H_{AB}$, the geometric distance between $A$ and $C$ as well as $B$ and $C$ must be strictly less that $A$ and $B$. Given there are only a finite number of vertices within the hypersphere of constant radius~$R_\mathcal{S}$, upper bounded by~$\rho R_\mathcal{S}^D$, the search will terminate in a finite number of steps resulting in a path with at most $O(R_\mathcal{S}^D)$~edges. Finally, we would like to bound the number of times a given edge is used in matchings of different stabilizer generators, restricted to~$\mathcal{L}$. Given all edges in the matching of two points~$A$ and $B$ must be withing the hyperspheres of radius~$R_\mathcal{S}$ of each point, there can only be a finite number of pairings of points for which a given edge can belong. The number of such pairings of points would necessarily be upper-bounded by~$O\left((R_\mathcal{S}^D)^2\right)$.
\end{proof}

In general, in addition to the desiderata 0-3 provided by the above theorem, we would also need desideratum 4, which demands the auxiliary graph has sufficient relative expansion. However, desideratum 4 is not necessary for using the repetition code adapter of Section~\ref{sec:rep_adapter} on at least some geometrically-local codes. For instance, in some $D$-dimensional toric codes \cite{dennis2002topological,bombin2007homological,vasmer2019three,jochym2021four}, the minimum-weight logical $Z$-type operators $\overline{Z}^{(0)},\overline{Z}^{(1)},\dots,\overline{Z}^{(D-1)}$ are 1-dimensional and run in orthogonal directions through the lattice. This means that auxiliary graph surgery on one of these cannot decrease the weight of another in the resulting deformed code.

\section{Adapters for multi-code logic}
\label{sec:multicode}

This section elaborates on the use of the universal adapter as a tool to perform logical measurements between different quantum LDPC codes. 

\begin{figure}[t]
    \centering
    \usetikzlibrary{decorations.pathreplacing}
\pgfmathsetmacro{\ewid}{0.7pt}

\newcommand{\measup}{
\draw[rounded corners=2pt,line width=\ewid] (0,0) rectangle (10,9) {};
\draw[] (5,6) arc (90:150:4pt);
\draw[] (5,6) arc (90:30:4pt);
\draw[-,>=stealth] (5,3) -- +(60:6pt);
}

\newcommand{\measdn}{
\draw[rounded corners=2pt,line width=\ewid] (0,0) rectangle (10,-9) {};
\draw[] (5,-6) arc (270:210:4pt);
\draw[] (5,-6) arc (270:330:4pt);
\draw[-,>=stealth] (5,-3) -- +(300:6pt);
}

\newcommand{\measright}{
\draw[fill=white,rounded corners=2pt,line width=\ewid] (0,0) rectangle (8,10) {};
\draw[] (5,5) arc (0:60:4pt);
\draw[] (5,5) arc (0:-60:4pt);
\draw[-,>=stealth] (3,6) -- +(-30:4.5pt);
}

\begin{tikzpicture}[scale=1.2,x=1pt,y=1pt]

\draw[color=black,line width=\ewid] (55.,-30.) node[left] {$(i)$};

\draw[color=black,line width=\ewid] (0.,15) -- (125,15);
\draw[color=black,line width=\ewid] (0.,15) node[left] {$|\psi\rangle$};

\draw[color=black,line width=\ewid] (0.,0) node[left] {$|+\rangle$};
\draw[color=black,line width=\ewid] (0.,0) -- (70.,0);

\begin{scope}[xshift=6,yshift=-6]
    \begin{scope}[xshift=6,yshift=30-5]
       \measup 
       \draw (12, 12) node {\scriptsize $a$};
    \end{scope}
\draw[fill=white,line width=\ewid] (2, 0) rectangle (20,28);
\draw (11, 15.5) node {$ZZ$};
\end{scope}

\begin{scope}[xshift=38]
  \draw[dashed,line width=\ewid,black!50] (0,31) -- (0,-20);  
\end{scope}
\begin{scope}[xshift=60]
  \draw[dashed,line width=\ewid,black!50] (0,31) -- (0,-20);  
\end{scope}

\begin{scope}[xshift=49]
\draw[fill=white,line width=\ewid] (-6, -6) rectangle (6, 6);
    \draw (0, 0) node {$T$};
\end{scope}

\begin{scope}[xshift=70,yshift=0]
    \begin{scope}[xshift=10,yshift=0]
        \draw[color=black] (0,0.5) -- (6,0.5); 
        \draw[color=black] (0,-0.5) -- (6,-0.5); 
         \draw[color=black] (14,0) node[left] {\scriptsize $b$};
         \begin{scope}[xshift=-6,yshift=-5]
            \measright
        \end{scope}
    \end{scope}
    \draw[fill=white,line width=\ewid] (-6, -6) rectangle (6, 6);
    \draw (0, 0) node {$X$};
\end{scope}

\begin{scope}[xshift=106]
\draw[fill=white,line width=\ewid] (-12, 10) rectangle (12,22);
\draw (0, 15) node {$Z^{\,b}S^{\,a}$};
\end{scope}

\begin{scope}[xshift=126,yshift=0]
\draw (10, 5) node {$=$};
\end{scope}

\begin{scope}[yshift=-70]

\draw[color=black,line width=\ewid] (55.,-30.) node[left] {$(ii)$};
\draw[color=black,line width=\ewid] (0.,15) -- (124,15);
\draw[color=black,line width=\ewid] (0.,15) node[left] {$|\psi\rangle$};

\draw[color=black,line width=\ewid] (0.,0) node[left] {$|T\rangle$};
\draw[color=black,line width=\ewid] (0.,0) -- (50.,0);

\begin{scope}[xshift=6,yshift=-6]
    \begin{scope}[xshift=6,yshift=30-5]
       \measup 
       \draw (12, 12) node {\scriptsize $a$};
    \end{scope}
\draw[fill=white,line width=\ewid] (2, 0) rectangle (20,28);
\draw (11, 15.5) node {$ZZ$};
\end{scope}

\begin{scope}[xshift=38]
  \draw[dashed,line width=\ewid,black!50] (0,31) -- (0,-20);  
\end{scope}

\begin{scope}[xshift=50,yshift=0]
    \begin{scope}[xshift=10,yshift=0]
        \draw[color=black] (0,0.5) -- (6,0.5); 
        \draw[color=black] (0,-0.5) -- (6,-0.5); 
         \draw[color=black] (14,0) node[left] {\scriptsize $b$};
         \begin{scope}[xshift=-6,yshift=-5]
            \measright
        \end{scope}
    \end{scope}
    \draw[fill=white,line width=\ewid] (-6, -6) rectangle (6, 6);
    \draw (0, 0) node {$X$};
\end{scope}

\begin{scope}[xshift=86]
\draw[fill=white,line width=\ewid] (-12, 10) rectangle (12,22);
\draw (0, 15) node {$Z^{\,b}S^{\,a}$};
\end{scope}

\begin{scope}[xshift=104]
\draw[fill=white,color=white] (0, -6.) rectangle (25, 51.);
\draw (10, 5) node {$=$};
\end{scope}


\begin{scope}[xshift=148]
\draw (-20, 15) node {$|\psi\rangle$};
\draw[color=black,line width=\ewid] (-12,15) -- (12,15);
\draw[fill=white,line width=\ewid] (-6, 15-6) rectangle (6, 15+6);
    \draw (0, 15) node {$T$};
\draw (24, 15) node {$T|\psi\rangle$};
\end{scope}
\end{scope}
    
\end{tikzpicture}
    \caption{Logical circuits to connect different quantum codes and implement non-Clifford gate using $(i)$ unitary logical gate implemented in ancillary code $(ii)$ gate teleportation. Based on measurement outcome $b$, the second register can be corrected to be in $\ket{+}$ final state.}\label{fig:logical_circuit_T_injec}
\end{figure}

\subsection{Logical gates by joint-measurements and symmetries}
Since the universal adapter can sparsely connect auxiliary graphs corresponding the logical operators irrespective of code structure or code block, one can now also consider joint measurements in multi-code systems that comprise of different codes  with various capabilities. The repetition code basis change can thus be used to design code deformation protocols to facilitate fault-tolerantly merging into LDPC codes with other desirable logical gates. One simple use case is simply to teleport logical information from one codeblock to another. Ideally, the codes have similar code distance, in order for a high-distance code not to lose distance during code deformation.

For instance, one can use the adapter to merge with a specialized code of comparable distance where the desired gate can be performed, or alternatively teleport a preprepared magic state \cite{Zhou2000gateteleport}, by applying the desired gate on ancilla beforehand.

A simple example is shown in Fig.~\ref{fig:logical_circuit_T_injec} for the case of $\mathsf{T}$ gate (which commutes with $ZZ$ measurement in merge step). The quantum circuit in Fig.~\ref{fig:logical_circuit_T_injec} also requires a potential Clifford correction in the form of an $\mathsf{S}$-gate, which can similarly be implemented using the quantum circuit above, but with the Pauli $Y$-eigenstate $\ket{i}:=\ket{0}+i\ket{1}$ as a resource state instead of $\ket{T}$. 

\begin{figure}[t]
    \centering
    \usetikzlibrary{decorations}
\usetikzlibrary{decorations.pathreplacing,decorations.pathmorphing}

\pgfmathsetmacro{\ewid}{0.7pt}

\newcommand{\measup}{
\draw[rounded corners=2pt,line width=\ewid] (0,0) rectangle (10,9) {};
\draw[] (5,6) arc (90:150:4pt);
\draw[] (5,6) arc (90:30:4pt);
\draw[-,>=stealth] (5,3) -- +(60:6pt);
}

\newcommand{\measdn}{
\draw[rounded corners=2pt,line width=\ewid] (0,0) rectangle (10,-9) {};
\draw[] (5,-6) arc (270:210:4pt);
\draw[] (5,-6) arc (270:330:4pt);
\draw[-,>=stealth] (5,-3) -- +(300:6pt);
}

\newcommand{\measright}{
\draw[fill=white,rounded corners=2pt,line width=\ewid] (0,0) rectangle (8,10) {};
\draw[] (5,5) arc (0:60:4pt);
\draw[] (5,5) arc (0:-60:4pt);
\draw[-,>=stealth] (3,6) -- +(-30:4.5pt);
}

\begin{tikzpicture}[scale=1.2,x=1pt,y=1pt]

\filldraw[color=white] (0., -30) rectangle (180, 70);

\draw[color=black,line width=\ewid] (0.,45.) -- (153,45.);
\draw[color=black,line width=\ewid] (0.,45.) node[left] {$|c\rangle$};

\draw[color=black,line width=\ewid] (0.,30.) node[left] {$|+\rangle$};
\draw[color=black,line width=\ewid] (0.,30.) -- (76.,30.);

\draw[color=black,line width=\ewid] (0.,15.) node[left] {$|0\rangle$};
\draw[color=black,line width=\ewid] (0.,15.) -- (76.,15.);

\draw[color=black,line width=\ewid] (0.,0.) node[left] {$|t\rangle$};
\draw[color=black,line width=\ewid] (0.,0.) -- (154,0.);

\draw[decorate,decoration={calligraphic brace,mirror},line width=0.8pt]
  (-18,30) -- (-18,15) node[midway,xshift=-2.0em]{\begin{tabular}{c} \scriptsize{Ancillary} \\ \scriptsize{code} \end{tabular}};

\begin{scope}[xshift=6,yshift=24]
    \begin{scope}[xshift=6,yshift=30-5]
       \measup 
       \draw (12, 12) node {\scriptsize $a$};
    \end{scope}
\draw[fill=white,line width=\ewid] (2, 0) rectangle (20,28);
\draw (11, 15.5) node {$ZZ$};
\end{scope}

\begin{scope}[xshift=6,yshift=-8]
\begin{scope}[xshift=6,yshift=3]
   \measdn 
    \draw (12, -12) node {\scriptsize $b$};
\end{scope}
\draw[fill=white,line width=\ewid] (2,0) rectangle (20,28);
\draw (11, 15.5) node {$XX$};
\draw (13, -20) node {\footnotesize $(i)$};

\end{scope}

\begin{scope}[xshift=34]
  \draw[dashed,line width=\ewid,black!50] (0,61) -- (0,-28);  
\end{scope}
\begin{scope}[xshift=55]
  \draw[dashed,line width=\ewid,black!50] (0,61) -- (0,-28);  
\end{scope}

\begin{scope}[xshift=45]
\draw (0, -28) node {\footnotesize $(ii)$};
\draw [line width=\ewid-0.15] (0,30.) -- (0,15.);
\filldraw (0, 30.) circle(1.2pt);
\draw[fill=white] (0, 15.) circle(3.pt);
\clip (0, 15.) circle(3.pt);
\draw [line width=\ewid-0.15](-3, 15.) -- (3, 15.); 
\draw [line width=\ewid-0.15](0, 12.) -- (0, 18.); 

\end{scope}

\begin{scope}[xshift=68,yshift=30]
    \draw (0, -58) node {\footnotesize $(iii)$};
    \begin{scope}[xshift=10,yshift=0]
        \draw[color=black] (0,0.5) -- (6,0.5); 
        \draw[color=black] (0,-0.5) -- (6,-0.5); 
         \draw[color=black] (14,0) node[left] {\scriptsize $c$};
         \begin{scope}[xshift=-6,yshift=-5]
            \measright
        \end{scope}
    \end{scope}
    \draw[fill=white,line width=\ewid] (-6, -6) rectangle (6, 6);
    \draw (0, 0) node {$X$};
\end{scope}

\begin{scope}[xshift=68,yshift=15]
    \begin{scope}[xshift=10,yshift=0]
        \draw[color=black] (0,0.5) -- (6,0.5); 
        \draw[color=black] (0,-0.5) -- (6,-0.5); 
         \draw[color=black] (14,0) node[left] {\scriptsize $d$};
         \begin{scope}[xshift=-6,yshift=-5]
            \measright
        \end{scope}
    \end{scope}
    \draw[fill=white,line width=\ewid] (-6, -6) rectangle (6, 6);
    \draw (0, 0) node {$Z$};
\end{scope}

\begin{scope}[xshift=103]
\draw[fill=white,line width=\ewid] (-11, 39) rectangle (11,51);
\draw (0, 45) node {$Z^{\,b+c}$};
\end{scope}

\begin{scope}[xshift=103]
\draw[fill=white,line width=\ewid] (-11, -6) rectangle (11,6);
\draw (0, 0) node {$X^{a+d}$};
\end{scope}

\begin{scope}[xshift=120]
\draw[fill=white,color=white] (0, -6.) rectangle (15, 51.);
\draw (7.5, 22.5) node {$=$};
\end{scope}

\begin{scope}[xshift=144]
\draw [line width=\ewid-0.15](0,45.) -- (0,0.);
\filldraw (0, 45.) circle(1.2pt);
\draw[fill=white] (0, 0.) circle(3.pt);
\clip (0, 0.) circle(3.pt);
\draw [line width=\ewid-0.15](-3, 0.) -- (3, 0.);
\draw [line width=\ewid-0.15](0, -3.) -- (0, 3.);
\end{scope}
    
\end{tikzpicture}
    \caption{Implementing a $\cnotgate$ within multi-code architectures. Logical circuit describing entire protocol: $(i)$ merge step, $(ii)$ logical $\cnotgate$, $(iii)$ split step: measuring out the ancillary code, followed by Pauli corrections $Z^{\, b + c}$ and $X^{a + d}$}\label{fig:logical_circuit_code_deformation}
\end{figure}

A second example, illustrated in Fig.~\ref{fig:logical_circuit_code_deformation}, implements a logical $\cnotgate$ gate in a similar manner. In the second step $(ii)$ of a logical $\overline{\cnotgate}$ is applied on the ancilla code. Notice that one could put any code in the ancilla as long as it has the ability to implement $\overline{\cnotgate}$ and also has at least the same distance $d$ as the original LDPC code. Note also that since $\cnotgate$ commutes with $ZZ$ and $XX$ parity measurements in step $(i)$, one could also apply the $\overline{\cnotgate}$ gate on the ancilla code $\ket{\overline{+0}}$ in advance to obtain a logical Bell pair $\ket{\overline{00}}+\ket{\overline{11}}$ and consume this resource state to perform $\overline{\cnotgate}$ using code deformation. 

Fig.~\ref{fig:logical_circuit_T_injec}(a) and Fig.~\ref{fig:logical_circuit_code_deformation} suggest an interesting perspective. If we imagine the joint measurements in those figures are performed by auxiliary graph surgery, then we briefly enter a deformed code that is a stitching together of the original code containing the logical information on which we want to compute and an ancillary code that possesses the logical gate we need. That deformed code is a sort of Frankenstein's monster code still encoding the logical information from the original code but with an addressable logical gate inherited from the ancillary code. One could remain in this deformed code and perform the gates in Fig.~\ref{fig:logical_circuit_T_injec}(a) and Fig.~\ref{fig:logical_circuit_code_deformation} rather than immediately splitting back into separate original and ancillary codes after the joint measurements.

In the next section, we illustrate the use of the adapter construction to implement an addressable $\cnotgate$ by connecting to a toric code. We note that while more practical schemes to implement $\cnotgate$s exist (including auxiliary graph surgery itself), we provide this ``toric code adapter" tool as an interesting example and a potential stepping stone to future designs that can leverage symmetries in more exotic codes, with different choices of gates we want to implement.

\subsection{Toric code adapter for addressable $\cnotgate$} \label{sec:toric_adapter}

Topological codes such as the toric and hyperbolic codes host a scheme to implement the $\overline{\cnotgate}$ logical map using physical $\cnotgate$s applied pairwise to qubits in the support of logical operators in topological codes, and moreover, the scheme is fault-tolerant and inherently parallelizable. This logical gate scheme makes use of a topological code deformation technique known as Dehn twists \cite{breuckmann2017hyperbolic}. In this section we review the Dehn twist logical $\overline{\cnotgate}$ for topological codes through the lens of our Tanner graph notation, and then present the use of the toric code adapter to implement Dehn twist-like $\overline{\cnotgate}$ \textit{ex-situ} on logical qubits from arbitrary multi-qubit LDPC codeblocks, while ensuring the code distance and sparsity of code is preserved. The auxiliary graphs interfacing between the LDPC code and toric code require only $O(d\log^3d)$ qubits and provide a more space-efficient way to share logical information between an arbitrary LDPC code to the toric code, in comparison with the naive approach of teleporting to a surface code to perform computation, such as previous methods that used $O(d^2)$-sized ancilla to mediate between a topological code and LDPC code \cite{xu2023constantreconfig,viszlai2024matchingneutral}. 


Let $c$ be the control logical qubit and $t$ be the target logical qubit of the original LDPC code between which we want to implement $\overline{\cnotgate}$. Consider their corresponding logical Pauli operators, specifically, $\bar{Z}^{(c)}$ and $\bar{X}^{(t)}$. If neither $\bar{Z}^{(c)}$ nor $\bar{X}^{(t)}$ contain other logical operators fully in their support, i.e.~they are \textit{irreducible}, then we can find alternative representatives of $\bar{Z}^{(c)}$ and $\bar{X}^{(t)}$ that are non-overlapping. This is due to the fact that any overlapping support can be cleaned from one of the two operators by multiplying with stabilizers (see Lemma \ref{lem:supportlemma} in Appendix). The qubits supporting the non-overlapping $\bar{Z}^{(c)}$ and $\bar{X}^{(t)}$ logicals are labeled as disjoint sets $\mathcal{L}_Z$ and $\mathcal{L}_X$.

The additional components of the overall construction are as follows: two auxiliary graphs $\mathcal{G}_{Z}$ and $\mathcal{G}_{X}$ satisfying graph desiderata 0-3 from section~\ref{sec:skiptree}, and the toric code ancilla, initialized suitably in a specific codestate. The toric code is connected to two ports in the auxiliary graphs via adapter edges as shown in Fig~\ref{fig:unitary_adapter_d_layer}. More concretely, the toric code is connected to the auxiliary graph $\mathcal{G}_Z$ through adapter edges labeled by check matrices $T_Z$ and $P_Z$, obtained as output from $\mathsf{SkipTree}(G_Z)$. The toric code is similarly attached to auxiliary graph $\mathcal{G}_{X}$ at $\lx$ through adapter edges labeled by check matrices $T_X$ and $P_X$, obtained as output from $\mathsf{SkipTree}(G_X)$. 
We emphasize that a special feature in this construction is that the toric code can be directly connected to the auxiliary graphs of the base LDPC code with $\mathsf{SkipTree}$, since it already has the repetition code-like overlap between a logical operator and stabilizers of the opposite Pauli logical.

\begin{figure*}[t]
    \centering
    \includegraphics[width=\linewidth]{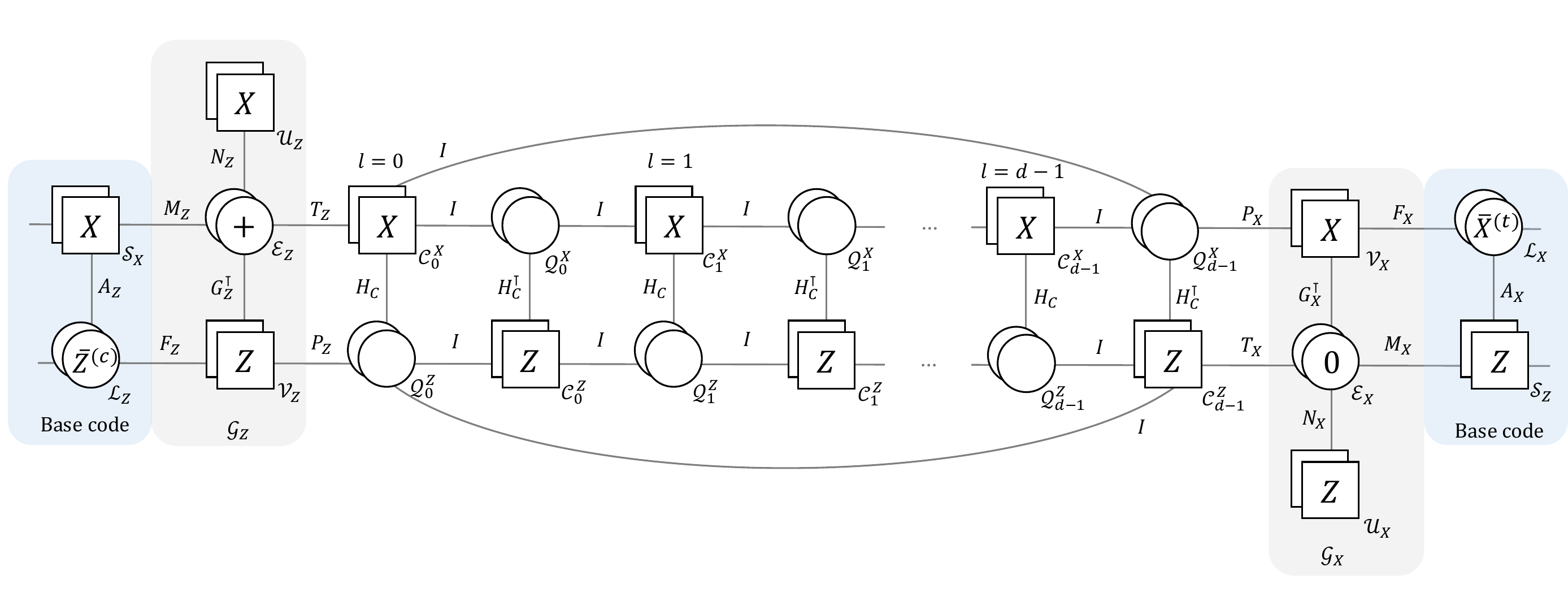}
    \caption{A description of the toric adapter: auxiliary graphs $\mathcal{G}_{Z}$ and $\mathcal{G}_{X}$ are built on the support of logicals $\bar{Z}^{(c)}$ and $\bar{X}^{(t)}$ in the original LDPC code (shown in blue). The distance-$d$ toric code is connected by adapter edges $T_Z$ and $P_Z$ to auxiliary graph $\mathcal{G}_{Z}$ and by adapter edges $T_X$ and $P_X$ to auxiliary graph $\mathcal{G}_{X}$. 
    }
    \label{fig:unitary_adapter_d_layer}
\end{figure*}

Following the convention introduced in section~\ref{subsec:notation_tanner_graph}, let us recap the toric code in terms of its compact Tanner graph. In the rest of this section, $i \in [d]$ where $ d\! \in \!\mathbb{N}$ denotes an element of the finite group with addition modulo $d$. It is implicitly assumed that operations on these elements are carried out modulo $d$, for example $i+1$ refers to $i+1\!\pmod{d}$.

\begin{definition}\label{defn-toric-ancilla} 
    The toric code is described by a Tanner graph with vertex sets labeled $\{\qz_i,\qx_i,\cz_i,\cx_i\}$ for $i \in [d]$ and edges between them defined as follows.

    \begin{itemize}
        \item Edges $(\qz_i,\cz_i)$ and edges $(\qx_i,\cx_i)$ are each labeled by parity check matrix $I$ for each $i \in [d]$, where $I \in \mathbb{F}_2^{d \times d}$ is the identity matrix.

        \item Overall, one can view this Tanner graph as consisting of $2d$ layers: vertex sets $\qz_i$ and $\cx_i$ comprise the `primal' layers, vertex sets $\qx_i$ and $\cz_i$ comprise `dual' layers, indexed by $i$. The primal and dual layers are connected by the identity check matrix $I \in \mathbb{F}_2^{d \times d}$.

        \item Tanner graph edge set $(\qz_i,\cx_i)$ in each primal layer $i$ is described by the canonical check matrix of the $d$-bit cyclic repetition code $H_C= I \,+\, C\in\mathbb{F}_2^{d\times d}$, where $C$ is the cyclic shift matrix $\ket{j}\bra{j+1}$. 
        \item Edge set $(\qx_i,\cz_i)$ in each dual layer $i$ is described by its inverse, $H_C^{\top}$.
    \end{itemize}
\end{definition}

    A compact diagram of the above can be found in Fig~\ref{fig:Tanner_graph_examples}(c). More explicitly, each layer contains $2d$ qubits: $d$ qubits each in sets labeled as $\qx_{i}$ and $\qz_{i}$. The check sets $\cz_i$ and $\cx_i$ each contain $d$ stabilizers, of $X$-type and $Z$-type respectively. 

    The compact Tanner graph can map to the familiar toric code lattice as follows: qubits from dual layers inhabit horizontal edges of the lattice while qubits from primal layers inhabit vertical edges. $X$-stabilizers in sets $\cx_{i}$ represent vertices in the lattice. $Z$ stabilizers in sets $\cz_{i}$ represent faces of the lattice.

\subsubsection{Outline of protocol using toric code adapter}

The circuit describing the entire protocol is the same as in Fig.~\ref{fig:logical_circuit_code_deformation} to implement $\cnotgate$ in a multi-code architecture. 
Begin with two logical qubits of the LDPC code, the control $\ket{c}$ and target $\ket{t}$. Assume access to a prepared toric ancilla logical state, initialized in the codestate $\ket{\overline{+0}}$.

\begin{enumerate}
    \item The first step $(i)$ is to merge fault-tolerantly by adding ancilla qubits and ancilla code initialized suitably and measuring the stabilizers of the new deformed code, described by Tanner graph in Fig.~\ref{fig:unitary_adapter_d_layer}. The number of logical qubits in the merged code is the same as the original LDPC code (shown in Lemma~\ref{lem:gmerge-k-logical-qubits}), and the distance of the original LDPC code is preserved (Theorem~\ref{thm:dist-preserve-static-merge}). 

    \item In the second step $(ii)$, a logical $\overline{\cnotgate}$ applied in the toric code. Section~\ref{subsec:code-evolution} describes how to implement this $\overline{\cnotgate}$ on the toric code, introducing and making use of Dehn twists.

    \item The third step $(iii)$ is to measure out the toric code ancilla. This unentangles the toric code from the original LDPC code. 
    
    \item Apply logical Pauli corrections based on the measurement outcomes.
    One could also apply logical Pauli corrections to obtain the $+1$ eigenstates of $\bar{X}^{(1)}$ and $\bar{Z}^{(2)}$ operators in the toric code, which resets the toric code ancilla to $\ket{\overline{+0}}$, the state we began with. 
    This step can be done transversally.
\end{enumerate}

This circuit implements $\overline{\cnotgate}$ between the control and target on the LDPC code, as will be discussed in the following section.

\subsubsection{Parameters of the deformed code}\label{subsubsec:static-gmerge-param}

When deforming the original code by measuring stabilizers of the new code, we want to make sure that no logical qubits get measured in the process, and also check if any gauge degrees of freedom are introduced. 
 
\begin{lemma}\label{lem:gmerge-k-logical-qubits}
    If the original LDPC code encodes $k$ qubits, then the merged code encodes $k$ qubits.
\end{lemma}

\begin{proof}
One can calculate the dimension of the logical space of the new code by calculating the rank of stabilizers and invoking the rank-nullity theorem. A detailed calculation is included in Appendix \ref{app:toric_adapter_kqubits}. 
\end{proof}

Once we have measured into the deformed code, the $\bar{Z}^{(c)}$ and $\bar{X}^{(t)}$ logicals from the original LDPC code have additional logical representatives which fully lie within the toric code ancilla. This can be seen as follows. The entire logical $\bar{Z}^{(c)}$ can be entirely cleaned from $\lz$ qubit by qubit and simultaneously moved to $\qz_0$, as below (shown in Fig~\ref{fig:g_int_Z}),
    \begin{align}\label{eqn:zrep_on_qz0}
        Z(\vec{1} \in {\lz})\; \mathcal{H}_Z(\vec{1} \in \vz) = Z(\vec{1} \in \qz_0)
    \end{align}

Notice that since the logical is supported on exactly the same number of qubits as before, the distance of $\bar{Z}^{(c)}$ is preserved under multiplication by $\mathcal{H}_Z(\vec{1} \in \vz)$. 

We can further move the support of $\bar{Z}^{(c)}$ from $\qz_0$ to $\qz_1$ by multiplying with all checks from $\cz_0$ to get $Z(\qz_0)\; \mathcal{H}_Z(\vec{1} \in \cz_0) = Z(\qz_1)$. This is because $\vec{1}H_{C}^\top=0$, since each of the cyclic repetition code checks is weight $2$. In fact, the $\bar{Z}^{(c)}$ logical representative can be moved to any of the $d$ primal layers while preserving the weight of $\bar{Z}^{(c)}$,
\begin{align}\label{eqn:zrep_on_qzi}
        Z(\vec{1} \in \qz_i)\; \mathcal{H}_Z(\vec{1} \in \cz_i) = Z(\vec{1} \in \qz_{i+1 (\textrm{mod}\, d)})
    \end{align}

\begin{figure}[t]
\label{fig:g_int_Z}
\subfloat[]{
\hspace{0.2in}
\begin{tikzpicture}[]
\pgfmathsetmacro{\ewid}{1pt}
\pgfmathsetmacro{\drwid}{0.8pt}
\pgfmathsetmacro{\sqrsz}{0.7cm}
\pgfmathsetmacro{\sqrdf}{0.12}
\pgfmathsetmacro{\crcsz}{0.42}
\pgfmathsetmacro{\crcdf}{0.09}
\pgfmathsetmacro{\basecodewid}{1.1}
\pgfmathsetmacro{\brnd}{0.02cm}
\definecolor{bl}{rgb}{0.63, 0.79, 0.95}
\definecolor{gry}{rgb}{0.75, 0.75, 0.75}


 \fill[bl, opacity=0.3]
        (-1.3, -0.2) arc[start angle=180, end angle=270, x radius=\brnd, y radius=\brnd] -- ++(\basecodewid, 0) arc[start angle=270, end angle=360, x radius=\brnd, y radius=\brnd] -- ++(0, 2.6) arc[start angle=0, end angle=90, x radius=\brnd, y radius=\brnd] -- ++(-\basecodewid, 0) arc[start angle=90, end angle=180, x radius=\brnd, y radius=\brnd]-- cycle;

        \draw[draw=gray,line width=\ewid] (-1,0) to (5.6,0);
        \draw[draw=gray,line width=\ewid] (-1,2) to (5.6,2);
        \draw[draw=gray,line width=\ewid] (0,0) to (0,2);
        
        
        \draw[fill=gry!30,line width=\drwid] (0-\crcdf,0+\crcdf) circle (\crcsz);
        \draw[fill=gry!30,line width=\drwid] (0,0) circle (\crcsz) node {$\bar{Z}^{(c)}$};
        \node[draw, line width=\drwid, fill=white, minimum size=\sqrsz] at (0-\sqrdf, 2+\sqrdf) {};
        \node[draw, line width=\drwid, fill=white,minimum size=\sqrsz] at (0, 2) {$X$};

        \node[] at (-0.4, 1) {{\footnotesize $A_Z$}};
        
        \node[] at (1.2, 0.2) {{\scriptsize $F_Z$}};
        \node[] at (1.3, 2.2) {{\scriptsize $M_Z$}};

        \node[] at (0.5, -0.4) {{ $\mathcal{L}_Z$}};
        \node[] at (0.7, 1.65) {{$\mathcal{S}_{X}$}};

        \begin{scope}[shift={(2.4,0)}]
        \draw[draw=gray,line width=\ewid] (0,0) to (0,4);
        \draw[fill=white,line width=\drwid] (0-\crcdf, 2+\crcdf) circle (\crcsz);
        \draw[fill=white,line width=\drwid] (0,2) circle (\crcsz) node {};
         \node[draw, line width=\drwid,fill=white, minimum size=\sqrsz] at (0-\sqrdf,0+\sqrdf) {};
        \node[draw, line width=\drwid,fill=white,minimum size=\sqrsz] at (0, 0) {$Z$};
        \node[draw, line width=\drwid,fill=white, minimum size=\sqrsz] at (0-\sqrdf,4+\sqrdf) {};
        \node[draw, line width=\drwid,fill=white,minimum size=\sqrsz] at (0, 4) {$X$};
        \node[] at (-0.25, 1) {{\footnotesize $G^{\top}_Z$}};
        \node[] at (1.2, 2.2) {{\scriptsize $T_Z$}};
        \node[] at (1.2, 0.2) {{\scriptsize $P_Z$}};
        \node[] at (-0.25, 3) {{\footnotesize $N_Z$}};
        \node[] at (0.5, 1.6) {{ $\mathcal{E}_{Z}$}};
        \node[] at (0.65, -0.4) {{ $\mathcal{V}_{Z}$}};
        \node[] at (0.65, 3.6) {{ $\mathcal{U}_{Z}$}};
        \end{scope}

        \begin{scope}[shift={(4.8,0)}]
        \draw[draw=gray,line width=\ewid] (0,0) to (0,2);
        \node[draw, line width=\drwid,fill=white, minimum size=\sqrsz] at (0-\sqrdf, 2+\sqrdf) {};
        \node[draw, line width=\drwid,fill=white,minimum size=\sqrsz] at (0, 2) {$X$};
        \draw[fill=gry!30,line width=\drwid] (0-\crcdf,0+\crcdf) circle (\crcsz);
        \draw[fill=gry!30,line width=\drwid] (0,0) circle (\crcsz) node {};
        \node[] at (-0.25, 1) {{\footnotesize $H_C$}};
        \node[] at (0.65, 1.5) {{ $\mathcal{C}^{X}_0$}};
        \node[] at (0.55, -0.45) {{ $\mathcal{Q}^{Z}_0$}};
        \node[] at (1.2, 0) {{ $...$}};
        \node[] at (1.2, 2) {{ $...$}};
        \end{scope}

\end{tikzpicture}
}
\newline
\subfloat[]{
\begin{tikzpicture}
\pgfmathsetmacro{\ewid}{1pt}
\pgfmathsetmacro{\drwid}{0.8pt}
\pgfmathsetmacro{\sqrsz}{0.7cm}
\pgfmathsetmacro{\sqrdf}{0.12}
\pgfmathsetmacro{\crcsz}{0.42}
\pgfmathsetmacro{\crcdf}{0.09}
\pgfmathsetmacro{\basecodewid}{0.8}
\pgfmathsetmacro{\brnd}{0.02cm}
\definecolor{bl}{rgb}{0.63, 0.79, 0.95}
\definecolor{clr}{rgb}{0.47, 0.62, 0.8}
\definecolor{gry}{rgb}{0.75, 0.75, 0.75}


\begin{scope}
\node[] at (-1.4, 0) {{ $...$}};
\node[] at (-1.4, 2) {{ $...$}};

 \fill[bl, opacity=0.3]
        (4.1, -0.2) arc[start angle=180, end angle=270, x radius=\brnd, y radius=\brnd] -- ++(\basecodewid, 0) arc[start angle=270, end angle=360, x radius=\brnd, y radius=\brnd] -- ++(0, 2.6) arc[start angle=0, end angle=90, x radius=\brnd, y radius=\brnd] -- ++(-\basecodewid, 0) arc[start angle=90, end angle=180, x radius=\brnd, y radius=\brnd]-- cycle;

        \begin{scope}
        \draw[draw=gray,line width=\ewid] (-1,0) to (5.6,0);
        \draw[draw=gray,line width=\ewid] (-1,2) to (5.6,2);
        \draw[draw=gray,line width=\ewid] (0,0) to (0,2);
        
        
        \draw[fill=gry!30,line width=\drwid] (0-\crcdf,2+\crcdf) circle (\crcsz);
        \draw[fill=gry!30,line width=\drwid] (0,2) circle (\crcsz) node {};
        \node[draw, line width=\drwid, fill=white, minimum size=\sqrsz] at (0-\sqrdf, 0+\sqrdf) {};
        \node[draw, line width=\drwid, fill=white,minimum size=\sqrsz] at (0, 0) {$Z$};

        \node[] at (-0.3, 1) {{\footnotesize $H_C^{\top}$}};
        
        \node[] at (1.2, 0.2) {{\scriptsize $T_X$}};
        \node[] at (1.3, 2.2) {{\scriptsize $P_X$}};

        \node[] at (0.72, -0.45) {{ $\mathcal{C}^{Z}_{d-1}$}};
        \node[] at (0.72, 1.55) {{$\mathcal{Q}^{X}_{d-1}$}};
        \end{scope}
        
        \begin{scope}[shift={(2.4,0)}]
        \draw[draw=gray,line width=\ewid] (0,2) to (0,-2);
        \draw[fill=white,line width=\drwid] (0-\crcdf, 0+\crcdf) circle (\crcsz);
        \draw[fill=white,line width=\drwid] (0,0) circle (\crcsz) node {};
         \node[draw, line width=\drwid,fill=white, minimum size=\sqrsz] at (0-\sqrdf,2+\sqrdf) {};
        \node[draw, line width=\drwid,fill=white,minimum size=\sqrsz] at (0, 2) {$X$};
        \node[draw, line width=\drwid,fill=white, minimum size=\sqrsz] at (0-\sqrdf,-2+\sqrdf) {};
        \node[draw, line width=\drwid,fill=white,minimum size=\sqrsz] at (0, -2) {$Z$};
        \node[] at (-0.25, 1) {{\footnotesize $G^{\top}_X$}};
        \node[] at (1.2, 2.2) {{\scriptsize $F_X$}};
        \node[] at (1.2, 0.2) {{\scriptsize $M_X$}};
        \node[] at (-0.25, -2+1) {{\footnotesize $N_X$}};
        \node[] at (0.6, -0.4) {{ $\mathcal{E}_{X}$}};
        \node[] at (0.65, 1.6) {{ $\mathcal{V}_{X}$}};
        \node[] at (0.65, -2-0.4) {{ $\mathcal{U}_{X}$}};
        \end{scope}

        \begin{scope}[shift={(4.8,0)}]
        \draw[draw=gray,line width=\ewid] (0,0) to (0,2);
        \node[draw, line width=\drwid,fill=white, minimum size=\sqrsz] at (0-\sqrdf, 0+\sqrdf) {};
        \node[draw, line width=\drwid,fill=white,minimum size=\sqrsz] at (0, 0) {$Z$};
        \draw[fill=gry!30,line width=\drwid] (0-\crcdf,2+\crcdf) circle (\crcsz);
        \draw[fill=gry!30,line width=\drwid] (0,2) circle (\crcsz) node {$\bar{X}^{(t)}$};
        \node[] at (-0.25, 1) {{\footnotesize $A_X$}};
        \node[] at (0.65, -0.45) {{$\mathcal{S}_{Z}$}};
        \node[] at (0.55, 1.5){{ $\mathcal{L}_X$}};
        \end{scope}
\end{scope}
\end{tikzpicture}
}
     \caption{(a) $\bar{Z}^{(c)}$ has equivalent representatives $Z(\lz)$ and $Z(\qz_0)$, shaded in gray, which are equivalent up to a product of $Z$ vertex checks $\mathcal{V}_{Z}$.
     (b) Likewise, $\bar{X}^{(t)}$ has equivalent representatives $X(\lx)$ and $X(\qx_{d-1})$, up to a product of $X$ vertex checks $\mathcal{V}_{X}$.}
\end{figure}

Proceeding analogously as for $\bar{Z}^{(c)}$, the logical $\bar{X}^{(t)}$ can be moved entirely to qubit set $\qx_{d-1}$ using the product of all vertex checks in $\vx$,
    \begin{align}\label{eqn:xrep_on_qxd1}
        X(\vec{1} \in {\lx})\; \mathcal{H}_X(\vec{1} \in \vx) = X(\vec{1} \in \qx_{d-1})
    \end{align}
  Similarly $X$-checks from any of the sets $\cx_{i}$ for $i \in [d]$ would simply move the entire logical operator between the dual layers. Thus code deformation provides an avenue for effectively `moving' logical information from the original LDPC code to the ancilla code by entangling the two systems.


In order to retain the error-correcting ability of the original quantum LDPC code, it is desirable that the code distance is preserved during code deformation.

\subsubsection{Toric code $\cnotgate$ using a Dehn twist
}\label{subsec:code-evolution}

Once the joint measurement value has been inferred, the stabilizers of the original LDPC code and the toric code are measured. This is the ``split" step between codes, though the logical information remains entangled. Next, following the logical circuit in Fig.~\ref{fig:logical_circuit_code_deformation}, the $\cnotgate$ is implemented on the toric code. This can be done with a depth-one physical CNOT circuit, followed by a qubit permutation, which together implement a Dehn twist \cite{breuckmann2017hyperbolic}.

Let us describe this Dehn twist circuit in our notation. Transversal physical $\cnotgate$ gates are applied within each layer of the toric code, between qubits in sets $\qz_{i}$ and $\qx_{i}$, as shown in Figure~\ref{fig:cyclic_permutations_codespace}. In each layer $i$, the transversal $\cnotgate$s act between corresponding qubits sharing the same index $j$.
Let gates between two qubit sets in layer $i$ be described as ${\cnotgate}(\qz_i,\qx_i)$ where we use the shorthand
${\cnotgate}(\mathcal{A},\mathcal{B})$ to denote transversal physical $\cnotgate$ gates between qubits in ordered sets $\mathcal{A} = \{a_j\}$ and $\mathcal{B} = \{b_j\}$ for $j \in [d]$, 
\begin{equation}\label{eqn:shorthand-cnot-notation}
    {\cnotgate}(\mathcal{A},\mathcal{B}) = \prod_{j=0}^{d-1} {\cnotgate}_{a_j b_j}
\end{equation}

\begin{figure*}[t]
    \centering
    \includegraphics[width=\linewidth]{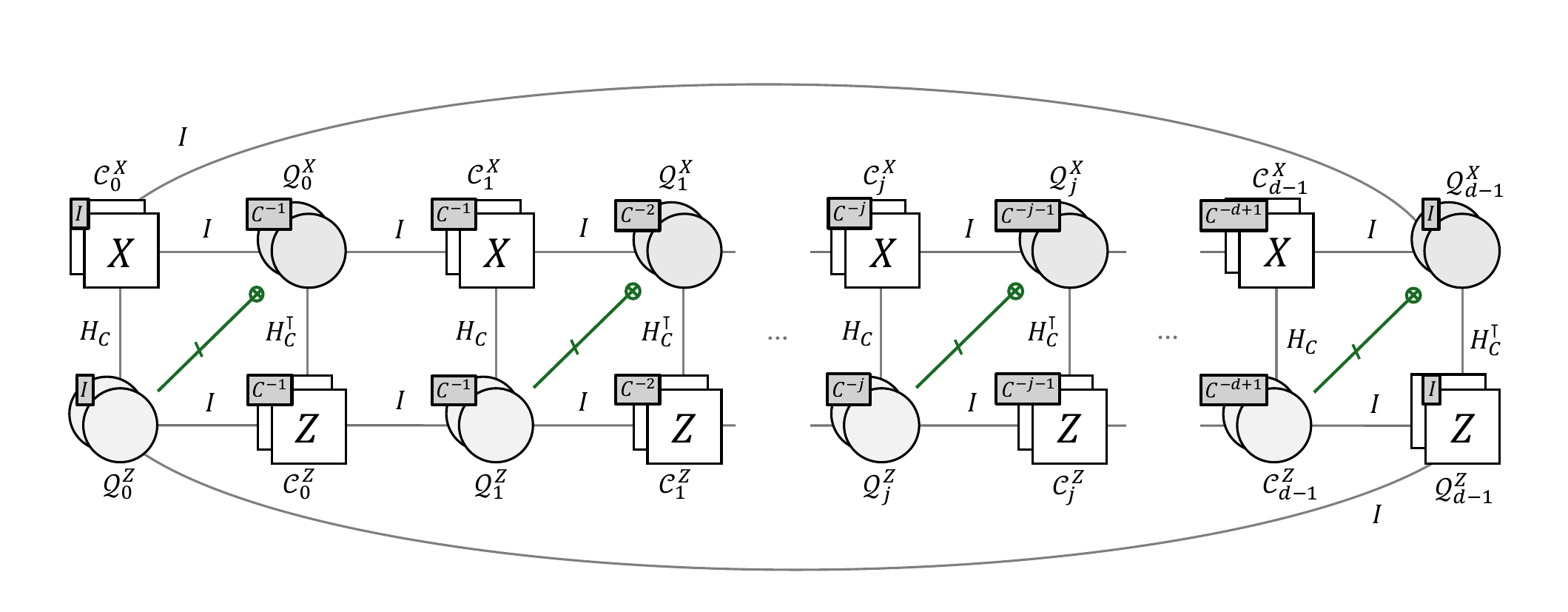}
    \caption{Transversal physical $\cnotgate$s are applied with qubits in sets $\qz_{i}$ as control qubits and $\qx_{i}$ as target qubits, in each layer $i$. Cyclic permutations on qubits and checks restore the system back to the original codespace. Here $C^{-1}$ refers to the inverse cyclic shift matrix by 1 with entries $\ket{r+1}\bra{r}\!{\pmod d}$.}
    \label{fig:cyclic_permutations_codespace}
\end{figure*}

\noindent The complete physical CNOT circuit can be written as a product of transversal gates from all layers $i \in [d]$,
\begin{equation}\label{eqn-dehn-twist-circuit}
  \prod_{i=0}^{d-1}{\cnotgate}(\qz_i,\qx_i)  
\end{equation}

\noindent Since all the qubit sets $\qz_{i}$ and $\qx_{i}$ are disjoint, these physical $\cnotgate$ gates can be implemented in parallel. Note that each set of transversal $\cnotgate$s effectively takes place within a single layer of the toric code.

The stabilizer evolution under the unitary circuit is described in Appendix~\ref{app:dehntwist_unitary}. The key takeaway is that following the unitary circuit (Eq.~\ref{eqn-dehn-twist-circuit}) with a permutation relabeling of the stabilizers and qubits, as shown in Fig.~\ref{fig:cyclic_permutations_codespace} preserves the codespace of the toric code.

In the toric code, permutations on qubits and checks have a geometric interpretation - the plaquettes ($Z$ stabilizers) and qubits on horizontal edges of the plaquette, are shifted by one unit cell with respect to the previous layer, along a nontrivial loop in the dual space defined by qubits supported on $\qx_0$.
The combined operation of transversal $\cnotgate$ gates followed by permutations determined by $\{\sigma_i,\pi_i\}$ preserve the toric code codespace and hence belong to the normalizer group of stabilizers of the toric code. In the following section we prove that this executes the logical $\overline{\cnotgate}$ map. Existing literature \cite{breuckmann2017hyperbolic,koenig2010turaev,lavasani2019dehn} refers to these as gates via Dehn twists. Dehn twists are useful to move logical qubits in storage and implement a logical two-qubit Clifford by twisting a toric code along a boundary and applying physical gates acting on pairs of qubits over increasing distances. Dehn twists are parallelizable although at the cost of a similar overhead to topological codes. 



We summarize this Dehn twist construction in the following lemma, which we prove in Appendix~\ref{subsec:dehntwist_toric} for completeness.
\begin{lemma}[$\overline{\cnotgate}$ Logical action via Dehn Twists]\label{thm:dehntwist_toric_action}
Consider the toric code, a $[[2d^2,2,d]]$ CSS code with Tanner graph described as in Definition \ref{defn-toric-ancilla}. A $\overline{\cnotgate}_{12}$ gate can be performed between its two logical qubits (w.l.o.g.$\!$ $1$ is chosen to be the control qubit) using the following circuit
\begin{align}\label{eqn:toric_dehntwist_cnot}
    \overline{\cnotgate}_{12} =  \; \prod_{i=0}^{d-1} \, 
    \mathsf{\Gamma}_i^{\prime}(\qz_{i}) \;
    \mathsf{\Gamma}_i(\qx_{i}) \; {\cnotgate}(\mathcal{Q}^{Z}_{i},\mathcal{Q}^{X}_{i})
\end{align}
where $\mathsf{\Gamma}_i(\mathcal{A})$ is a unitary corresponding to permutations by $C^{-i-1}$ on qubits $e_j \in \mathcal{A}$, and $\mathsf{\Gamma}_i^{\prime}(\mathcal{A})$ is a unitary corresponding to permutations by $C^{-i}$ on qubits $e_j \in \mathcal{A}$, as shown in Fig.~\ref{fig:cyclic_permutations_codespace}
and ${\cnotgate}(\mathcal{A},\mathcal{B})$ is shorthand for transversal $\cnotgate$ gates between qubits in ordered sets $\mathcal{A}$ and $\mathcal{B}$, $|\mathcal{A}|=|\mathcal{B}|$, defined as in Eq.~\ref{eqn:shorthand-cnot-notation}.
\end{lemma}

This useful result from topological code literature can be extended efficiently to an arbitrary LDPC code using universal adapter based code-switching, in Appendix~\ref{app:toric_adapter_cnot_exsitu}, Theorem~\ref{thm:ldpc-dehntwist-cnot}.

Further, we show that this deformation preserves the code distance of the original quantum LDPC code,
\begin{theorem}\label{thm:dist-preserve-static-merge}
    The merged code has distance $d$ if the original code had distance $d$, provided auxiliary graphs $\mathcal{G}_Z$ and $\mathcal{G}_X$ satisfy the graph desiderata 0-3 in Theorem~\ref{thm:graph_desiderata}.   
\end{theorem}
\begin{proof}
    Proof in Appendix~\ref{app:toric_adapter_distance}.
\end{proof}

\section{Examples}\label{sec:examples}

 To demonstrate the flexibility and low additional space overhead of the joint logical measurement scheme, as well as understand the constant factors in a non-asymptotic setting, let us consider measuring the parity of logical operators in two different high-rate quantum codes: a~$[[98,6,12]]$ bivariate bicyclic code~\cite{bravyi2024highthreshold}, and a lifted product code~\cite{panteleev2022almostlinear} with parameters $[[200,20,10]]$. 

Auxiliary graphs to measure individual logicals are constructed in (Section~\ref{subsec-eg-individ-auxiliary-graphs}) and then connected with universal adapters for multi-code surgery (Section~\ref{subsec-eg-multi-code-BB-LP-adapter}) as well as for logical measurements within the same high-rate codeblock (Section~\ref{subsec-eg-intra-code-BB-BB-adapter}). Both situations present their own challenges: in the former case, different logicals operators can have very different structure relative to the underlying stabilizer group, making connecting the auxiliary graphs non-trivial, despite the convenience of having disjoint logical support. In the latter case, logical operators could potentially overlap on a subset of physical qubits. This requires generalizing beyond an injective port function $f$ from the logical support to construct auxiliary graphs. 

Since the focus is on code deformation techniques rather than peculiarities of specific codes, we keep the description of the code families below concise.

The \textbf{lifted product code} is obtained from two base protographs, which are matrices over the quotient polynomial ring $R_{\ell} = \mathbb{F}_2[x] /\left(x^{\ell}-1\right)$, where $\ell$ is a parameter choice called the \textit{lift} size. For a polynomial $g(x) = g_{0} + g_{1}x + ... + g_{\ell-1}x^{\ell-1} \in R_{\ell}$, define $g(x)^{T}$ as $g_{0} + g_{\ell-1}x + ... + g_{1}x^{\ell-1} \in R_{\ell}$. The lift of $g(x)$ is an $\ell \times \ell$ circulant matrix $\mathbb{B}(g(x))$ where the first column is given by the coefficients of $g(x)$, and $i$-th column is given by the coefficients of $x^{i-1}g(x)$. The lift $\mathbb{B}(A)$ of a matrix $A$ over $R_{\ell}$ is obtained by replacing each element of $A$ by its lift. Given two matrices $A_1$ and $A_2$, of size $m_1 \times n_1$ and $m_2 \times n_2$ respectively with entries in $R_{\ell}$, the lifted product code $\operatorname{LP}\left(A_1, A_2\right)$ is the CSS code with stabilizer matrices

\begin{align}
H_X & =\mathbb{B}\left(\left[\begin{array}{ll}
A_1 \otimes I_{n_2} & I_{m_1} \otimes A_2^T
\end{array}\right]\right) \label{eqn:defn-lp-code-hx} \\
H_Z & =\mathbb{B}\left(\left[\begin{array}{ll}
I_{n_1} \otimes A_2 & A_1^T \otimes I_{m_2}
\end{array}\right]\right) \label{eqn:defn-lp-code-hz}
\end{align}
The block length of the above code is $\ell\left(n_1 m_2+n_2 m_1\right)$. When $\ell=1$, the LP code specializes to the hypergraph product code~\cite{tillich2013quantum}.

To construct a small example, we consider a lifted product code where both protographs in the product are related as $A_1 = A_2 = A$, so the construction takes a single base matrix $A$ and lift $\ell$ as input. Through a search in parameter space, we obtained a $[[200,20,10]]$ lifted product code based on the following lift parameter and base matrix:

\begin{equation}
    \begin{array}{|c|c|c|}
\hline \text { LP Code } & \ell & \text { Matrix over } R_{\ell} \\
\hline \hline \mathrm{LP} & 8 & \left(\begin{array}{cccc}
x^2 & 1 & 1 & x^2 \\
1 & x & x^2 & x \\
x^2 & x & x^3 & x^2
\end{array}\right) \\
\hline
\end{array}
\end{equation}
In the canonical stabilizer basis (Eqs.~\ref{eqn:defn-lp-code-hx}, \ref{eqn:defn-lp-code-hz}), all the stabilizers of both $X$ and $Z$ type each have weight exactly $7$. The total qubit check degrees are either $6$ or $8$.

Complementarily to the above product-based construction, the \textbf{bivariate bicyclic code}~\cite{bravyi2024highthreshold,kovalev2013hyperbicycle} check matrices are defined in terms of two variables $x$ and $y$, and can be directly constructed as

\begin{equation}\label{eqn:defn-bb-code}
H_X=[A \mid B], \quad H_Z=\left[B^{\top} \mid A^{\top}\right],    
\end{equation}
where matrices $A, B \in \mathbb{F}_2[x, y]$ are polynomials in the matrices $x$ and $y$,
\begin{equation}\label{eqn-bb-defn-x-and-y-circulants}
    x=S_{\ell} \otimes I_m \quad \text { and } \quad y=I_{\ell} \otimes S_m
\end{equation}
and $I_{r}$ is the $r\times r$ identity matrix,  $S_{r}$ is the $r \times r$ cyclic shift matrix, i.e.~$\bra{i}S_{r}=\bra{i+1\;\text{mod}\;r}$. Note that $x^l = y^m = I_{lm}$.

From above, one can see $AB = BA$, since $xy = yx$. Also, $A^{\top}=A(x, y)^{\top}=A\left(x^{\top}, y^{\top}\right)=A\left(x^{-1}, y^{-1}\right)$ and likewise for $B^{\top}$. It follows that $H_XH_Z^T=0$. There are $\ell m$ $X$ checks and $\ell m$ $Z$ checks in this generating set, though only $(n-k) / 2$ checks of either type are independent for a code encoding $k$ qubits.

We implemented a computerized search and obtained a new bivariate bicyclic code $[[98,6,12]]$ (not previously mentioned in Ref.~\cite{bravyi2024highthreshold}, but equivalent to a code in the database of Ref.~\cite{lin2023quantumtwoblockgroupalgebra}). This code uses fewer physical qubits than the $[[144,12,12]]$ Gross code for the same code distance, but has half the number of logical qubits.

The $[[98,6,12]]$ code is obtained from the BB construction by choosing $l= 7, m = 7$, and
\begin{equation}
    A = x^3+y^3+y^4, \quad B = y^6+x^2+x^5.
\end{equation}
Since both $A$ and $B$ polynomials have $3$ terms each, the $X$ and $Z$ stabilizers of this code in the canonical stabilizer basis (Eqn.~\ref{eqn:defn-bb-code}) are all weight $6$. The total qubit degrees are also exactly $6$.

Consider a $\overline{ZZ}$ logical parity measurement between logicals from two different LDPC codes, where $\overline{Z}_1$ is from $\mathrm{BB}_1$ code (now subscripted as code 1), and  $\overline{Z}_2$ from $\mathrm{LP}_2$ code (now subscripted as code 2). For simplicity, let us choose logical operators both of weight $14$. We also consider a joint-measurement between logicals $\overline{Z}_1$ (weight $14$) and $\overline{Z}_3$ (weight $12$) within the same high-rate code, $\mathrm{BB}_1$. The support of these operators within their respective codes is shown in Table~(\ref{tab:eg-logicalZsupports}) below.

\begin{table*}
\centering
\begin{tabular}{|c|c|c|c|c|}
\hline \hline
\text{Code} & $[[n,k,d]]$ & Logical & \text{$\overline{Z}$ 
Logical operator support} & \text { Weight } \\
\hline
$\mathrm{BB}_1$ & $[[98,6,12]]$ & $\overline{Z}_1$ & 6,  8, 13, 17, 31, 32, 33, 35, 36, 37, 41, 50, 51, 93 & 14 \\
$\mathrm{LP}_2$ & $[[200,20,10]]$ & $\overline{Z}_2$ & \; 24,  25,  26,  29,  30,  56,  58,  59,  60,  61,  90,  93,  94, 121 \; &  14 \\
$\mathrm{BB}_1$ & $[[98,6,12]]$ & $\overline{Z}_3$ & \; 10, 17, 35, 39, 42, 43, 53, 55, 61, 70, 84, 89 \; & 12 \\
\hline \hline
\end{tabular}
\caption{\label{tab:eg-logicalZsupports}Logical  operator supports for $\overline{Z}_1$ from $\mathrm{BB}_1$ code and $\overline{Z}_2$ from $\mathrm{LP}_2$ code. Logicals $\overline{Z}_1$ and $\overline{Z}_3$ are distinct logical operators in the same code, $\mathrm{BB}_1$. Note, the indices $1$, $2$ and $3$ subscripting the logical operators only serve a purpose to refer to an operator in this example and are unrelated to indexing of logical qubits in any specific logical basis. The indices subscripting the physical Pauli operators in the table correspond to the same ordering as columns from parity check matrices in Eq.~\ref{eqn:defn-bb-code} for BB$_1$ and Eqs.~\ref{eqn:defn-lp-code-hx}, \ref{eqn:defn-lp-code-hz} for LP${_2}$, where indexing of physical qubits begins from $0$.}
\end{table*}

\subsection{Auxiliary graphs for individual logical measurements}\label{subsec-eg-individ-auxiliary-graphs}

The scheme for building the auxiliary graph is generally applicable for any quantum LDPC code, and proceeds analogously for both the $\mathrm{BB}_1$ code and the lifted product $\mathrm{LP}_2$ code.

\begin{table*}[t]
\centering
\begin{tabular}{|c|c|c|c|}
\hline
& & & \\
\text{Original Code $[[\, n,k,d \,]]$}& $\mathrm{BB}_1$ $[[\, 98,6,12\,]]$ & $\mathrm{LP}_2$ $[[\,200,20,10\,]]$ & $\mathrm{BB}_1$ $[[\, 98,6,12\,]]$\\
\hline
\hline
& & & \\
Total phys qubits & $98$ (data) + $98$ (ancilla) & $200$ (data) + $192$ (ancilla) & $98$ (data) + $92$ (ancilla) \\
(data $+$ ancilla) & = $196$ & $=$ $392$ & = $190^{***}$ \\
in original codeblock & & & \\
\hline
Max stabilizer weights (X,Z) & $6,6$ & $7,7$ & $6,6$ \\
Max qubit degree & $6$ & $8$ & $6$ \\
\hline
\hline
& & & \\
Auxiliary graph & $G_1(V_1,E_1)$ to measure $\overline{Z}_1$ & $G_2(V_2,E_2)$ to measure $\overline{Z}_2$ & $G_3(V_3,E_3)$ to measure $\overline{Z}_3$ \\
\hline
& & & \\
Parameters of deformed code & $[[\, 121 , 5 , 12 \,]] $ & $[[\, 220 , 19 , 10 \,]] $ & $[[\, 115 , 5 , 12 \,]] $ \\
\hline
& & & \\
No. of edges (addnl data qubits) &  23   (incl 2 extra edges$^{*}$) & 20 & $17$ (incl 1 extra edge$^{**}$) \\
No. of vertices (addnl Z checks) & 14 & 14 & $12$\\
No. of cycles (addnl X checks) & 10 (incl 2 new cycles$^{*}$) & 7 & $6$ (incl 1 new cycle$^{**}$) \\
\hline
& & & \\
Cycle basis & $^{*}$Cellulated to obtain cycles & No cellulation or & $^{**}$Cellulated to obtain cycles \\
& at most weight $6$. & decongestion needed. & at most weight $6$ \\
& & & \\
\hline
& & & \\
Total additional overhead & In total, $23$ (data) & In total, $20$ (data) & In total, $17$ (data) \\
& + $24$ (new checks) & + $21$ (new checks) & + $18$ (new checks) \\
& $= 47$ phys qubits added & 
$= 41$ phys qubits added & $= 35$ phys qubits added \\
\hline
Max stabilizer weights (X,Z) & $7,6$ & $8,7$ & $7,6$ \\
Max qubit degree & $7$ & $8$ & $7$ \\
\hline
& & & \\
Total No. of phys qubits & $121$ (data) + $122$ (ancilla)  & $220$ (data) + $213$ (ancilla) & $115$ (data) + $110$ (ancilla) \\
(data $+$ ancilla) & $= 243$ & = $433$  & = $225$ \\
in deformed code & & & \\
\hline \hline
\end{tabular}
\caption{\textbf{A tale of three auxiliary graphs}. Parameters and sparsity of the deformed codes to individually measure logicals in the two example LDPC codes. Here, logicals $\overline{Z}_1$, $\overline{Z}_2$, $\overline{Z}_3$ are from BB$_1$, LP$_2$, BB$_1$ codes respectively. The deformed code in either case has one less logical qubit (which has been measured), but preserves code distance. For all $3$ constructions above, the underlying graph already had sufficient expansion to preserve code distance, and did not require additional edge qubits after cellulation to boost relative expansion. $^{***}$The reader may observe that the original BB$_1$ codeblock size is different in columns 1 and 3. This is because the above reports $G_1$ deformed code parameters when deforming $X$ checks from the canonical stabilizer basis (Eqn.~\ref{eqn:defn-bb-code}) for BB codes which has redundancy in checks, for fair comparison with LP$_2$, which also has redundancy in its canonical stabilizer basis (Eqn.~\ref{eqn:defn-lp-code-hx}). On the other hand, $G_3$ deformed code parameters are using full-rank $X$ checks. We observed that auxiliary graph $G_1$ had the same structure and overhead for both canonical and full-rank choices of BB$_1$ $X$ stabilizers, while $G_3$ required $4$ less physical qubits. The sparsity of deformed code is left underoptimized for simplicity as it is $\leq 8$. The $\mathrm{LP}_2$ auxiliary graph has cycle basis lengths bounded to at most $7$ edges, and since the original $\mathrm{LP}_2$ code already had stabilizers of weight $7$, we did not cellulate cycles in $G_2$ of edge length $7$. Some of the original code $X$ stabilizer weights increased by exactly 1 since they deformed to gain on a single new edge qubit. $\overline{Z}_1$ and $\overline{Z}_2$ will serve as inputs for inter-code joint-measurement in Section~\ref{subsec-eg-multi-code-BB-LP-adapter}, and $\overline{Z}_1$ and $\overline{Z}_3$ will serve as inputs for intra-code joint-measurement within the BB$_1$ codeblock in Section~\ref{subsec-eg-intra-code-BB-BB-adapter}.}
\label{tab:tale_of_three_auxiliary_graphs_BB_LP}
\end{table*}

In case of both $[[98,6,12]]$ $\mathrm{BB}_1$ code as well as the $[[200,20,10]]$ $ \mathrm{LP}_2$ code, the structure of $H_{X}$ restricted to the support of ${\overline{Z}}$, is graph-like, which is to say that each $X$ stabilizer from the original code anti-commutes with exactly $0$ or $2$ vertex $Z$ checks. This bypasses the need for decomposing each hyperedge into disjoint edges. Corresponding to each pair of vertex checks that anti-commute within the support of any $X$ stabilizer, an additional ``edge" qubit is introduced. This is the first step that introduces any additional physical data qubits. For the $[[98,6,12]]$ $\mathrm{BB}_1$ code, $21$ edges are added to the auxiliary graph, which corresponds to $21$ physical qubits introduced at this stage. Similarly, for the $[[200,20,10]]$ $\mathrm{LP}_2$ code, $20$ edge qubits are introduced into its auxiliary graph at this stage.

Now the checks in the deformed codes form an abelian group. However, there are additional gauge degrees of freedom due to cycles in the auxiliary graphs. One way to gauge fix the subsystem code is to add $X$ stabilizers supported on the edges of a generating set of cycles in the auxiliary graph. 

The $[[98,6,12]]$ $\mathrm{BB}_1$ logical auxiliary graph in its present form has $8$ independent cycles in its cycle basis. For one choice of such a basis, the \textit{edge length} (the number of edges present in the cycle) of cycles are given by $[\,3, 3, 5, 5, 5, 6, 7, 9\,]$. Let us assume a limit to cycle edge lengths to be $6$, the same as the stabilizer weight in the original $\mathrm{BB}_1$ code. We cellulate by adding an additional edge to split the longest cycle on $9$ edges into two cycles of edge lengths $5$ and $6$, and also cellulate the next longest cycle of edge length $7$ by adding an edge to split into two cycles of edge lengths $4$ and $5$. In total, $2$ edges have been added. Post cellulation, the cycle basis in the auxiliary graph of $[[98,6,12]]$ logical $\overline{Z}_1$ has $10$ cycles in its cycle basis. The weights of new $X$ checks defined on these cycles are $[\,3, 3, 4, 5, 5, 5, 5, 5, 6, 6]$, bounded at weight $6$. Further, each edge is in at most $4$ cycles. In fact, most edges are in only 1 or $2$ cycles. So qubit $X$-check degree is naturally bounded with no further need for decongestion.

At this stage, we have a sparse, deformed $[[121 , 5 , 12]]$ code in which one logical qubit from the original $[[98,6,12]]$ $\mathrm{BB}_1$ code has been measured. Using BP-OSD~\cite{Roffe_LDPC_Python_tools_2022} as well as integer programming (CPLEX~\cite{ibm-cplex}), we find the deformed code already has code distance of $12$, same as the original code. Hence, there is no need to add more edge qubits to boost relative expansion in this auxiliary graph, which we will use as the final auxiliary graph for $\mathrm{BB}_1$ code logical, $\overline{Z}_1$. Let the graph incidence matrix for this graph be denoted $G_1(V_1,E_1)$ where $V_1$ is the set of vertices and $E_1$ the set of edges. In total $|E_1| = 21$ (perfect matching step) + $2$ (cellulation) $=23$ edge qubits are required. The total additional overhead of the single logical $\overline{Z}_1$ measurement is $23$ (data qubits) + $14$ (Z-check ancilla) + $10$ (X-check ancilla) $=47$ extra physical qubits. The $\mathrm{BB}_1$ codebock originally originally contained $196$ physical qubits ($98$ data qubits and $98$ overcomplete set of checks), so this deformed code allows extending the quantum memory to a partial logic processor at only $37.7\%$ additional space cost.

Let us return to the construction of the graph for the $[[200,20,10]]$ lifted product code logical $\overline{Z}$ operator. The cycle basis of this graph has $7$ elements. One particular choice of cycle bases generated using NetworkX~\cite{networkX} has cycles with edge lengths $[4, 4, 4, 6, 6, 7, 7]$, with longest cycle having $7$ edges. Note, in this case, since the original $\mathrm{LP}_2$ code already has stabilizer weights $7$, we choose not to cellulate these cycles. $7$ new $X$ cycle checks are defined, as supported on edge qubits in these cycles. Each edge appears in at most $5$ cycles. In fact, the predominant nonzero cycle degree of edges is $1$ and there is only a single edge which participates in $5$ cycles. Hence for simplicity we choose to omit the thickening step for decongesting this cycle basis. 

Based on BP-OSD and mixed-integer programming using CPLEX, the deformed $\mathrm{LP}_2$ code has code distance of $10$, the same as the original code, so we do not add additional edges.  Let $G_2(V_2,E_2)$ denote the graph incidence matrix of this final auxiliary graph for $\mathrm{LP}_2$ code logical $\overline{Z}_2$. Only $|E_2|=20$ edges were added between $|V_2|=14$ vertices to create an auxiliary graph for measuring a weight $|V_2|$ logical. Speaking in terms of physical overhead, in total this auxiliary system required $20$ (data qubits) + $14$ (Z check ancilla) + $7$ (X check ancilla) $= 41$ additional physical qubits, which is only about $10.5\%$ additional space overhead over the original memory cost of $392$ physical qubits (including both data and checks) to store the quantum code. The deformed code obtained from $\mathrm{LP}_2$ has parameters $[[220,19,10]]$.

\begin{figure}[t]
    \centering
        \includegraphics[width=\linewidth]{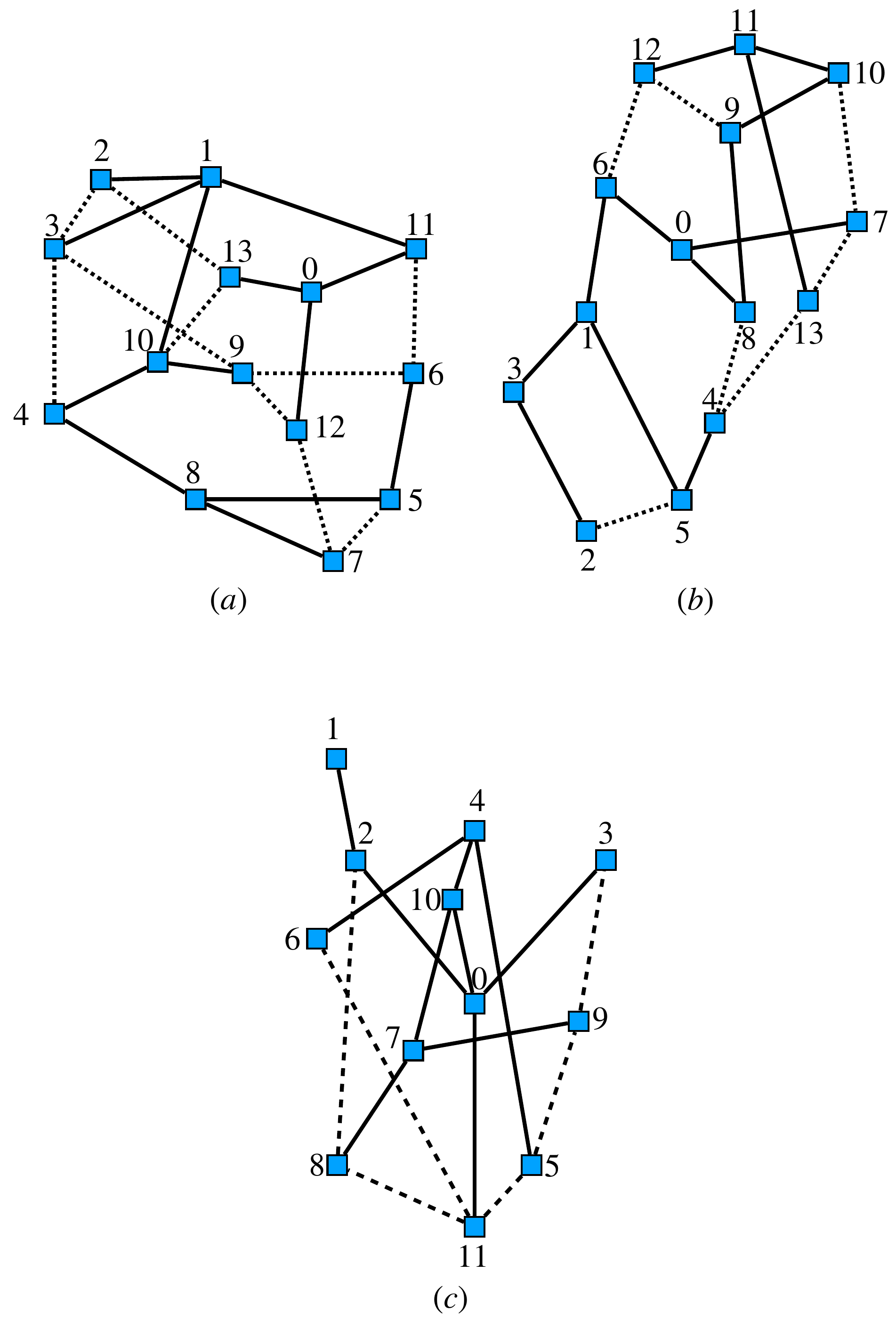}
    \caption{Auxiliary graphs for chosen BB and LP code logicals. Solid edges are the subset of edges in the original graph $G_i(V_i,E_i)$ that belong to a chosen spanning tree to input to the $\mathsf{SkipTree}$ algorithm in Sections~\ref{subsec-eg-multi-code-BB-LP-adapter} and~\ref{subsec-eg-intra-code-BB-BB-adapter}. Dotted edges are the redundant edges in the graph which will be ignored for $\mathsf{SkipTree}$. The vertex labeling $0,1,...|V_i-1|$ is the output obtained from the $\mathsf{SkipTree}$ algorithm.
    (a) $G_1(V_1,E_1)$ represents the graph to measure a weight-$14$ logical~$\overline{Z}_1$ in the $[[98,6,12]]$~$\text{BB}_1$~code.
    (b) $G_2(V_2,E_2)$ represents the graph to measure a weight-$14$ logical~$\overline{Z}_2$ in the $[[200,20,10]]$~$\mathrm{LP}_2$~code.
    (c) $G_3(V_3,E_3)$ represents the graph to measure a weight-$12$ logical~$\overline{Z}_3$ in the $[[98,6,12]]$~$\text{BB}_1$~code.}
    \label{fig:bb-lp-aux-graphs}
\end{figure}

The procedure to design the auxiliary graph to individually measure $\mathrm{BB}_1$ logical $\overline{Z}_3$ follows analogously to the previous two examples, with the difference being that $\overline{Z}_3$ has weight $12$, so $12$ $Z$ vertex checks are added in this individual auxiliary graph. We also observed that using a full-rank basis of $X$ stabilizers in the original code required fewer additional edge qubits during the deformation through perfect matchings step (only $16$ edge qubits), in contrast to deforming $X$ stabilizers in the canonical stabilizer basis (Eqn.~\ref{eqn:defn-bb-code}) which required $18$ edge qubits. (For completeness, we mention that the auxiliary graph $G_1$ for measuring $\overline{Z}_1$ and its physical overhead remained identical for both the canonical stabilizer basis (Eqn.~\ref{eqn:defn-bb-code}) as well as full rank $X$ stabilizer basis). We denote the incidence matrix of this graph as $G_3(V_3,E_3)$, where $V_3$ is the set of vertices and $E_3$ the set of edges. The number of edge qubits $E_3$ is $16$ (perfect matching step) + $1$ (cellulation) = $17$. The total overhead for this measurement graph is $17$ (data qubits) + $12$ (Z check ancilla) + $6$ (X check ancilla) = $35$ additional physical qubits, which is less than that for the $\overline{Z}_1$ logical due to lower weight of the logical representative. The deformed code obtained from $\mathrm{BB}_1$ to measure $\overline{Z}_2$ has parameters $[[115,5,12]]$. Constructions for all three auxiliary graphs for single logical measurements, along with their space overhead, are summarized in Table~\ref{tab:tale_of_three_auxiliary_graphs_BB_LP}. 

Notice that the multiplicative factor of $O(\log^3(d))$ for decongestion guarantees did not play a role in this specific example due to already low constants, although this step can also be very useful for hardware with limited qubit degree. 

Without optimizing further for these specific codes, we now move on to illustrating the joint measurement scheme for the chosen logical operators, using auxiliary graphs $G_1$ and $G_2$ as inputs.

\subsection{BB-LP inter-code joint logical measurements using a universal adapter}\label{subsec-eg-multi-code-BB-LP-adapter}

The goal is to measure the parity of these two logicals in different quantum LDPC codes without inferring either of the individual logical eigenvalues individually. Such a primitive is very useful for schemes such as state teleportation, lattice surgery, and more general code switching. In the present form, the structures of the two auxiliary graphs are very different (See Fig.~\ref{fig:bb-lp-aux-graphs} for $G_1(V_1,E_1)$ and $G_2(V_2,E_2)$) and it is not immediately apparent how to connect these. A key insight in this work is that the repetition code can be used to mediate between these two structures. 

Applying the slightly optimized flag-based $\mathsf{SkipTree}$ algorithm (Algorithm~\ref{alg:skiptree_HR} in Appendix~\ref{app:skip_fullrank}) to the auxiliary graph $G_1$ for $[[98,6,12]]$ $\mathrm{BB}_1$ code logical, we obtain sparse $T_1$ and permutation $P_1$, such that $T_1 G_1 P_1=H_{R(14)}$, where $H_{R(d)}$ is the canonical basis of the repetition code, here of distance $d=14$. The algorithm thus provides a new set of edges between the same vertices in $G_1$, as well as a relabeling of the $14$ vertices such that one can follow a path in order in the new graph (in other words, one can visit newly labeled vertices in order $0,1,2,...,13$), see Fig.~\ref{fig:bb-lp-aux-graphs}(a) for such an example. In the new basis and relabeling, the deformed $X$ stabilizers from original code appear to support an $X$-type repetition code on the edge qubits. $T_1$ obtained from $\mathsf{SkipTree}(G_1)$ is $(3,2)$-sparse, as expected (Theorem \ref{thm:skiptree}).

Next, also apply $\mathsf{SkipTree}$ algorithm to the auxiliary graph $G_2$ for $[[200,20,10]]$ lifted product code logical, and obtain sparse $T_2$ and permutation $P_2$, such that $T_2 G_2 P_2=H_{R(14)}$, where $T_2$ is also $(3,2)$-sparse. We obtain an analogous relabeling of vertices in $V_2$, from $0$ to $13$, see Fig.~\ref{fig:bb-lp-aux-graphs}(b).

For the inter-code joint measurement, let us now consider the combined system of both the $\mathrm{BB}_1$ bivariate bicyclic codeblock and $\mathrm{LP}_2$ lifted product codeblock, with their respective individual auxiliary graphs $G_1$ and $G_2$ constructed in the previous section. At present, these are two disconnected graphs. The procedure to connect these is straightforward; simply add new edges $E_A$ that connect vertices with the same label in auxiliary graphs $G_1$ and $G_2$, e.g.~the vertex labeled $0$ in $G_1$ to the vertex labeled $0$ in $G_2$, vertex labeled $1$ in $G_1$ to vertex labeled $1$ in $G_2$, and so on, for the $w=14$ pairs of vertices in the two auxiliary graphs. This requires $14$ edges, which in physical terms translates to $14$ additional physical data qubits.

The additional cycles in this adapted graph create new $X$ cycle checks. Alternatively to the graph-based view, a very straightforward way to construct the adapter $X$-check matrix is to follow our compact Tanner graph notation, which is used to describe the jointly deformed code in Fig.~\ref{fig:joint_measurement}, where $T_1$, $T_2$, $P_1$, $P_2$ define compact Tanner graph edges, that represent parity check matrices between relevant qubit sets and check sets. For just the adapter $X$ checks, the compact Tanner graph diagram contains the following information,
\begin{align}\label{eqn-adapter-x-checks-BB-LP}
  H_X|_{\supp{(E_1)}} &= T_1 \nonumber \\ H_X|_{\supp{(E_2)}} &= T_2  \nonumber \\
  H_X|_{\supp{(E_A)}} &= H_{R(14)}
\end{align}
where $E_A$ refers to the set of edges just added as part of the inter-code adapter, and $H_{R(14)}$ is the full-rank canonical check matrix of the repetition code of distance $14$.

\begin{table}[t]
\centering
\begin{tabular}{|c|c|}
\hline
& \\
\text{Original Codes}& $\mathrm{BB}_1$ $[[98,6,12]]$, and \\
& $\mathrm{LP}_2$ $[[200,20,10]]$ \\
\hline
& \\
Auxiliary graph & $G_1(V_1,E_1)\sim_{\mathcal{A}}G_2(V_2,E_2)$ \\
(using adapter) & to measure $\overline{Z}_1 \overline{Z}_2$ \\
\hline \hline
& \\
No.\! edges (addnl data qubits) & $57$ \\
No.\! vertices (addnl Z checks) & $28$\\
No.\! cycles (addnl X checks) & $17$ (incl 2 new cycles) \\
& + $13$ (adapter) \\
\hline
& \\
Total additional overhead & $47+41$ (individual aux graphs) \\
& + $27$ (for adapter) \\
&  $=$ $115$ qubits added in total \\
\hline
& \\
Adapter checks
& at most weight $8$.\\
\hline
& \\
Parameters of deformed code &  $[[\, 355,25,10\, ]]$ \\
\hline
& \\
Max stabilizer weights (X,Z) & $8$, $7$ \\
Max qubit degree & $9$ \\
\hline \hline
\end{tabular}
\caption{\textbf{Inter-codeblock logical measurements, using a universal adapter}. The logical operators being measured are from the $[[98,6,12]]$ bivariate bicyclic code and $[[200,20,10]]$ lifted product code. A modular graph to measure $\overline{Z_1 Z_2}$ is composed from individual auxiliary graphs $G_1$ and $G_2$ (which measure $\bar{Z}_1$ and $\bar{Z}_2$ respectively) described in Table.~\ref{tab:tale_of_three_auxiliary_graphs_BB_LP}, and the universal adapter qubits. The cycles that occur within individual auxiliary graphs are already listed in Table.~\ref{tab:tale_of_three_auxiliary_graphs_BB_LP}, so only adapter cycle checks which occur between the two graphs are mentioned here.}
\label{tab:params_auxiliary_graphs_BB-LP_intercode}
\end{table}

The new $X$ stabilizers introduced to interface between the bivariate bicyclic $\mathrm{BB}_1$ code and lifted product $\mathrm{LP}_2$ code are at most weight $8$, as guaranteed by $\mathsf{SkipTree}$ algorithm. For a run of the algorithm with random choices of spanning trees for $G_1$ and $G_2$, the adapter $X$-check weights were $[4, 5, 5, 5, 5, 5, 5, 5, 6, 6, 6, 7, 8]$.

\begin{remark}
We note that specific choices of spanning tree can yield lower adapter check weights. For example, both $G_1$ and $G_2$ contain a \textit{Hamiltonian path}, which is a path that visits each vertex in the graph exactly once. Inspect the graphs in Fig.~\ref{fig:bb-lp-aux-graphs} to see that they have such paths. As noted in Section~\ref{sec:skiptree}, Hamiltonian paths can be used to construct an adapter -- simply add edges connecting corresponding vertices in the paths. This is exactly the solution given by the Appendix~\ref{app:skip_fullrank} version of the $\mathsf{SkipTree}$ algorithm when the spanning trees of both auxiliary graphs $G_1$ and $G_2$ are chosen to be the Hamiltonian paths with roots at ends of the paths. The adapter checks will then all have weight four.

For the purpose of demonstrating the flexibility of the $\mathsf{SkipTree}$ algorithm in a general setting, we choose to apply it to randomly generated spanning trees for each graph, which are shown in Fig.~\ref{fig:bb-lp-aux-graphs}. Generally, even for a graph with a Hamiltonian path, it could be desirable for other practical reasons to choose a different spanning tree, such as the presence of defective couplers, or a desire to not increase the degree of certain edge qubits (note, qubits on edges belonging to the spanning tree will have their degree increase by at least one, as the edge qubit will be forced to also participate in at least one adapter check).
\end{remark}


The complete set of $X$ stabilizers for the adapted code consists of the $X$ stabilizers the original codes (some of which are deformed onto auxiliary graph edges $E_1$ and $E_2)$, $X$ cycle checks in both auxiliary graphs, as well as the new checks from the universal adapter (Eqn.~\ref{eqn-adapter-x-checks-BB-LP}).

The merged code described by the adapted auxiliary graph $G_1(V_1,E_1)\sim_{\mathcal{A}}G_2(V_2,E_2)$ contains $25$ logical qubits, which is one less than the logical qubits originally in the two codeblocks (BB$_1$ codeblock contained $6$ logical qubits and the LP$_2$ codeblock contained $20$ logical qubits). This corroborates that a single logical operator has been measured out. It is observed using CPLEX and BP-OSD that the deformed code has code distance $10$, which is the minimum of the code distances of the two input codes, as expected based on Section~\ref{sec:rep_adapter}. The overall parameters of this deformed code formed by fusing the bivariate-bicylic $+$ lifted product code is $[[355,25,10]]$

The adapter itself uses only $w$ extra qubits and $w-1$ checks, here costing $27$ physical qubits for $w=14$.

For the purpose of illustrating typical sparsity of deformed codes in this scheme, we present final check matrices of deformed codes with the arbitrary spanning tree inputs to $\mathsf{SkipTree}$ as supplementary information~\cite{zenodopcmadapter}.

\subsection{BB-BB intra-code joint logical measurements}\label{subsec-eg-intra-code-BB-BB-adapter}

Pauli-based computation on high-rate quantum codes often relies on joint logical qubit measurements within the same codeblock. Consider a $\overline{Z_1 Z_3} $ measurement within the $[[98,6,12]]$ $\mathrm{BB}_1$ code, where the support of the logicals is described in Table.~\ref{tab:eg-logicalZsupports}, and $\overline{Z}_1$ is same as in the inter-code joint-measurement in Sec.~\ref{subsec-eg-multi-code-BB-LP-adapter}. From Table~\ref{tab:eg-logicalZsupports}, we see that logicals $\overline{Z}_1$ and $\overline{Z}_3$ overlap on exactly two physical qubits, indexed $17$ and $35$ here.

There are two possibilities to proceed. The most general setting for measuring logicals with overlapping support is handled using the set-valued port function introduced in Appendix~\ref{app:set-valued_port_function}, where the same physical qubit supports a vertex check for each logical operator supported on it. The second approach is offered by support lemma~\ref{lem:supportlemma} in Appendix~\ref{app:supportlemma}, we know that any pair of $\overline{Z}$ logicals can be made disjoint by cleaning the overlapping support through stabilizers. However, in some cases multiplication with stabilizers could potentially lead to higher weight representatives, which also affects the qubit overhead. Since overlapping logical operators are a common occurrence in high-rate quantum codes, we demonstrate the first approach.

The auxiliary graphs for $\mathrm{BB}_1$ logical $\overline{Z}_1$ and $\overline{Z}_3$ were constructed in~\ref{subsec-eg-individ-auxiliary-graphs}, and the $\overline{Z}_1$ graph used for inter-code joint-measurement in Sec.~\ref{subsec-eg-multi-code-BB-LP-adapter} can be reused now for the purpose of the intra-codeblock joint measurement, which is a benefit of a modular joint measurement scheme. This could also potentially reduce requirements on chip design in hardware such as superconducting qubits.

First, we consider how the $\mathrm{BB}_1$ $\overline{Z}_1$ and $\mathrm{BB}_1$ $\overline{Z}_3$ auxiliary graphs attach to the original LDPC codeblock. Recall that the logical operators do not have disjoint support on the physical qubits in the codeblock, in fact share support on two physical qubits in the original $\mathrm{BB}_1$ code. In practice, the auxiliary graphs to measure $\overline{Z}_1$ and $\overline{Z}_3$ can be simultaneously attached to the original code. When viewed as a single, larger graph, this is atypical as it deviates from the scheme in Sec.~\ref{sec-gauging-meas} with there now being \textit{two} vertex checks from one physical qubit, for all physical qubits in $\supp{\overline{Z}_1} \cap \supp{\overline{Z}_3}$, meaning that the port function $f$ for the combined graph is no longer injective, but instead set-valued. This case falls under the generalization presented in Theorem~\ref{thm:joint_logical_measurement}, which discusses joint-measurements when $\overline{Z}$ logicals are \textit{sparsely overlapping}. In this case, $X$ checks from the original code, which have support on physical qubits in this overlap are deformed onto two separate matchings since there are two vertices corresponding to each shared physical qubit. When taking a product of all vertex checks from both auxiliary graphs, the Pauli $Z$ support on the physical qubits in the intersection is cancelled out by the two different vertex checks, and we obtain a $\overline{Z}$ operator supported only on $\supp(\overline{Z_1 Z_3})$. 

\begin{table}[t]
\centering
\begin{tabular}{|c|c|}
\hline
& \\
\text{Original Code}& $\mathrm{BB}_1$ $[[98,6,12]]$  \\
& \\
\hline
& \\
Auxiliary graph & $G_1(V_1,E_1)\sim_{\mathcal{A}}G_3(V_3,E_3)$ \\
(using adapter) & to measure $\overline{Z}_1 \overline{Z}_3$ \\
\hline \hline
& \\
No. edges (addnl data qubits) & 52 \\
No. vertices (addnl Z checks) & $26$\\
No. cycles (addnl X checks) & $16$ (incl 3 new cycles) \\
& + $11$ (adapter) \\
\hline
& \\
Total additional overhead & $47+35$ (individual aux graphs) \\
& + $23$ (for adapter) \\
&  $=$ $105$ qubits added in total \\
\hline
& \\
Adapter checks & at most weight $8$.\\
& \\
\hline
& \\
Parameters of deformed code &  $[[\,150,5,12\,]]$\\
\hline
& \\
Max stabilizer weights (X,Z) & $8,6$ \\
Max qubit degree & $9$ \\
& \\
\hline
& \\
Total phys qubits (data $+$ & $150$ (data) + $145$ (ancilla) \\
 ancilla) in deformed code & = $295$ \\
 & \\
\hline \hline
\end{tabular}
\caption{\textbf{Intra-codeblock logical measurements, using a universal adapter}. A modular graph to measure $\bar{Z}_1\bar{Z}_3$ in the $[[98,6,12]]$ $\mathrm{BB}_1$ code, composed of individual auxiliary graphs to measure $\bar{Z}_1$ and $\bar{Z}_3$, and the adapter qubits. Parameters and sparsity the deformed codes used for joint logical measurement. For comparison, the $\mathrm{BB}_1$ codeblock has $98$ (data) + $92$ (ancilla for full rank check matrices) = $190$ physical qubits before code deformation. }
\label{tab:params_auxiliary_graphs_BB_intracode}
\end{table}

Following the above description of how the two auxiliary graphs connect to the original code, it now remains to connect the two auxiliary graphs to each other. Connecting the $Z$ vertex checks of the two graphs to an additional set of adapter qubits will also ensure that the product of only $V_1$ or $V_3$ vertex checks cannot be used to learn the $\overline{Z}_1$ or $\overline{Z}_3$ logical measurement value alone.

At this stage, a second feature to note in this example is that the logicals $\overline{Z}_1$ and $\overline{Z}_3$ are of unequal weight. In such a case, the adapter is constructed using a repetition code of distance equal to the minimum of the two logical weights, here, $12$.

Apply $\mathsf{SkipTree}$ algorithm to $G_3(V_3,E_3)$, to obtain output $T_3$, an $11 \times 11$ binary matrix and permutation $P_3$, a $12 \times 12$ binary matrix, such that $T_3 G_3 P_3 = H_{R(12)}$. This provides a relabeling of indices for vertex checks $V_3$, which are labeled $0$ to $11$ (or $1$ to $12$ if $1$-indexing) for vertices in graph $G_3$, see Fig.~\ref{fig:bb-lp-aux-graphs}(c). In Sec.~\ref{subsec-eg-individ-auxiliary-graphs}, we had already obtained $T_1$ and $P_1$ such that $T_1 G_1 P_1 = H_{R(14)}$, which provide a relabeling of indices of vertex checks $V_1$, which are labeled $0$ to $13$ (or $1$ to $14$ if $1$-indexing).

To construct the adapted graph, we will add $\min(w_1,w_2)=12$ more edges $E_a$, by simply adding a new edge to connect each vertex with the same label from either graph. Since $w_1=14 < 12$, we only connect edges to the first $12$ vertices in $G_1$, and therefore $|E_1|=12$. The additional cycles created in this adapted graph correspond to new $X$ cycle checks, which had repetition code-like support. The solutions to $\mathsf{SkipTree}(G_1)$ and $\mathsf{SkipTree}(G_3)$, given by $T_1$, $P_1$ and $T_3$, $P_3$ respectively, simply define parity check matrices between relevant qubit and check sets in the compact Tanner graph notation (Fig.~\ref{fig:joint_measurement}). Since qubit set $E_1$ is chosen to be of cardinality $|E_1| = 12$, for the deformed code Tanner graph we only need the first $12$ rows of $T_1$ instead of all $14$, defined as $T_1'$ (alternatively written as $T_1' = T_1[:\!11]$ in pythonic terms). So in a similar vein to Eqn.~\ref{eqn-adapter-x-checks-BB-LP} we obtain,
\begin{align}
  H_X|_{\supp{(E_1)}} &= T_1' \nonumber \\ H_X|_{\supp{(E_3)}} &= T_3 \nonumber \\
  H_X|_{\supp{(E_{a})}} &= H_{R(12)}
\end{align}
where recall $E_a$ refers to the set of edges just added as part of the intra-code adapter, and $H_{R(12)}$ is the full-rank canonical check matrix of the repetition code of distance $12$.

The additional cost of the adapter alone is $12$ (data qubits) + $11$ (X ancilla) $= 23$ extra physical qubits. In total, $52$ (data qubits) + $26$ (Z check ancilla) + $27$ (X check ancilla) $= 105$ additional physical qubits are required. Starting from the original $[[98,6,12]]$ codeblock, which consists of $190$ physical qubits when using the full rank stabilizer basis, the deformed code for joint $\overline{Z_1 Z_3}$ measurement has parameters $[[150,5,12]]$ and a total of $295$ physical qubits, assuming a separate ancilla qubit for each stabilizer readout as has been consistent throughout this study. The deformed code distance was verified using CPLEX~\cite{ibm-cplex} and BPOSD~\cite{Roffe_LDPC_Python_tools_2022}.

In contrast, a specialized graph to measure $\overline{Z_1 Z_3}$ requires a deformed code with parameters $[[137,5,12]]$, constructed using additional $39$ edge qubits and $40$ new checks, making the combined system a total of $269$ physical qubits. Despite the lower aggregate physical qubit count of a bespoke graph for $\bar{Z}_1\bar{Z}_2$, the clear advantage of a modular graph is the reusability of components for different joint measurements instead of being restricted for a specialized task. The adapter size $\simeq 23$ qubits can be reused independently of the code choice. 

Parity check matrices of original and deformed LDPC codes are available as Supplementary information at Zenodo~\cite{zenodopcmadapter,univ_adapters_gitrepo}.

\section{Discussion} \label{sec:discussion}

We present a $O(d\log^3d)$ space overhead scheme to implement joint Pauli logical measurements, a gateset to implement logical Clifford gates fault-tolerantly on and between arbitrary quantum LDPC codes. The universal adapter scheme also finds use in efficiently connecting different LDPC codes with similarly low overhead, thus forming a useful primitive for multi-code architectures including magic state factories for universal computation. 

Our sparse basis transformation method by means of the $\mathsf{SkipTree}$ algorithm helped give rigorous guarantees on the merged code to be LDPC. Our construction not only provides guarantees in the asymptotic regime, but also boasts favorable constants which can encourage usage for fault-tolerant computation in the near-term: the universal adapter costs only $2d$ extra qubits and checks, and the additional checks are of weight at most $8$ (and likely less in practice).


Auxiliary graph LDPC surgery essentially reduces the problem of code deformation to a graph design problem. There are assumptions placed in graph desiderata and adapters of Sections~\ref{sec-gauging-meas} and \ref{sec:rep_adapter} that are likely not strictly necessary, and we explored cases where some of these assumptions are relaxed in Section~\ref{sec:leveraging_code_properties}, including a relaxation on expansion in some cases, and decongestion for some geometrically local codes.

 We found the notion of relative expansion defined in Definition~\ref{def:expansion} useful for our proofs of code distance and our graph constructions. It is also significantly weaker than requiring the graph have constant global expansion. However, in some cases, even relative expansion may be unnecessary. One simple example of this case is the fault-tolerant lattice surgery between two surface code patches, which can be viewed as a special case of the auxiliary graph surgery framework \cite{williamson2024gauging} (and an example of CSS code surgery \cite{cowtan2024css}), but the relative expansion of the auxiliary graph decreases with the code distance as $O(1/d)$. This seems to be sufficient partly because the surface code does not encode more than one logical qubit. Is there a way to more accurately capture the necessary expansion properties of the auxiliary graph in general?

In future work we also anticipate a scope to attach adapters to separate ports and the potential to define a notion of multi-port relative expansion. 


The results for joint measurements in this work are applicable to logical operators that overlap on at most a constant number of qubits. While this is significantly stronger than the naive scenario of dealing with physically disjoint logical operators, it would be interesting to be able to measure the product of overlapping logical operators in general within the framework we have presented.

Decoding through the process of auxiliary graph surgery involves decoding the merged code(s) created in the process. In Ref.~\cite{cross2024linear}, a modular decoder is introduced which can split this decoding problem (for one type of Pauli error) into two spatially separate pieces -- an LDPC decoding problem on the original code, and a matching problem on the auxiliary graph. We expect the same modular decoding can be done for auxiliary graph LDPC surgery, but splitting the decoding problem further when measuring a large product of logical operators would seem to be necessary for fast decoding of LDPC surgery in practice, unless this situation can otherwise be avoided.

As a byproduct of the discussion on targeted $\overline{\cnotgate}$ logical gates using the toric code adapter, we also provided an alternate proof of Dehn twists on a toric code using the simple language of linear algebra instead of geometric topology \cite{koenig2010turaev,zhu2020dehn}, which could perhaps be of independent interest. We anticipate the flexibility of the universal adapters developed in this paper to facilitate similar constructions that enable known gates from a richer variety of codes with similar distance to be connected to an efficient quantum LDPC memory.



\section*{Acknowledgements}
The authors thank Andrew Cross, Sunny He, Anirudh Krishna and Dominic Williamson for inspiring discussions. TY also thanks Sunny for suggesting the use of the term universal adapter during a discussion in September 2024 and for pointing out that actually $\log^3d$ layers are sufficient for decongestion. This work was done in part while ES was visiting the Simons Institute for the Theory of Computing, Berkeley, supported by DOE QSA grant \#FP00010905. Research at Perimeter Institute
is supported in part by the Government of Canada through
the Department of Innovation, Science and Economic Development Canada and by the Province of Ontario through the Ministry of Colleges and Universities.

\bibliographystyle{unsrt}
\bibliography{refs}

\appendix


\section*{Appendices A-G}\label{app:header}
\addcontentsline{toc}{section}{\nameref{app:header}}
\addtocontents{toc}{\protect\setcounter{tocdepth}{0}}

\newcommand{\nocontentsline}[3]{}
\newcommand{\tocless}[2]{\bgroup\let\addcontentsline=\nocontentsline#1{#2}\egroup}

\tocless\section{Proofs for Theorem~\ref{thm:graph_desiderata}}\label{app:desiderata_proofs}

\tocless\subsection{Codespace of the deformed code in auxiliary graph surgery}

Here we show that desideratum 0 in Theorem~\ref{thm:graph_desiderata} is sufficient for the gauging logical measurement to measure the desired operator $\overline{Z}=Z(\mathcal{L})$. It is also necessary without loss of generality, as we remark below.

\begin{lemma}\label{lem:deformed_codespace}
Provided $\mathcal{G}$ is connected, the deformed code in auxiliary graph surgery (see Fig.~\ref{fig:gauging_measurement}) has one less logical qubit than the original code and $\overline{Z}=Z(\mathcal{L})$ is in the stabilizer group of the deformed code.
\end{lemma}
\begin{proof}
Let $\overline{X}^{(i)},\overline{Z}^{(i)}$ for $i=0,1,\dots,k-1$ denote a symplectic basis of logical operators for the original code. That is, $\overline{X}^{(i)}$ and $\overline{Z}^{(j)}$ anticommute if and only if $i=j$. Note these do not have to be $X$ and $Z$ type operators as we are considering an arbitrary (potentially non-CSS) stabilizer code. We let $\overline{Z}^{(0)}$ be the operator $\overline{Z}=Z(\mathcal{L})$ being measured.

We apply stabilizer update rules \cite{gottesman1997Stabilizer} to a basis of the normalizer of the original stabilizer group. Initially, this is the normalizer group of the original code as well as single qubit operators $\{X(e)\}_{e\in\mathcal{E}}$ on each edge qubit (since the edges qubits are initialized in $\ket{+}^{\otimes|\mathcal{E}|}$ product state). We measure all stabilizers of the deformed code Fig.~\ref{fig:gauging_measurement}.

Suppose $\Lambda$ is one of the stabilizers of the original code or one of the logicals $\overline{X}^{(i)},\overline{Z}^{(i)}$ for $i=1,2,\dots,k-1$. These all commute with $\overline{Z}$. If the $X$-type support of $\Lambda$ on $\mathcal{L}$ is $\mathcal{L}_\Lambda$, then $|\mathcal{L}_\Lambda|$ is even, and $\Lambda$ anticommutes with an even number of checks in $\mathcal{V}$, exactly those on vertices $f(\mathcal{L}_\Lambda)$. Find a perfect matching $\mu(\mathcal{L}_\Lambda)$ of those vertices in the graph $G$ (which exists because $G$ is connected). We see we can update $\Lambda$ to commute with all $\mathcal{V}$ checks by multiplying by the appropriate combination of initial edge stabilizers $X(e)$, namely,
\begin{equation}
\Lambda\rightarrow\Lambda\prod_{e\in\mu(\mathcal{L}_\Lambda)}X(e).
\end{equation}
Thus, by updating $\overline{X}^{(i)},\overline{Z}^{(i)}$ for $i=1,2,\dots,k-1$ this way, there is a symplectic basis for $k-1$ logical qubits in the deformed code. 

The original single-qubit stabilizers $X(e)$ on edge qubits must be updated to form products $X(c\in\mathcal{E})$ that commute with all vertex checks $\mathcal{V}$. This necessitates they form cycles in the graph, i.e.~$c\,G=0$. Therefore, the new $X$-type stabilizers $\mathcal{U}$, defined as elements of a suitable cycle basis of graph $G$, capture a complete basis of these combinations.

Moreover, the product of all checks in $\mathcal{V}$, i.e.~$\mathcal{H}_Z(\mathcal{V})$, equals $\overline{Z}=Z(\mathcal{L})$. This product of checks has no support on edge qubits $\mathcal{E}$ because $\vec1G^\top=0$. Thus, $\overline{Z}$ is in the stabilizer group of the deformed code.
\end{proof}
\noindent We remark that desiderata 0 is technically a little too strong to reach the conclusion of this lemma and thus for the claim in Theorem~\ref{thm:graph_desiderata}. All we really need is $f(\mathcal{L})$ to be a connected subset of vertices in $\mathcal{G}$, i.e.~it is possible to get between any two vertices of $f(\mathcal{L})$ via a path in $\mathcal{G}$. However, this places $f(\mathcal{L})$ entirely within one connected component of $\mathcal{G}$ and other connected components can easily be shown to represent stabilizer states separable from the deformed code. Thus, we consider $\mathcal{G}$ to be connected without loss of generality.

\tocless\subsection{Code distance of the deformed code in auxiliary graph surgery}

Here we show that desideratum 4 of Theorem~\ref{thm:graph_desiderata} is sufficient to guarantee the deformed code preserves the code distance.

\begin{lemma}\label{lem:deformed_code_distance}
Provided $\mathcal{G}$ and an injective port function $f$ are chosen such that $\beta_d(\mathcal{G},\mathrm{im}\;f)\ge1$ where $d$ is the code distance of the original code, the deformed code in auxiliary graph surgery (see Fig.~\ref{fig:gauging_measurement}) has code distance at least $d$.
\end{lemma}
\begin{proof}
Similar proofs exist in \cite{cross2024linear,williamson2024gauging}. Our new contribution is to use relative expansion.

Consider an arbitrary nontrivial logical operator of the deformed code $\overline{L}=L_XL_Z$ where $L_X$ is an $X$-type Pauli and $L_Z$ a $Z$-type Pauli. Let $\overline{L}\equiv\overline{L}'$ indicate the equivalence of $\overline{L}$ and $\overline{L}'$ modulo the checks of the deformed code. We start by finding such an equivalent $\overline{L}'$ with convenient qubit support. Then, we show that the weight of $\overline{L}'$ cannot be reduced below $d$, the original code's distance, by multiplying by any combination of the deformed code's checks. Thus, the weight of $\overline{L}$ cannot be reduced in this way either.

Consider the restriction of $L_Z$ to the qubits $\mathcal{E}$, denoted $Z(u\in\mathcal{E})$. Because $\overline{L}$ commutes with all checks in $\mathcal{U}$, we must have $Nu^\top=0$. By construction of $N$, the column \textit{nullspace} of $N$ equals the column space of $G$, i.e.~$NG=0$. This also means the row space of $N$ equals the row nullspace of $G$. Thus, there exists $v$ such that $u^\top=Gv^\top$ and $u=vG^\top$. As a result, $L_Z$ can be cleaned from $\mathcal{E}$ by multiplying by appropriate checks from $\mathcal{V}$. Explicitly, 
\begin{equation}
\overline{L}\equiv\overline{L}'=\overline{L}\mathcal{H}_Z(v\in\mathcal{V})
\end{equation}
and $\overline{L}'$ has no $Z$-type support on $\mathcal{E}$.

Now the restriction of $\overline{L}'$ to the qubits of the original code must commute with all checks of the original code and is therefore a logical operator $\overline{\Lambda}$ of the original code. That is,
\begin{equation}
\overline{L}'=\overline{\Lambda}X(c\in\mathcal{E})
\end{equation}
for some vector $c$.

Suppose $\overline{\Lambda}$ is a trivial logical operator of the original code, i.e.~a product of original code checks. Then, multiplying by those checks to remove it from the original code, we obtain
\begin{equation}
\overline{L}'\equiv X(c'\in\mathcal{E})
\end{equation}
Of course, this operator must commute with all vertex $Z$ checks $\mathcal{V}$, so $cG=0$. This implies $\overline{L}'$ is a cycle in the graph, so it will be a product of the checks $\mathcal{U}$. This is a contradiction with our assumption that $\overline{L}$ is a nontrivial logical operator of the deformed code.

Therefore, $\overline{\Lambda}$ is a nontrivial logical operator of the original code, and the weight of $\overline{L}'=\overline{\Lambda}X(c\in\mathcal{E})$ is at least $d$ even if the $X$-type support on $\mathcal{E}$ is ignored. To show $\overline{L}'$ cannot be reduced in weight by multiplying by checks from $\mathcal{V}$, we show that $\overline{\Lambda}$ cannot be reduced in weight this way. 

Suppose $\overline{\Lambda}=\Lambda_X\Lambda_Z$ and consider $|\overline{\Lambda}\mathcal{H}_Z(v\in\mathcal{V})|$ for any choice of vector $v$ indicating a set of vertices $\mathcal{V}$. We define the sets of qubits $\mathcal{L}^*:=\mathcal{L}\setminus\supp{\Lambda_X}$ and $\mathcal{R}:=\supp{\overline{\Lambda}}\setminus\mathcal{L}^*$. We also let $w$ be the restriction of $v$ to the vertices $f(\mathcal{L}^*)\subseteq\mathcal{V}$. Thus, $\mathcal{W}:=f^{-1}(\supp{w})$ is a subset of qubits in $\mathcal{L}^*$. The intersection profile of $\mathcal{L}^*$, $\mathcal{W}$, and $\supp{\Lambda_Z}$ defines four subsets $\alpha,\beta,\gamma,\delta\subseteq\mathcal{L}^*$. These sets are illustrated in Fig.~\ref{fig:distance_proof_sets}. From the figure, it is clear that 
\begin{align}\label{eq:d_lower_1}
d&\le|\overline{\Lambda}|=|\gamma|+|\delta|+|\mathcal{R}|,\\
\label{eq:d_lower_2}
d&\le|\overline{\Lambda}\;\overline{Z}|=|\alpha|+|\beta|+|\mathcal{R}|.
\end{align}

\begin{figure}[t]
    \centering
    \includegraphics[width=0.45\textwidth]{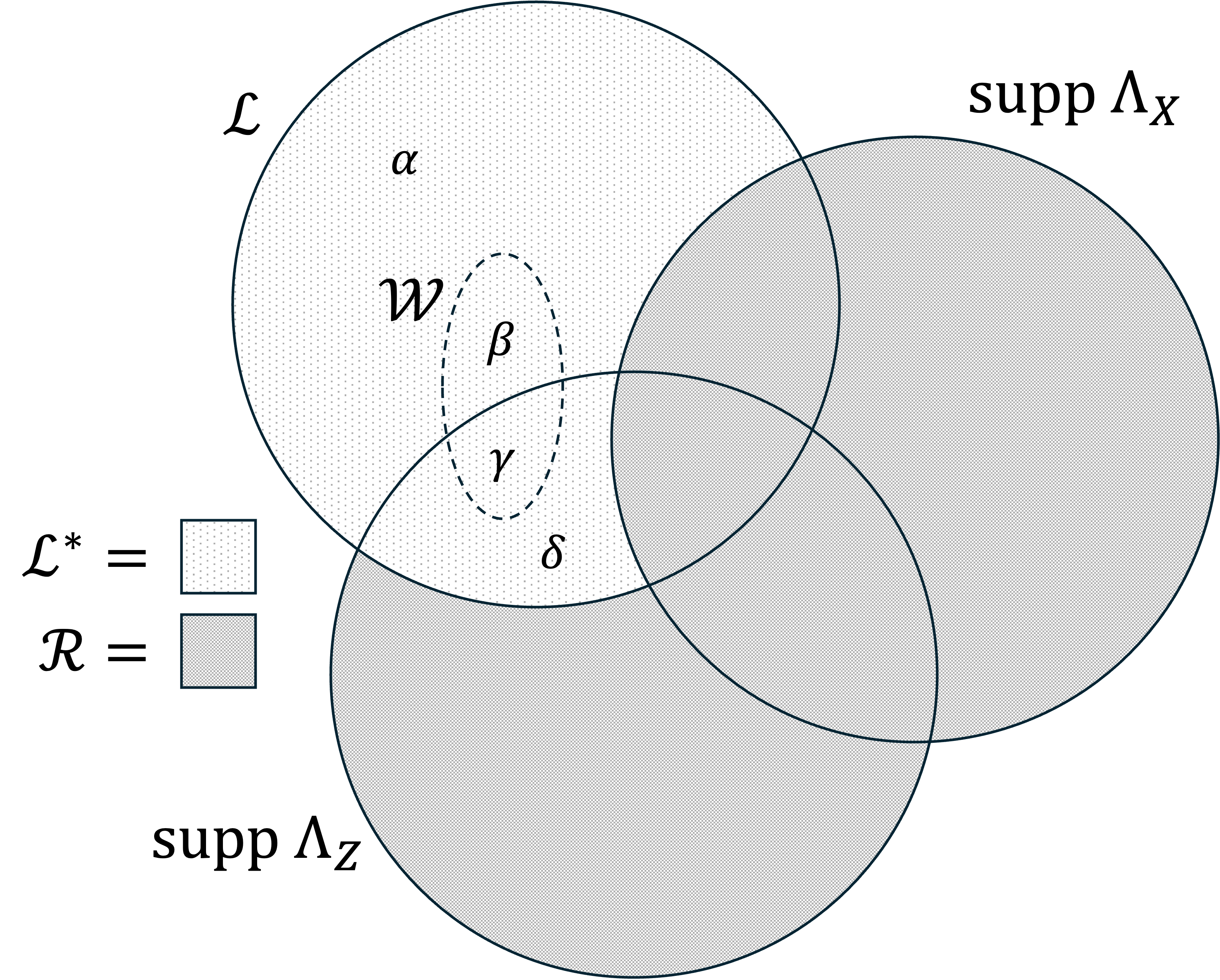}
    \caption{A summary of the qubit sets used in the final part of the proof of Lemma~\ref{lem:deformed_code_distance}, the deformed code distance in auxiliary graph surgery. These are all subsets of the qubits of the original code. Here $\mathcal{L}$ supports the logical operator $\overline{Z}=Z(\mathcal{L})$ being measured. Another logical operator $\overline{\Lambda}=\Lambda_X\Lambda_Z$ generally intersects $\overline{Z}$. The dashed region within $\mathcal{L}^*$ represents $\mathcal{W}$, which defines the four subregions $\alpha,\beta,\gamma,\delta\subseteq\mathcal{L}^*$.}
    \label{fig:distance_proof_sets}
\end{figure}

By assumption the relative expansion $\beta_d(\mathcal{G},f(\mathcal{L}))$ is at least $1$. Because  $f(\mathcal{L^*})\subseteq f(\mathcal{L})$, Lemma~\ref{lem:relative_expansion_relation} implies $\beta_d(\mathcal{G},f(\mathcal{L^*}))\ge\beta_d(\mathcal{G},f(\mathcal{L}))\ge1$. 

We now calculate
\begin{align}
|\overline{\Lambda}\mathcal{H}_Z(v\in\mathcal{V})|&=|\beta|+|\delta|+|\mathcal{R}|+|vG^\top| \nonumber \\
&\ge|\beta|+|\delta|+|\mathcal{R}|+\min(d,|w|,|\mathcal{L}^*|-|w|) \nonumber \\
&=|\beta|+|\delta|+|\mathcal{R}|+\min(d,|\mathcal{W}|,|\mathcal{L}^*|-|\mathcal{W}|) \nonumber \\
&=|\beta|+|\delta|+|\mathcal{R}|+\min(d,|\beta|+|\gamma|,|\alpha|+|\delta|)
\end{align}
where the inequality makes use of relative expansion on $f(\mathcal{L}^*)$. If the minimum evaluates to $d$, then it is immediate that $|\overline{\Lambda}\mathcal{H}_Z(v\in\mathcal{V})|\ge d$. If the minimum evaluates to $|\beta|+|\gamma|$, use Eq.~\eqref{eq:d_lower_1} to conclude $|\overline{\Lambda}\mathcal{H}_Z(v\in\mathcal{V})|\ge d$. If the minimum evaluates to $|\alpha|+|\delta|$, use Eq.~\eqref{eq:d_lower_2} to conclude similarly.
\end{proof}

\tocless\section{A set-valued port function and the general proof of Theorem~\ref{thm:joint_logical_measurement}}\label{app:set-valued_port_function}
Here we generalize the main text to allow the port function to be a set-valued (sometimes called multi-valued) function. Let $\mathcal{Q}$ denote the set of qubits of the original code, and suppose we have auxiliary graph $\mathcal{G}=(\mathcal{V},\mathcal{E})$. The set-valued port function is denoted $f:\mathcal{Q}\rightarrow2^\mathcal{V}$, where $2^\mathcal{V}$ is the power-set of $\mathcal{V}$. Thus, each qubit $q\in\mathcal{Q}$ is connected to a set of vertices $f(q)$, possibly empty. We place two constraints on $f$.
\begin{itemize}
\item For simplicity, we demand that $f$ be injective in the sense that for all $q\neq q'$, $f(q)\cap f(q')=\emptyset$. Thus, each vertex is connected to at most one qubit.  
\item If $\overline{Z}=Z(\mathcal{L})$ is the logical we wish to measure, we require $|f(q)|$ be odd for all $q\in\mathcal{L}$ and even for all $q\in\mathcal{Q}\setminus\mathcal{L}$. This is explained next.
\end{itemize}

Vertex checks are defined almost identically to the main text, Eq.~\eqref{eq:vertex_checks}, with only a notational difference because $f(q)$ is now a set.
\begin{equation}
A_v=\bigg\{
\begin{array}{ll}
Z(q)\prod_{e\ni v}Z(e),& \exists q\in\mathcal{L}, f(q)\ni v\\
\prod_{e\ni v}Z(e),& \mathrm{otherwise}
\end{array}
\end{equation}
The product of all vertex checks is
\begin{equation}
\prod_{v\in\mathcal{V}}A_v=\prod_{q\in\mathcal{Q}}Z(q)^{|f(q)|}
\end{equation}
To measure a logical operator $\overline{Z}=Z(\mathcal{L})$ using gauging, it must equal the product of vertex checks $\overline{Z}=\prod_{v\in\mathcal{V}}A_v$. This leads to the second constraint on the set-valued $f$ we stated above.

We often apply $f$ to a subset of $\mathcal{Q}'\subseteq\mathcal{Q}$ by defining $f(\mathcal{Q}')=\bigcup_{q\in\mathcal{Q}'}f(q)$. Suppose a stabilizer $s$ in the original code has $X$-type support on qubits $\mathcal{L}_s$. Then, because $s$ must commute with the logical $Z(\mathcal{L})$, $|\mathcal{L}_s\cap\mathcal{L}|$ must be even, which in turn implies $|f(\mathcal{L}_s)|$ is even. This means we can create a perfect matching $\mu(\mathcal{L}_s)\subseteq\mathcal{E}$ of vertices $f(\mathcal{L}_s)$ in the auxiliary graph. With this more general definition of $\mu(\mathcal{L}_s)$, the stabilizer $s$ is deformed onto the edge qubits just as in the main text
\begin{equation}
s\rightarrow s\prod_{e\in\mu(\mathcal{L}_{s})}X(e).
\end{equation}

Finally, we define the image of $f$ to be $\mathrm{im}f=f(\mathcal{Q})\subseteq\mathcal{V}$. Just as in the main text, we call this subset of vertices the port.

The more general presentation of $f$ as a set-valued function and the generalized definitions of perfect matching $\mu$ and image $\mathrm{im}f$ leave the statement of graph desiderata in Theorem~\ref{thm:graph_desiderata} largely unchanged. However, the degree of a qubit $q\in\mathcal{Q}$ in the original code now is allowed to increase by more than one. It increases in fact by $|f(q)|$. Thus, we modify desideratum 1 to consist of two parts, the first being the same as before and the second expressing this qubit degree constraint.
\begin{enumerate}[label={1\alph*.}]
\item $\mathcal{G}$ has $O(1)$ vertex degree.
\item For all $q\in\mathcal{Q}$, $|f(q)|=O(1)$.
\end{enumerate}

Desiderata 0-3 including the new desideratum 1 are easily proven using a set-valued $f$. However, desideratum 4 is not obviously true. We provide a proof here, generalizing Lemma~\ref{lem:deformed_code_distance} in Appendix~\ref{app:desiderata_proofs}.

\begin{lemma}\label{lem:generalized_4}
Provided $\mathcal{G}$ and the set-valued port function $f$ are chosen such that $\beta_d(\mathcal{G},\mathrm{im}\;f)\ge1$ where $d$ is the code distance of the original code, the deformed code in auxiliary graph surgery (see Fig.~\ref{fig:gauging_measurement}) has code distance at least $d$.
\end{lemma}
\begin{proof}
We follow the proof of Lemma~\ref{lem:deformed_code_distance}. In fact, the proof is identical to the point that we obtain a nontrivial logical operator of the deformed code $\overline{L}'=\overline{\Lambda}X(c\in\mathcal{E})$, where $\overline{\Lambda}$ is a nontrivial logical operator of the original code, and must show that $\overline{\Lambda}$ cannot be reduced in weight by multiplying by vertex checks.

Let us therefore begin from that point with $\overline{\Lambda}=\Lambda_X\Lambda_Z$ and lower bound the weight of $\overline{\Lambda}\mathcal{H}_Z(v\in\mathcal{V})$ for any vector $v$ indicating a subset of vertices. We use $\mathcal{Q}$ to denote the set of original code qubits, $\mathcal{L}\subseteq\mathcal{Q}$ the support of the $Z$-type logical we are measuring, and $\mathcal{P}$ to denote $\mathrm{im}f=f(\mathcal{Q})$. We have by assumption $\beta_d(\mathcal{G},\mathcal{P})\ge1$.

Since every vertex in the port is connected to a single qubit, it is sensible to define $f^{-1}:\mathcal{P}\rightarrow\mathcal{Q}$, and when applied to a subset $\mathcal{U}\subseteq\mathcal{P}$ we define $f^{-1}(\mathcal{U})=\{q\in\mathcal{Q}:|f(q)\cap\mathcal{U}|\text{ is odd}\}$. The upshot of this notation is that the product of all vertex checks from $\mathcal{U}$, i.e.~$\mathcal{H}_Z(\mathcal{U})$, is supported on the original code qubits as $Z(f^{-1}(\mathcal{U}))$ (in addition to its support on edge qubits).

Now, reduce the port by excluding vertices from $f(\mathrm{supp}\;{\Lambda_X})$, or $\mathcal{P}^*=\mathcal{P}\setminus f(\mathrm{supp}\;{\Lambda_X})$. By Lemma~\ref{lem:relative_expansion_relation} we have $\beta_d(\mathcal{G},\mathcal{P}^*)\ge1$. We let $w$ be the restriction of $v$ onto $\mathcal{P}^*$ and define $\mathcal{W}=f^{-1}(\mathrm{supp}\;w)$. When multiplying $\overline{\Lambda}=\Lambda_X\Lambda_Z$ by the checks $\mathcal{H}_Z(v\in\mathcal{V})$ only checks $\mathcal{H}_Z(w\in\mathcal{V})$ modify its support on $\mathcal{Q}$ because checks connected to $\mathrm{supp}\;{\Lambda_X}$ cannot eliminate the $X$-type support on those qubits. The $Z$-type support on $\mathcal{Q}$ is modified by exactly a factor of $Z(\mathcal{W})$.

\begin{figure}
    \centering
    \includegraphics[width=0.9\linewidth]{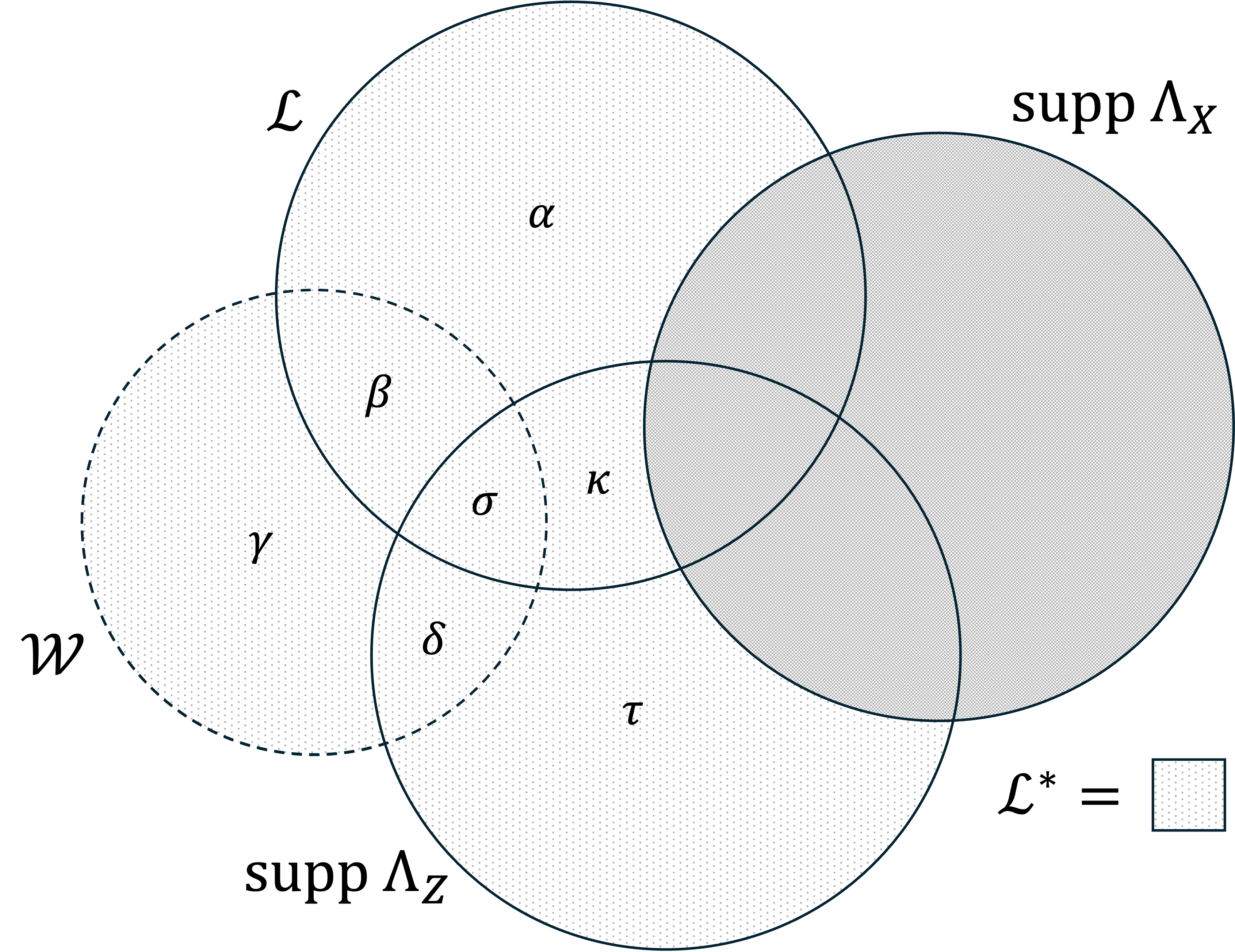}
    \caption{A summary of the qubit sets used in the final part of the proof of Lemma~\ref{lem:generalized_4}. These are all subsets of the qubits of the original code. In particular, we define a set of qubits $\mathcal{L}^*=(\mathcal{L}\cup\mathrm{supp}\;\Lambda_Z\cup\mathcal{W})\setminus\mathrm{supp}\;\Lambda_X\subseteq\mathcal{Q}$ such that $f(\mathcal{L}^*)\subseteq\mathcal{P}^*$.}
    \label{fig:distance_proof_sets_2}
\end{figure}

We summarize the previous discussion in Fig.~\ref{fig:distance_proof_sets_2}, in which subsets of $\mathcal{Q}$ are illustrated in relation to one another. Using the notation from that figure, we see that
\begin{align}\label{eq:d_gen_lower_1}
d&\le|\overline{\Lambda}|=|\delta|+|\sigma|+|\kappa|+|\tau|+|\Lambda_X|\\\label{eq:d_gen_lower_2}
d&\le|\overline{\Lambda}\;\overline{Z}|=|\alpha|+|\beta|+|\delta|+|\tau|+|\Lambda_X|
\end{align}
We also observe that
\begin{equation}
|w|\ge|\mathcal{W}|=|\beta|+|\gamma|+|\delta|+|\sigma|
\end{equation}
because there must be at least one vertex in $w$ for each qubit in $\mathcal{W}$. Similarly,
\begin{equation}
|\mathcal{P}^*|-|w|\ge|\alpha|+|\gamma|+|\delta|+|\kappa|
\end{equation}
This is because $f(q)$ and $f(q)\cap\mathrm{supp}\;w$ must differ by at least one for every qubit $q$ in $\alpha\cup\gamma\cup\delta\cup\kappa$. Specifically, $\mathcal{P}^*$ contains an even number of vertices connected to $q\in\gamma\cup\delta$, while $\mathrm{supp}\;w$ contains an odd number. Likewise, $\mathcal{P}^*$ contains an odd number of vertices connected to $q\in\alpha\cup\kappa$, while $\mathrm{supp}\;w$ contains an even number.

Using expansion relative to $\mathcal{P}^*$, we now calculate
\begin{align}
|\overline{\Lambda}\mathcal{H}_Z(v\in\mathcal{V})|&=|\beta|+|\gamma|+|\kappa|+|\tau|+|\Lambda_X|+|vG^\top|\\
&\ge |\beta|+|\gamma|+|\kappa|+|\tau|+|\Lambda_X| \nonumber \\ & \qquad +\min(d,|w|,|\mathcal{P}^*|-|w|)\\
&\ge |\beta|+|\gamma|+|\kappa|+|\tau|+|\Lambda_X| \nonumber \\ & \qquad +\min(d,|\beta|+|\gamma|+|\delta|+|\sigma|, \nonumber \\ & \qquad \qquad  |\alpha|+|\gamma|+|\delta|+|\kappa|).
\end{align}
No matter which of the three arguments the minimum evaluates to, the result is at least $d$, making use of Eqs.~\eqref{eq:d_gen_lower_1} and \eqref{eq:d_gen_lower_2} for arguments two and three, respectively.
\end{proof}

Finally, let us complete the proof of Theorem~\ref{thm:joint_logical_measurement}.

\begin{reptheorem}{thm:joint_logical_measurement}
Consider a set of nontrivial, sparsely overlapping logical operators $\overline{Z}_0,\overline{Z}_1,\dots,\overline{Z}_{t-1}$ that can all be made $Z$-type simultaneously by applying single-qubit Cliffords. Provided $t$ auxiliary graphs that satisfy the graph desiderata of Theorem~\ref{thm:graph_desiderata} to measure these $t$ logical operators, there exists an auxiliary graph to measure the product $\overline{Z}_0\overline{Z}_1\dots\overline{Z}_{t-1}$ satisfying the desiderata of Theorem~\ref{thm:graph_desiderata}. In particular, if the individual deformed codes are LDPC with distance $d$, the deformed code for the joint measurement is LDPC with check weights and qubit degrees independent of $t$ and has distance $d$.
\end{reptheorem}
\begin{proof}
Denote the qubit supports of the individual logical operators by $\mathcal{L}_i$ for $i=0,1,\dots,t-1$ and the set of all qubits in the original code by $\mathcal{Q}$. Denote the auxiliary graphs used to measure the individual logical operators by $\mathcal{G}_i=(\mathcal{V}_i,\mathcal{E}_i)$, the set-valued port functions by $f_i:\mathcal{Q}\rightarrow\mathcal{V}_i$, and the ports by $\mathcal{P}_i=f_i(\mathcal{L}_i)$. 

Even if the original port functions $f_i$ were not set-valued, to deal properly with potentially overlapping logicals $\overline{Z}_i$, we create a port function $f$ for the final graph $\mathcal{G}$ that is set-valued. Namely, $f(q)=\bigcup_{i=0}^{t-1}f_i(q)$. This is the only substantive difference with the proof in the main text, which considered only the case of non-overlapping logicals.

To deal with irreducible logical operators and ensure connected subsets of the port functions, we proceed as in the main text proof to add edges to the graphs $\mathcal{G}_i$ and pick port subsets $\mathcal{P}'_i\subseteq\mathcal{P}_i$ that are connected. This process works for set-valued port functions $f_i$ just as well, so we do not rewrite it here.

Also following the proof in the main text, we create adapters to join all the graphs together. For each $i=0,1,\dots,t-2$, we create an adapter $\mathcal{A}_i$ between $\mathcal{P}_i^*\subseteq\mathcal{P}'_i$ and $\mathcal{P}_{i+1}^*\subseteq\mathcal{P}'_{i+1}$. These subsets can always be chosen so that $|\mathcal{P}_i^*|=|\mathcal{P}_{i+1}^*|\ge d$ and both $\mathcal{P}_i^*$ and $\mathcal{P}_{i+1}^*$ induce connected subgraphs in their respective auxiliary graphs $\mathcal{G}_i$ and $\mathcal{G}_{i+1}$. Iterative application of Lemma~\ref{lem:skip_adapter} creates the adapted graph
\begin{equation}
\mathcal{G}=(\mathcal{V},\mathcal{E})=\mathcal{G}_0\sim_{\mathcal{A}_0}\mathcal{G}_1\sim_{\mathcal{A}_1}\mathcal{G}_2\sim_{\mathcal{A}_2}\dots\sim_{\mathcal{A}_{t-2}}\mathcal{G}_{t-1}.
\end{equation}

The vertex set of $\mathcal{G}$ is $\mathcal{V}=\bigcup_{i=0}^{t-1}\mathcal{V}_i$ and the port function $f:\mathcal{Q}\rightarrow2^{\mathcal{V}}$ is defined as we stated above: $f(q)=\bigcup_{i=0}^{t-1}f_i(q)$. Note that the port is simply the union of the ports of the original graphs $\mathcal{P}=f(\mathcal{Q})=\bigcup_{i=0}^{t-1}\mathcal{P}_i$. 

The adapted graph $\mathcal{G}$ and set-valued port function $f$ satisfy the graph desiderata of Theorem~\ref{thm:graph_desiderata} (including the two-part desideratum 1 described earlier in this section). For desideratum 0, note that $\overline{Z}=\bigcup_{i=0}^{t-1}\overline{Z}_i$ is supported exactly on qubits $\mathcal{L}=\mathcal{L}_0\triangle\mathcal{L}_1\triangle\dots\triangle\mathcal{L}_{t-1}$, where $\triangle$ denotes the symmetric difference of sets. Moreover, $|f(q)|$ is odd if and only if $q\in\mathcal{L}$, so that the product of all vertex checks is indeed $\overline{Z}=Z(\mathcal{L})$.

Desideratum 1a is satisfied because we do not increase vertex degree by more than a constant, and 1b is satisfied because the logicals $\overline{Z}_i$ are sparsely overlapping. Desideratum 2 is inherited from the original graphs. Specifically, suppose stabilizer $s$ has sparse perfect matchings $\mu_i$ in each of the original graphs $\mathcal{G}_i$. Only a constant number of these matchings are non-empty because $s$ is constant size and the logicals are sparsely overlapping. Then, the union of those perfect matchings is sparse in the adapted graph.

We used Lemma~\ref{lem:skip_adapter} to construct adapters that guarantee desideratum 3 is satisfied. Lemma~\ref{lem:adapted_expansion} implies $\beta_d(\mathcal{G},\mathcal{P})\ge1$ and thus we have desideratum 4.
\end{proof}

\tocless\section{Thickening to increase relative expansion}\label{app:relative_exp_lemma}

Recall that thickening a graph means taking its Cartesian product with a path graph, see Definition~\ref{def:thickening}. Here we show that thickening a graph can increase its relative expansion.
\begin{lemma}[Relative Expansion Lemma]\label{lem:expansion_lemma}
For any connected graph $\mathcal{G}_0$, subset of vertices $\mathcal{U}$, and integer $t>0$, the $L\ge1/\beta_t(\mathcal{G}_0,\mathcal{U})$ times thickened graph $\mathcal{G}_0^{(L)}$ has relative expansion $\beta_t(\mathcal{G}_0^{(L)},\mathcal{U}')\ge1$ for all $\mathcal{U}'=\mathcal{U}\times\{l\}$, $l=0,1,\dots,L-1$ (i.e.~$\mathcal{U}'$ is a subset of vertices in the thickened graph identified with $\mathcal{U}$ in any one copy of $\mathcal{G}_0$).
\end{lemma}
\begin{proof}
This follows the proof of Lemma 5 in Ref.~\cite{cross2024linear} but suitably generalized to relative expansion. Let $G\in\mathbb{F}_2^{m\times n}$ be the incidence matrix of $\mathcal{G}_0$ and $G_L=\mathbb{F}_2^{m_L\times n_L}$ be the incidence matrix of the thickened graph $\mathcal{G}^{(L)}$. Explicitly, $m_L=mL+n(L-1)$, $n_L=nL$, and
\begin{equation}
G_L=\left(\begin{array}{c}I_L\otimes G\\H_R\otimes I_n\end{array}\right),
\end{equation}
where $H_R$ is the $L-1\times L$ check matrix of the repetition code, i.e.~$H_C$ missing its last row.

Let $v_j\in\mathbb{F}_2^n$ for $j=0,1,\dots,L-1$ be vectors indicating choices of vertices from each copy of $\mathcal{G}_0$. Also define $r=\mathrm{argmin}_{j}\min(t,|u_j|,n-|u_j|)$ where $u_j$ is $v_j$ restricted to vertices in $\mathcal{U}$. Making judicious use of the triangle inequality and the assumption that $L\beta_t(\mathcal{G}_0,\mathcal{U})\ge1$, we calculate
\begin{align}
|(v_0\text{\space}v_1\dots v_{L-1})G_L^\top|&=\sum_{j=1}^{L-1}|v_{j-1}+v_j|+\sum_{j=0}^{L-1}|v_jG^\top| \nonumber \\
&\ge|v_r+v_l| \nonumber \\ & \; + L\beta_t(\mathcal{G}_0,\mathcal{U})\min(t,|u_r|,n-|u_r|) \nonumber \\
&\ge|v_r+v_l|+\min(t,|u_r|,n-|u_r|) \nonumber \\
&\ge|u_r+u_l|+\min(t,|u_r|,n-|u_r|) \nonumber \\
&\ge\min(t,|u_l|,n-|u_l|).
\end{align}
Because $u_l$ is exactly the restriction of $(v_0\text{\space}v_1\dots v_{L-1})$ to the vertices $\mathcal{U}'=\mathcal{U}\times\{l\}$ of the thickened graph, this implies $\beta_t(\mathcal{G}_0^{(L)},\mathcal{U}')\ge1$. 
\end{proof}

Because of the Relative Expansion Lemma, we can satisfy Theorem~\ref{thm:graph_desiderata}, desideratum 4 by sufficiently thickening any initial graph $\mathcal{G}_0=(\mathcal{V}_0,\mathcal{E}_0)$ and choosing a port function $f:\mathcal{L}\rightarrow\mathcal{V}_0\times\{0,1,\dots,L-1\}$ to be injective on $\mathcal{V}_0\times\{l\}$ for any $l=0,1,\dots,L-1$. This discussion clarifies the role of thickening as it is used in Ref.~\cite{cross2024linear} in cases where the initial graph $\mathcal{G}_0$ (though it is more generally a hypergraph in that reference) is not sufficiently expanding.

\tocless\section{Support lemma for irreducible logical operators}\label{app:supportlemma}

Here we record a helpful support lemma formerly appearing in Ref.~\cite{cowtan2024css}, Appendix~H, and Ref.~\cite{cross2024linear}.

\begin{lemma}\cite{cross2024linear}\label{lem:supportlemma}
Let $H=[H_X|H_Z]$ be the (symplectic) parity check matrix of an $n$-qubit stabilizer code, $\overline{Z}$ be an irreducible logical $Z$-type operator, and $H'_X$ denote the sub-matrix of $H_X$ restricted to qubits in the support of $\overline{Z}$. Then, $H'_Xv^\top=0$ implies $v=0$ or $v=\vec1$. Equivalently, $H'_X$ is a check matrix of the classical repetition code.
\end{lemma}
\begin{proof}
If $v$ is not $0$ or $\vec1$, then it implies a logical $Z$-type operator $Z(v)$ of the stabilizer code exists that is entirely supported within the support of $\overline{Z}$. This contradicts the definition of $\overline{Z}$ being irreducible.
\end{proof}

This lemma has a simple corollary.
\begin{corollary}\label{cor:supportlemma}
If $\overline{Z}$ is an irreducible $Z$-type logical operator of a stabilizer code, and $P_X$ is an $X$-type operator (not necessarily logical) commuting with $\overline{Z}$, then there is a stabilizer $S$ of the code such that $SP_X$ does not have $X$-type support overlapping $\overline{Z}$. If the code is CSS, then $S$ can be chosen to be $X$-type and $SP_X$ does not overlap $\overline{Z}$ at all.
\end{corollary}
\begin{proof}
Because it commutes with $\overline{Z}$, $P_X$ must overlap $\overline{Z}$ on an even number of qubits. Lemma~\ref{lem:supportlemma} implies the $X$ check matrix of the code restricted to $\supp{\overline{Z}}$ generates all even weight $X$ operators. Thus, we can find a stabilizer $S$ performing as claimed.
\end{proof}

\tocless\section{$\mathsf{SkipTree}$ for the full-rank check matrix of the repetition code}\label{app:skip_fullrank}

We present a $\mathsf{SkipTree}$ algorithm variant specifically for solving $TGP=H_R$ for sparse matrix $T$ and permutation matrix $P$. The need for this is illustrated by a simple example. If $G=H_R$ is the incidence matrix of the path graph, applying Algorithm~\ref{alg:skiptree} creates a $(2,2)$-sparse matrix $T$. Clearly, however, the sparsest possible $T$ matrix is $(1,1)$-sparse, i.e.~$T=P=I$.

Because Algorithm~\ref{alg:skiptree} is built to solve $TGP=H_C$ it always leaves some nodes unlabeled as it moves down the tree so that it can eventually find a way back to the root. Instead, when solving $TGP=H_R$ returning to the root is unnecessary, and only some branches of the spanning tree must be explored via the ``skipping" behavior of Algorithm~\ref{alg:skiptree}. This motivates the addition of a flag, $\mathsf{skip}\in\{\mathsf{True},\mathsf{False}\}$, that can toggle the skipping behavior on and off. Adding this, we obtain Algorithm~\ref{alg:skiptree_HR} that produces sparser solutions for $T$. For example, it solves the case of the path graph optimally.

\begin{algorithm}[H]
\caption{\label{alg:skiptree_HR}Given connected graph $G\in\mathbb{F}_2^{\;m\times n}$, find $T\in\mathbb{F}_2^{\;n-1\times m}$ and permutation $P\in\mathbb{F}_2^{\;n\times n}$ such that $TGP=H_R$. Both $T$ and $P$ have $O(n)$ nonzero entries and can be constructed and returned as sparse matrices.}
\begin{algorithmic}[1]
    \Procedure{$\mathsf{SkipTreeHR}$}{$G$}
        \State $S\leftarrow$ a spanning tree of $G$ \protect\Comment{has incidence matrix $S_I\in\mathbb{F}_2^{\;n-1\times n}$ that we do not need to store}
        \State Index $\leftarrow0$ 
        \State $\mathrm{Label}\leftarrow$ empty list of length $n$
        \Procedure{$\mathsf{LabelFirst}$}{$v$,$\;\mathsf{skip}$}
            \State $\mathrm{Label}[\mathrm{Index}]\leftarrow v$
            \State Index $\leftarrow$ Index + 1
            \For{each child of vertex $v$ in $S$} \protect\Comment Recall the youngest child is the last in the for-loop.
                \If{child is youngest and $\mathsf{skip}=\mathsf{False}$}
                    \State $\mathsf{LabelFirst}(\mathrm{child},\mathsf{skip}=\mathsf{False})$
                \Else
                    \State $\mathsf{LabelLast}(\mathrm{child})$
                \EndIf
            \EndFor
        \EndProcedure
        
        \Procedure{$\mathsf{LabelLast}$}{$v$}
            \For{each child of vertex $v$ in $S$} 
                \State $\mathsf{LabelFirst}(\mathrm{child},\mathsf{skip}=\mathsf{True})$
            \EndFor
            \State $\mathrm{Label}[\mathrm{Index}]\leftarrow v$
            \State Index $\leftarrow$ Index + 1
        \EndProcedure

        \State $\mathsf{LabelFirst}(0,\mathsf{skip}=\mathsf{False})$ \protect\Comment{Root is 0. After this line, $\mathrm{Label}[l]=v$ means vertex $v$ is labeled $l$.}
        \State $P \leftarrow n\times n$ matrix with $P_{vl}=1$ \text{iff} $\mathrm{Label}[l]=v$.
        \State $\tilde T\leftarrow$ matrix with $n-1$ rows and $n-1$ columns.
        \State $\tilde T_{le}=1$ \text{iff} edge $e$ is part of the shortest path in $S$ from $\mathrm{Label}[l]$ to $\mathrm{Label}[l+1]$. \protect\Comment{now $\tilde TS_I=H_RP^\top$}
        \State Add zero columns to $\tilde T$, obtaining $T$ so that $TG=\tilde TS_I$.
        \State Return $T$, $P$.

    \EndProcedure
\end{algorithmic}
\end{algorithm}

\newpage
\tocless\section{Proof of Theorem~\ref{thm:expansionless_joint_toric}}\label{app:proof_expansionless_joint_toric}

\begin{reptheorem}{thm:expansionless_joint_toric}
Let $\overline{Z}_r$ and $\overline{X}_r$ be arbitrary non-overlapping logical operators in a distance $d_r$ quantum LDPC code, referred to as the right code. Consider another distance $d_l\ge\max(|\overline{Z}_r|,|\overline{X}_r|)$ quantum LDPC code, the left code, encoding just two logical qubits and possessing two non-overlapping, weight $d_l$ logical operators $\overline{Z}_l$ and $\overline{X}_l$. Suppose the weight of $\overline{Z}_l\overline{X}_l$ cannot be reduced to less than $2d_l$ by multiplying by stabilizers and logical operators of the left code other than $\overline{Z}_l$, $\overline{X}_l$, and $\overline{Z}_l\overline{X}_l$. The toric code is an example of such a left code. Then, we can construct two auxiliary graphs, each of size $O(d_l\log^3d_l)$ to measure $\overline{Z}_l\overline{Z}_r$ and $\overline{X}_l\overline{X}_r$, and only those logical operators, simultaneously. Moreover, the deformed code is LDPC and has distance at least $d_r$.
\end{reptheorem}
\begin{proof}
We being by noting $\overline{Z}_l$, $\overline{X}_l$, $\overline{Z}_r$, and $\overline{X}_r$ may be any logical operators with the assumed properties -- in particular, $\overline{Z}_r$ and $\overline{X}_r$ may even be equal modulo stabilizers in code $r$ -- but they are suggestively labeled with $X$ and $Z$ because we can assume without loss of generality they are $Z$-type and $X$-type by simply applying single-qubit Cliffords. We create two graphs, $\mathcal{G}_Z=(\mathcal{V}_Z,\mathcal{E}_Z)$ with $Z$-type vertex checks to measure $\overline{Z}_l\overline{Z}_r$ and $\mathcal{G}_X=(\mathcal{V}_X,\mathcal{E}_X)$ with $X$-type vertex checks to measure $\overline{X}_l\overline{X}_r$. Similarly to the proof of Theorem~\ref{thm:expansionless_joint} (see (ii) in that proof), we assume for both graphs that the port set of vertices for the left code is a superset of the port set of vertices for the right code. This implies that for all vectors $v_z$ and $v_x$ indicating subsets of $\mathcal{V}_Z$ and $\mathcal{V}_X$, respectively, we have
\begin{equation}\label{eq:port_fucntions_toric}
|v_zF_l^Z|\ge|v_zF_r^Z|,\quad|v_xF_l^X|\ge|v_xF_r^X|.
\end{equation}
We can build the graphs to satisfy graph desiderata 0-3 from Theorem~\ref{thm:graph_desiderata} using the techniques of Section~\ref{sec-gauging-meas}. This is the same process described in a bit more detail in the proof of Theorem~\ref{thm:expansionless_joint}. Each graph has $d_l$ port vertices and a total of $O(d_l\log^3d_l)$ vertices and edges once the cycles are sparsified using thickening to satisfy desideratum 3.

\begin{figure*}[t]
    \centering
    \includegraphics[width=0.8\textwidth]{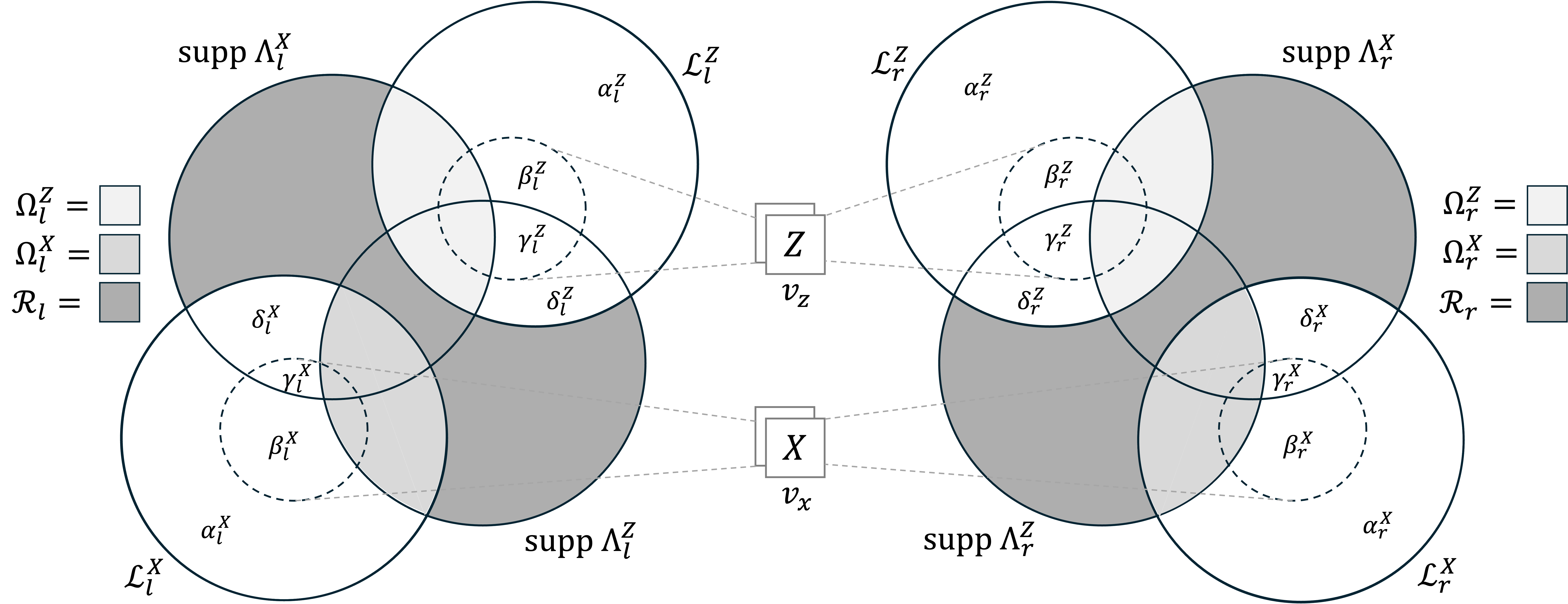}
    \caption{Notation for the proof of Theorem~\ref{thm:expansionless_joint_toric}. On the left (right), all sets shown are subsets of $l$ ($r$) code qubits. Vertex checks $v_z$ ($v_x$) are connected to sets of qubits $\supp v_zF^Z_l$ and $\supp v_zF^Z_r$ ($\supp v_xF^X_l$ and $\supp v_xF^X_r$) indicated by dashed circles.}
    \label{fig:expansionless_joint_toric_venn}
\end{figure*}

Now, we move on to prove a lower bound on the deformed code distance. Suppose we have an arbitrary nontrivial logical operator $\overline{\Lambda}=\overline{\Lambda}_l\overline{\Lambda}_r$ which is a product of logicals $\overline{\Lambda}_l=\overline{\Lambda}^X_l\overline{\Lambda}^Z_l$ and $\overline{\Lambda}_r=\overline{\Lambda}^X_r\overline{\Lambda}^Z_r$ from the original left and right codes. Also, we suppose $\overline{\Lambda}$ commutes with the operators being measured, $\overline{Z}_l\overline{Z}_r$ and $\overline{X}_l\overline{X}_r$. In the deformed code, $\overline{\Lambda}$ may pick up support on the edge qubits of both auxiliary graphs, which we ignore. We seek to bound the weight of $\overline{\Lambda}\mathcal{H}_Z(v_z\in\mathcal{V}_Z)\mathcal{H}_X(v_x\in\mathcal{V}_X)$ for all vectors $v_z,v_x$ indicating choices of vertices from the graphs $\mathcal{G}_Z,\mathcal{G}_X$. We establish notation succinctly in Fig.~\ref{fig:expansionless_joint_toric_venn} in a manner similar to Fig.~\ref{fig:expansionless_joint_venn}. Using that notation, the objective is to lower bound
\begin{align}\label{eq:bound_this_toric}
&|\overline{\Lambda}\mathcal{H}_Z(v_z\in\mathcal{V}_Z)\mathcal{H}_X(v_x\in\mathcal{V}_X)| \nonumber \\ & \quad \ge \sum_{s=l,r}\bigg(|\mathcal{R}_s|+\sum_{t=Z,X}|\beta^t_s|+|\delta^t_s|+|\Omega^t_s|\bigg).
\end{align}

Suppose, without loss of generality, the logical Pauli operators of the two logical qubits in left code are $\overline{Z}_l=\overline{Z}_l^{(1)}$, $\overline{X}_l^{(1)}$, $\overline{Z}_l^{(2)}$, and $\overline{X}_l=\overline{X}_l^{(2)}$. Then, we can assume $\overline{\Lambda}_l$ is in the group $\langle\overline{X}_l^{(1)},\overline{Z}_l^{(2)},S_l\rangle$, where $S_l$ is the stabilizer group of the left code. If this is not initially the case, multiply $\overline{\Lambda}$ by $\mathcal{H}_Z(\vec1\in\mathcal{V}_Z)=\overline{Z}_l\overline{Z}_r$ or $\mathcal{H}_X(\vec1\in\mathcal{V}_X)=\overline{X}_l\overline{X}_r$ or both to make it so, while preserving the problem of bounding Eq.~\eqref{eq:bound_this_toric} for all vectors $v_z,v_x$. This allows us to apply the theorem assumption and assert
\begin{equation}\label{eq:logical_operator_bound}
2d_l\le|\overline{\Lambda}_l\overline{Z}_l\overline{X}_l|=|\overline{Z}_l|-|\gamma_l^Z|-|\delta_l^Z|+|\overline{X}_l|-|\gamma_l^X|-|\delta_l^X|+|\mathcal{R}_l|.
\end{equation}
Also, because $\overline{\Lambda}_l$ is in $\langle\overline{X}_l^{(1)},\overline{Z}_l^{(2)},S_l\rangle$, it necessitates that $\overline{\Lambda}_r$ is a nontrivial logical operator, because $\overline{\Lambda}=\overline{\Lambda}_l\overline{\Lambda}_r$ is both nontrivial and commutes with $\overline{Z}_l\overline{Z}_r$ and $\overline{X}_l\overline{X}_r$. Therefore, we also have
\begin{equation}\label{eq:logical_operator_bound_2}
d_r\le|\overline{\Lambda}_r|=|\gamma_r^Z|+|\delta_r^Z|+|\Omega_r^Z|+|\gamma_r^X|+|\delta_r^X|+|\Omega_r^X|+|\mathcal{R}_r|.
\end{equation}

In addition, we created port functions to satisfy Eq.~\eqref{eq:port_fucntions_toric}. Translating that equation into the notation of Fig.~\ref{fig:expansionless_joint_toric_venn}, we obtain
\begin{align}\label{eq:port_containment}
|\beta_l^Z|+|\gamma_l^Z|+|\Omega_l^Z|&\ge|\beta_r^Z|+|\gamma_r^Z|,\\\nonumber
|\beta_l^X|+|\gamma_l^X|+|\Omega_l^X|&\ge|\beta_r^X|+|\gamma_r^X|.
\end{align}

Starting from Eq.~\eqref{eq:bound_this_toric}, we calculate
\begin{align*}
&|\overline{\Lambda}\mathcal{H}_Z(v_z\in\mathcal{V}_Z)\mathcal{H}_X(v_x\in\mathcal{V}_X)| \nonumber \\
&\ge|\mathcal{R}_l|+|\mathcal{R}_r|+\sum_{t=Z,X}|\delta_l^t|-|\gamma_l^t|+|\beta_r^t|+|\gamma_r^t| \nonumber \\
& \qquad +|\beta_r^t|+|\delta_r^t|+|\Omega_r^t|&&\text{[Eq.~\eqref{eq:port_containment}]}\\
&\ge d_r+|\mathcal{R}_l|+\sum_{t=Z,X}2|\beta_r^t|-|\gamma_l^t|+|\delta_l^t|&&\text{[Eq.~\eqref{eq:logical_operator_bound_2}]}\\
&\ge d_r+2d_l-(|\overline{Z}_l|+|\overline{X}_l|)+\sum_{t=Z,X}2|\beta_r^t|+2|\delta_l^t|&&\text{[Eq.~\eqref{eq:logical_operator_bound}]}\\
&\ge d_r+2d_l-(|\overline{Z}_l|+|\overline{X}_l|).
\end{align*}
Use $|\overline{Z}_l|=|\overline{X}_l|=d_l$ to conclude the lower bound on the deformed code distance is $d_r$.
\end{proof}

\tocless\section{Proofs on the toric code adapter}\label{app:toric_adapter_proofs}

\tocless\subsection{Proof that merged code has $k$ logical qubits}\label{app:toric_adapter_kqubits}

\begin{replemma}{lem:gmerge-k-logical-qubits}
If the original LDPC code encodes $k$ qubits, then the merged code encodes $k$ qubits.
\end{replemma}

\begin{proof}
Let the qubits in the rest of the original LDPC code (i.e. the complement of qubit sets $\lz$ and $\lx$, supporting logicals $\bar{Z}^{(c)}$ and $\bar{X}^{(t)}$, each corresponding to a different logical qubit) be labeled by $\mathcal{R}$. 
Note that irreducible logicals $\bar{Z}^{(c)}$ and $\bar{X}^{(t)}$ can be assumed to be disjoint without loss of generality, as any overlapping support between the two can be cleaned by stabilizer multiplication (see Corollary~\ref{cor:supportlemma} for proof).

We can partition the check matrix $H_X$ from the original LDPC code into three blocks: $H_X|_{\supp{(\bar{Z}^{(c)})}}$, the support of the original code $X$-checks in the overlap of the specific $\bar{Z}^{(c)}$ logical, given by some binary matrix $A_Z$, and similarly, $H_X|_{\supp{(\bar{X}^{(t)})}}$, the original code $X$-checks in the support of the $\bar{X}^{(t)}$ logical, given by binary matrix $C_Z$ and support on the remaining physical qubits $\mathcal{R}$ by $B_Z$. Overall the $X$-check matrix of the original code is given by $H_X = \begin{bmatrix} B_Z & A_Z & C_Z\end{bmatrix}$.

The ancilla system consists of three parts - qubits belonging to set $\mathcal{E}_Z$, initialized in single-qubit $\ket{+}$ states, qubits belonging to set $\mathcal{E}_X$, initialized in single-qubit $\ket{0}$ states and the toric code ancilla state. Initializing qubits in single-qubit $\ket{+}$ or $\ket{0}$ states is equivalent to adding single-qubit $X$ and $Z$ checks to the code.
Qubits in sets $\qz_{0}$ are initialized with single qubit $X$ stabilizers, and qubits in sets $\qx_{0}$ are initialized with single qubit $Z$ stabilizers. 
The toric code ancilla is initialized in a specific eigenstate of the toric code - namely the simultaneous $+1$ eigenspace of the logicals $\bar{Z}^{(2)}$ and $\bar{X}^{(1)}$. 

The number of physical qubits $n^{\prime}$ in the system, including the physical qubits from the original code and ancilla qubits 
\begin{align}
    n^{\prime} &= n + |\mathcal{E}_X| + |\mathcal{E}_Z| + d|\mathcal{Q}^X_{0}| + d|\mathcal{Q}^Z_{0}| \nonumber \\
    &= n + |\mathcal{E}_X| + |\mathcal{E}_Z| + d|\mathcal{V}_X| + d|\mathcal{V}_Z|
\end{align}
Here we use the fact that there is a one-to-one correspondence between physical qubits introduced and the new checks in $\mathcal{V}^X$ and $\mathcal{V}^Z$, up to a permutation.

It is trivially true that the single qubit stabilizers are linearly independent of one another, due to their distinct qubit support. Further since the count of these stabilizers is exactly equal to the count of ancilla qubits, they do not add any logical qubits to the code. 

Next, we deform the code by measuring stabilizers of the new code. In every layer, a $X$-check is introduced for each qubit in the support of $\bar{X}^{(t)}$, which means the number of $X$ checks in each layer, $|\cx_{i}| = |\lx| (= d)$ for $i = 0,...,d-1$. Since there is a one-to-one correspondence between qubits in $\lx$ and checks in port at $\mathcal{V}^{X}$ by injective map $f$ (set up in section \ref{sec-gauging-meas}), we have $|\cx_{i}| = |\mathcal{L}_{X}| = |\mathcal{V}^{X}|$.

The desired new code after code deformation has $X$-checks given by
\begin{widetext}
    \begin{align}
        H_{X}^{\prime} &= 
     \bordermatrix{ & \mathcal{R} & \mathcal{L}_Z & \mathcal{L}_X & \mathcal{E}_Z & \mathcal{Q}^Z_{0} & \mathcal{Q}^X_{0} & \mathcal{Q}^Z_{1} & \mathcal{Q}^X_{1} & ... & \mathcal{Q}^X_{d-2} & \mathcal{Q}^Z_{d-1}  & \mathcal{Q}^X_{d-1}  & \mathcal{E}_X \cr
       \mathcal{S}_X & B_Z & A_Z & C_Z & M_Z & & & & \cr
       \mathcal{U}_Z &  &  &  & N_Z & & & & \cr
       \mathcal{C}^{X}_{0} &  &  &  & T_Z & H_C & I & & & & & & I \cr
       \mathcal{C}^{X}_{1} & & & & &  &I & H_C & I && & & & \cr
       \vdots & & & & & & & &  & \ddots & & & \cr
       \mathcal{C}^{X}_{d-1} & & & & & & & & & & I & H_C & I \cr
       \mathcal{V}_{X} &  &  & F  & & & & & & & &  & P_X & G_X^{\intercal}\cr
       } \qquad
    \end{align}

Perform row operation $\mathcal{C}^{X}_{0} \rightarrow \mathcal{C}^{X}_{0} + \sum_{i=1}^{d-1} \mathcal{C}^{X}_{i} $.

\begin{align}
    \rightarrow & \qquad \bordermatrix{ & \mathcal{R} & \mathcal{L}_Z & \mathcal{L}_X & \mathcal{E}_Z & \mathcal{Q}^Z_{0} & \mathcal{Q}^X_{0} & \mathcal{Q}^Z_{1} & \mathcal{Q}^X_{1} & ... & \mathcal{Q}^X_{d-2} & \mathcal{Q}^Z_{d-1}  & \mathcal{Q}^X_{d-1}  & \mathcal{E}_X \cr
       \mathcal{S}_X & B_Z & A_Z & C_Z & M_Z & & & & \cr
       \mathcal{U}_Z &  &  &  & N_Z & & & & \cr
       \mathcal{C}^{X}_{0} &  &  &  & T_Z & H_C & 0 & H_C & 0 & & 0 & H_C & 0 \cr
       \mathcal{C}^{X}_{1} & & & & &  &I & H_C & I && & & & \cr
       \vdots & & & & & & & &  & \ddots & & & \cr
       \mathcal{C}^{X}_{d-1} & & & & & & & & & & I & H_C & I \cr
       \mathcal{V}_{X} &  &  & F  & & & & & & & &  & P_X & G_X^{\intercal}\cr
       } \qquad
\end{align}
\end{widetext}

In order to calculate the rank of $H_X^{\prime}$, we will consider the ranks of submatrices consisting of sets of rows, where it is clear that no row from one set can be expressed as a linear combination of rows from another set. 

First, consider the sets of checks labeled by $\mathcal{S}_{X}$. The rank is the same as the rank of the original check matrix $H_X$. Next, consider the checks from the gauge-fixing set $\mathcal{U}_Z$. We will use the fact that the dimension of any cycle basis of any graph with $E$ edges, $V$ vertices and $p$ connected components is given by its cyclomatic number, $|E| - |V| + p$. Recall the $X$ checks in $\mathcal{U}_Z$ are added to deal with redundancies due cycles present in the new $Z$ checks. Each $X$ check corresponds to a cycle in the graph given by edges (qubits) in $\mathcal{E}_Z$ and vertices in $\mathcal{V}_Z$. Since there is only one connected component ($p=1$), $\mathrm{rank}(N_Z) = |\mathcal{E}_Z| - |\mathcal{V}_Z| + 1$.

Next, we consider the check matrix rows corresponding to $\mathcal{C}^{X}_{0}$ checks. The check matrix rows of $\mathcal{C}^{X}_{0}$ now only consists of $T_Z$ and $H_C$ as submatrices. These will determine $\mathrm{rank}(\cx_{0})$, where $\mathrm{rank}(\cx_{0}) = \max(\mathrm{rank}(T_Z),\mathrm{rank}(H_C))$, as the rank of a matrix is equal to the number of rows of the largest square submatrix that has a nonzero determinant. We know $T_Z$ is a $|\mathcal{V}_Z| \times |\mathcal{E}_Z|$ matrix, and further, that one of the rows of $T_Z$ is redundant because the $\mathsf{SkipTree}$ algorithm returns to the original root vertex. We know the rank of $H_C$ is also $|\mathcal{V}_Z| - 1$, since it has $d_Z = |\mathcal{V}_Z|$ rows in total, with one of the checks redundant.

Next, consider the checks corresponding to sets $\cx_{i}$ for $1 \leq i \leq d$. These rows are all identical modulo a cyclic shift, as they each have support $I$ on $\mathcal{Q}^{X}_{i-1}$ from the previous layer, $I$ on $\qx_{i}$ and $H_C$ on $\qz_{i}$ within the same layer. Notice these rows are linearly independent of each other, for which it is sufficient to observe that their qubit support remains distinct even after row-column operations. The rank of the submatrix consisting of these rows is then the sum of the ranks of the individual sets of rows. Thus we obtain $\sum_{i=1}^{d-1} \mathrm{rank}(\mathcal{C}^{X}_{i}) = (d-1)|\mathcal{V}^X|$. 
\begin{align}
    \mathrm{rank}(H_X^{\prime}) &= \mathrm{rank}(\mathcal{B}^X) + \mathrm{rank}(N_Z) + \mathrm{rank}(H_C) \nonumber \\ & \qquad + \; \sum_{i=1}^{d-1} \mathrm{rank}(\mathcal{C}^{X}_{i}) + \mathrm{rank}(\mathcal{V}^{X}) \nonumber \\
    &= \mathrm{rank}(\mathcal{B}^X) + (|\mathcal{E}_Z| - |\mathcal{V}_Z| + 1) + (|\mathcal{V}_Z| - 1) \nonumber \\
    & \qquad + \; (d-1)|\mathcal{V}_X| + |\mathcal{V}_{X}| \nonumber \\
    &= \mathrm{rank}(H_X) + |\mathcal{E}_Z| + d|\mathcal{V}_X|
\end{align}

Similarly,
\begin{align}
    \mathrm{rank}(H_Z^{\prime}) = \mathrm{rank}(H_Z) + |\mathcal{E}_X| + d|\mathcal{V}_Z|
\end{align}

The number of logical qubits $k^{\prime}$ in the deformed code
\begin{align}
    k^{\prime} &= n^{\prime} - \mathrm{rank}(H_X^{\prime}) -\mathrm{rank}(H_Z^{\prime}) \nonumber \\
    &= (n + |\mathcal{E}_X| + |\mathcal{E}_Z| + d|\mathcal{V}_X| + d|\mathcal{V}_Z|)  \nonumber \\ &\qquad - \mathrm{rank}(H_X^{\prime}) - \; \mathrm{rank}(H_Z^{\prime}) \nonumber \\
    &= n - \mathrm{rank}(H_X) - \mathrm{rank}(H_Z) + |\mathcal{E}_X| + |\mathcal{E}_Z| + d|\mathcal{V}_X| \nonumber \\
    & \qquad + \; d|\mathcal{V}_Z| - |\mathcal{E}_X| - d|\mathcal{V}_Z| - |\mathcal{E}_X| - d|\mathcal{V}_Z| \nonumber \\
    &= k
\end{align}

\end{proof}

\tocless\subsubsection{Initial code deformation step}\label{subsec:code-deformation}

The standard approach to code deformation into a merged code that includes the original code augmented by ancillary qubits is to begin measuring stabilizers of the new code. In order to merge into the new code without reducing the code distance, the initial state on the ancillary qubits in $\mathcal{E}_Z$ and $\mathcal{E}_X$ (which will support the edges of auxiliary graphs in the merged code) and qubits in sets $\qz_i$ and $\qx_i$ for $i \in [d]$ (which will support the toric code adapter in merged code), needs to be chosen carefully. 

First consider the initial state on the edge qubits. The simplest starting state to prepare in practice is a product state, as it simply consists of physical qubits without any encoding. Not all arbitrary product state would be ideal candidate initial states. For instance, initializing $\mathcal{E}_{Z}$ in $\ket{0}^{\otimes |\mathcal{E}_{Z}|}$ with single-qubit $Z$ stabilizers $Z(e_j \in \mathcal{E}_Z)$ or initializing $\mathcal{E}_{X}$ in $\ket{+}^{\otimes |\mathcal{E}_{X}|}$ with single-qubit $X$ stabilizers $X(e_j \in \mathcal{E}_X)$ is not ideal because these single-qubit stabilizers do not commute with the original LDPC code stabilizers of the opposite type they will merge with ($\mathcal{S}_{X}$ and $\mathcal{S}_{Z}$ respectively). The outcome of these stabilizer measurements is no longer deterministic during the merging step, and hence information from original code stabilizer measurements prior to the merging step is lost. 
For a fault-tolerant merging process, instead we introduce single-qubit stabilizers which commute with the original code stabilizers they overlap with, so $X(e_j \in \mathcal{E}_Z)$ for $j \in [|\mathcal{E}_{Z}|]$, and $Z(e_j \in \mathcal{E}_X)$ for $j \in [|\mathcal{E}_{X}|]$, that is, qubits in sets $\mathcal{E}_Z$ are initialized as $\ket{+}$ states and qubits in sets $\mathcal{E}_X$ are initialized as $\ket{0}$ states.

The initial state of the toric code ancilla state involves some subtleties as well. In section \ref{subsubsec:static-gmerge-param} we have already seen that there are multiple logical representatives for $\bar{Z}^{(c)}$ and $\bar{X}^{(t)}$ across the $d$ layers. We do not want to introduce new stabilizers which could clean support from any of these representatives, as that would bring down the minimum distance of merged code. If one chooses to begin with a product state of single-qubit stabilizers, some possibilities can be precluded by inspection: initializing $\qz_i$ in $\ket{0}$ states with single-qubit $Z$ stabilizers or initializing $\qx_i$ in $\ket{+}$ with single-qubit $X$ stabilizers. This is because if $\qz_i$ is initialized with single-qubit $Z$ stabilizers, any of these stabilizers cleans support from a $\bar{Z}^{(c)}$ logical representative $Z(\vec{1} \in \qz_i)$ for any layer $i$ ($\because$ Eqns~\ref{eqn:zrep_on_qz0}, \ref{eqn:zrep_on_qzi}), and reduces the minimum distance of the merged code. In fact, the product of all single qubit $Z$ stabilizers from a set $\qz_i$ could clean $O(d)$ support or even the entire $\bar{Z}^{(c)}$ logical. The same issue arises if qubits in $\qx_i$ are initialized with single-qubit $X$ stabilizers, as any of these stabilizers cleans support of $\bar{X}^{(t)}$ representative $X(\vec{1} \in \qx_i)$ and reduces the $X$ distance. For this reason, we are constrained to initialize a state that does not contain operators $Z(s \in \qz_i)$ or $X(s \in \qx_i)$ for any $s \in \mathbb{F}_2^{d}$, in any layer $i$, within its stabilizer space. One way to ensure this is to instead initialize with single-qubit $X$ and $Z$ stabilizers of the opposite type on these qubits, i.e. $X(e_j \in \qz_i)$ and $Z(e_j \in \qx_i)$ for all $j \in [d]$, in each layer $i$, or in statevector terms, to initialize in the product state $\ket{+}^{\otimes d^{2}}\! \ket{0}^{\otimes d^{2}}$, where all qubit sets $\qz_{i}$ as $\ket{+}$ states and $\qx_{i}$ as $\ket{0}$ states. 

While the product state comprising $\ket{+}$ on qubits in $\qz_i$ and $\ket{0}$ on qubits in $\qx_i$ is straightforward to prepare, a snag is that a threshold for the merging process may not exist, due to the anti-commuting gauge operators present in the deformed subsystem code. When one starts measuring the stabilizers of the toric code, the initial single-qubit $X$ checks on $\qz_i$ anticommute with the new $Z$ toric code checks $\cz_i$ being measured. The anti-commuting operators imply that measurement outcomes are no longer deterministic. The only deterministic measurements are products of stabilizers $\prod_{i=0}^{d-1}Z(e_j \in \qz_i)$ and also $\prod_{i=0}^{d-1}X(e_j \in \qx_i)$, which are $O(d)$ size, thereby not providing enough information for error-correction.

A second approach is to directly add products of single-qubit $Z$ stabilizers or $X$ stabilizers, specifically $\prod_{i=0}^{d-1}Z(e_j \in \qz_i)$ and $\prod_{i=0}^{d-1}X(e_j \in \qx_i)$ to the stabilizer group for the initial state. These are also representatives of $\bar{Z}^{(2)}$ and $\bar{X}^{(1)}$ on the toric code, so this means initializing directly in the $\ket{\overline{+0}}$ logical state. The qubits are not in a product state but rather in a highly entangled state belonging to the toric code codespace. Note, even for choosing a suitable state within the toric codespace, one can rule out the $+1$ eigenspaces of both $\bar{Z}^{(1)}$ and $\bar{X}^{(2)}$ logicals, because once measured, the logical operators are added to the stabilizer group and the newly added stabilizers ($\bar{Z}^{(1)}$, $\bar{X}^{(2)}$ in this scenario) would clean the support of $\bar{Z}^{(c)}$ or $\bar{X}^{(t)}$ and reduce the distance of our merged code. This reduces the feasible subspace for initializing the toric code ancilla, from $4$ codestates to exactly one codestate, $\ket{\overline{+0}}$. We decide to blackbox the preparation of the toric code state. Methods to prepare such a state exist, for example by simply measuring the logical operators $\bar{X}^{(1)}$ and $\bar{Z}^{(2)}$ (using auxiliary graph surgery \cite{williamson2024gauging}, sec~\ref{sec-gauging-meas} or homomorphic measurements \cite{shuang2023homomorphic}) to prepare their simultaneous $+1$ eigenstate, applying Pauli corrections where necessary.


It is possible that the product state $\ket{+}^{\otimes d^{2}}\! \ket{0}^{\otimes d^{2}}$ on the toric code is effective for small-size demonstrations, despite the code deformation step lacking a threshold, since errors can still be effectively detected. For a scheme applicable in the asymptotic regime, we use $\ket{\overline{+0}}$ as the initial state in the toric code adapter.

With the edge qubits and preprepared toric codestate initialized as described, code deformation takes place by measuring stabilizers of the new merged code (Tanner graph in Fig.~\ref{fig:unitary_adapter_d_layer}). At the end of the protocol duing the split step $(iii)$, we measure the initial stabilizers again, that is, stabilizers of the original LDPC code, single-qubit $X$ stabilizing edge qubits $\mathcal{E}_Z$, single-qubit $Z$ stabilizing edge qubits $\mathcal{E}_X$ and measure the logical operators $\bar{X}^{(1)}$ and $\bar{Z}^{(2)}$ on the toric code, applying corrections to reset the toric code to $\ket{\overline{+0}}$.

\tocless\subsection{Unitary $\cnotgate$ using Dehn twists in the toric code}
\label{app:dehntwist_unitary}

\tocless\subsubsection{Notation on permutations}
\label{subsec:convention_permutations}

Before we proceed, it is relevant to establish convention to distinguish between permutations on (qu)bits and checks for any one type of stabilizer for a CSS quantum code, or any classical code. Consider a set $\ccal$ of checks, and set $\qcal$ of (qu)bits. An individual check indexed $j \in [|\ccal|]$ is written as $e_j \in \mathbb{F}_2^{|\ccal|}$, where $e_j$ is the binary row vector with a $1$ only in the $j^{\textrm{th}}$ position. A qubit indexed $j \in [|\qcal|]$ is written as $e_j^{\top} \in \mathbb{F}_2^{|\qcal|}$. Let the check matrix $A \in \mathbb{F}_2^{|\ccal| \times \qcal|}$ describe the collection of edges in the Tanner graph between individual checks from set $\ccal$ and qubits from $\qcal$, where this collection itself can be viewed as an edge in the abstract Tanner graph between \textit{sets} of qubits and checks. Use row vector $u \in \mathbb{F}_2^{|\ccal|}$ to denote selections of checks from $\ccal$, and use column vector $v^{\top} \in \mathbb{F}_2^{|\qcal|}$ to denote selections of (qu)bits from $\qcal$, where $1$ indicates if the (qu)bit or check is present in the choice. Then $uA$ is a column vector that gives the (qu)bits supporting checks in $u$, and $Av^{\top}$ gives the syndrome corresponding to bitstring $v^{\top}$.  Permuting checks in $\ccal$ under $\pi: u \rightarrow u\pi$ is equivalent to left-multiplying $A$ by $\pi$, since $(u\pi)A = u(\pi A)$. In order to keep the product $uA$ (bits supporting checks $u$) invariant, a permutation on checks by $\pi$ needs to be accompanied by left-multiplication of $A$ by $\pi^{\top}$ and vice-versa, as $uA = (u\pi)(\pi^{\top}A)$ $(\because \pi\pi^{\top} = I)$. 

Next, note that permuting qubits in $\qcal$ under permutation $\sigma: v \rightarrow v\sigma$ or equivalently $v^{\top} \rightarrow (v\sigma)^{\top} = \sigma^{\top} v^{\top}$ transforms the syndrome the same way as right-multiplying $A$ by $\sigma^{\top}$ instead, since $A(\sigma^{\top} v^{\top}) = (A\sigma^{\top}) v^{\top}$. Again, it follows that to keep the syndrome invariant, a right-multiplication of $A$ by $\sigma$ must be accompanied by  permutation $\sigma$ on qubits $v$ and vice-versa, since $Av^{\top} \!= (A\sigma)(\sigma^{\top} v^{\top}) \; (\because \sigma\sigma^{\top}\!\!=\!I) = (A\sigma)(v\sigma)^{\top}$.

Overall, the effective transformation on edge $A$ of the Tanner graph, for check structure and syndrome to remain invariant after permutation $\pi$ on checks and $\sigma$ on qubits, is $ A \longrightarrow \pi^{\top} A \, \sigma$, as shown in Figure~\ref{fig:perm_convention}.
Conversely, if the edge $A$ transforms as $A \longrightarrow \pi^{\top} A \, \sigma$, then to preserve check structure and syndrome (and in the case where the syndrome is $0$, to remain in the same codespace), permutations $\pi$ on checks and $\sigma$ on qubits need to be applied.

\begin{figure}[h]
\centering
\captionsetup{justification=centering}
 \begin{tikzpicture}[]
        \draw[draw=gray,line width=0.6pt] (0,0) to (3,0);
        
        \node[draw, line width=0.7pt, fill=white, minimum size=0.6cm] at (-0.1, 0.1) {};
        \node[draw, line width=0.7pt, fill=white,minimum size=0.6cm] at (0, 0) {};
        \draw[fill=white,line width=0.7pt] (2.92,0.08) circle (0.35);
        \draw[fill=white,line width=0.7pt] (3,0) circle (0.35) node {};

        \node[] at (1.5, 0.3) {$A$};

        \node[] at (-0.1, 0.7) {\large $\pi$};
        \node[] at (3, 0.7) {\large $\sigma$};
        \node[] at (0, -0.6) {\large $\mathcal{C}$};
        \node[] at (3, -0.6) {\large $\mathcal{Q}$};
        \node[] at (0, -1.15) {$u \in \mathbb{F}_{2}^{|\mathcal{C}|}$};
        \node[] at (3.2, -1.15) {$v \in \mathbb{F}_{2}^{|\mathcal{Q}|}$};

\end{tikzpicture}
 \caption{Permuting qubits and checks: $A \longrightarrow \pi^{\top} A \, \sigma$}
\label{fig:perm_convention}
  \end{figure}

\tocless\subsubsection{Tracking stabilizer evolution}\label{subsec:dehn_twist_permutations}

Let us examine the effect of each set of transversal $\cnotgate$s on the quantum state initialized in the codespace of merged code, by describing their action on the stabilizer tableau \cite{aaronson04}. First consider the action of the circuit on $X$-type stabilizers, as $\cnotgate$ gates do not mix $X$ and $Z$ type operators. The action on $Z$-type stabilizers proceeds similarly, with controls and targets exchanged.
Consider the $i\!=\!0$ layer of the toric code. $X$-stabilizers from the set $\cx_{0}$ are transformed as
\begin{align}\label{eqn-X-stab-layer-0}
    \begin{pmatrix} T_Z & H_C & I \end{pmatrix} \; \xrightarrow{ \cnotgate(\mathcal{Q}_{0}^{Z},\mathcal{Q}_{0}^{X})}
\; &\begin{pmatrix} T_Z & H_C & I+H_C \end{pmatrix}  \nonumber \\
= &\begin{pmatrix} T_Z & H_C & C \end{pmatrix}
\end{align}
In the above we used the simple observation that $I+H_C = C$, where $C$ refers to the cyclic shift matrix by 1 with entries $\ket{j}\bra{j+1}{\pmod d}$. Thus Eq.~\ref{eqn-X-stab-layer-0} shows that applying physical $\cnotgate$s implements a transformation $I \rightarrow C$ on the edge between $\cz_{0}$ and $\qx_{0}$ in the Tanner graph. Following the reasoning in section~\ref{subsec:convention_permutations}, in order to preserve the codespace a permutation $\sigma_0$ (now additionally indexed $0$ to indicate layer $0$) is applied to qubits $v$ in $\qx_{0}$ such that
the desired net transformation of edge $I \rightarrow C\sigma_{0} = I$, which implies $\sigma_{0} = C^{-1}$. Equivalently, apply $\sigma_0^{\top}$ to $v^{\top}$  where $\sigma_0^{\top}\!=(C^{-1})^{\top}\!= C \;(\because C$ is unitary$)$. Applying the permutation, we get $C e_j^{\top} = e_{j+1}^{\top}$ for all $j \in [d]$ and so each qubit indexed $j$ is relabeled with index $j+1$. Thus transversal physical $\cnotgate$s followed by $C^{-1}$ on qubits in $\qx_0$ preserve the stabilizer space of $X$ stabilizers in $\cx_0$.

Meanwhile, $Z$-stabilizers from the set $\mathcal{C}_{0}^{Z}$ transform under action of these transversal $\cnotgate$s to
\begin{align}\label{eqn:Z-stab_permuted}
    \begin{pmatrix} I & H_C^{\top} & I \end{pmatrix} \; \xrightarrow{ \cnotgate(\mathcal{Q}_{0}^{Z},\mathcal{Q}_{0}^{X})}
\; &\begin{pmatrix} I + H_C^{\top} & H_C^{\top} & I \end{pmatrix}  \nonumber \\
= &\begin{pmatrix} C^{-1} \quad \; & H_C^{\top} & I \end{pmatrix}
\end{align}

\noindent Here the transversal $\cnotgate$s apply a cyclic shift $C^{-1}$ to the Tanner graph edge $I$ between $Z$ check set $\cz_{0}$ and the qubit set $\qz_{0}$. The  transformation on check matrix $I \rightarrow \pi_0^{\top}C^{-1} = I$ must be accompanied either by a permutation $\pi_0 =C^{-1}$ on checks in $\cz_{0}$, or $C$ on qubits in $\qz_{0}$. Since the labels for qubits in $\qz_{0}$ are already fixed by $P_Z$, this time we incorporate a permutation $\pi_0$ into the set $\cz_{0}$ of $Z$ checks.

One can verify that the remaining edge in layer $i\!=\!0$, the edge between qubits from $\qx_{0}$ (permuted by $\sigma_0=C^{-1}$) and checks in $\cz_{0}$ (permuted by $\pi_0=C^{-1}$), remains invariant under these permutations. The overall transformation on edge $(\qx_{0},\cz_{0})$ is given by
\begin{align}
    H_C^{\top} \longrightarrow \; & (C^{-1})^{\top} H_C^{\top} \, C^{-1} \tag{$\because A \rightarrow \pi^{\top}A\,\sigma $} \\
    = \; &C \, (I+C^{-1}) \, C^{-1} \nonumber \\ = \; &I+C^{-1} = \; H_C^{\top} \label{eqn-permutation-ccl}
\end{align}

 \noindent In the above, we used the fact that $H_C^{\top}=I+C^{-1}$ and $C^{\top}\!\!=C^{-1}$ because $C$ is unitary. 

 To summarize, physical $\cnotgate(\qz_0,\qx_0)$ followed by permutations  $\sigma_0$  on qubits in $\qx_0$ and $\pi_0$ on $Z$-checks in $\cz_0$, preserve all original Tanner graph edges belonging to the $0^{\textrm{th}}$ layer of merged code, where
 \begin{align}\label{eqn-sigma-pi-layer0} 
     \sigma_0 = C^{-1} \; \text{ and } \; \pi_0 = C^{-1}.
 \end{align}

Qubits $\qx_{0}$ and checks $\cz_{0}$ also have outgoing edges to $\cx_{1}$ and $\qz_{1}$ in the next primal layer, $i=1$, described by check matrices equal to the identity matrix $I$ as per the definition~\ref{defn-toric-ancilla} (see Fig~\ref{fig:cyclic_permutations_codespace}). 
To ensure this check structure remains preserved even after permutations $\sigma_0$ and $\pi_0$ on $\qx_{0}$ and $\cz_{0}$, we need to apply adjustment permutations $\pi_1^{\prime}$ on check set $\cx_{1}$ and $\sigma_1^{\prime}$ on qubit set $\qz_{1}$. Here we have denoted permutations in the primal layers with a prime$^\prime$. Specifically we want $\pi_1^{\prime}$ on $\cx_1$ such that $I \rightarrow (\pi^{\prime}_{1})^{\top} I\,\sigma_0 = I$ and also $\sigma_1^{\prime}$ on $\qz_1$ such that $I \rightarrow \pi^{\top}_{0}I\,\sigma_1^{\prime} = I$. This is satisfied when $\pi_{1}^{\prime} = \sigma_0$ and $\sigma_1^{\prime} = \pi_{0}$. From Eq.~\ref{eqn-sigma-pi-layer0} we know $\sigma_0=\pi_0=C^{-1}$. Therefore, permutations $\pi_1^{\prime}$ on checks in $\cx_{1}$ and $\sigma_1^{\prime}$ on qubits in $\qz_{1}$, restore the edges connecting layers $i\!=\!0$ and $i\!=\!1$ of the toric code, where
\begin{align}\label{eqn-sigma-pi-prime-layer1} 
     \sigma_1^{\prime} = C^{-1} \; \text{ and } \; \pi_1^{\prime} = C^{-1}.
 \end{align}

\noindent By extension of the same reasoning, it is clear that in order to preserve the edges $I$ connecting dual layer $i$ and primal layer $i+1 {\pmod 2}$, any permutations $\sigma_i$ on $\qx_i$ qubits and $\pi_i$ on $\cz_i$ checks in layer $i$ necessitate permutations $\pi_{i+1}^{\prime}$ on $\cx_i$ checks and $\sigma_{i+1}^{\prime}$ on $\qz_i$ qubits in the next layer $i+1$, given by:
\begin{equation}\label{eqn-perm-diff-layers}
    \sigma_{i+1}^{\prime} = \pi_i \; \text{ and } \; \pi_{i+1}^{\prime} = \sigma_i
\end{equation}

Returning to the stabilizer tableau, we see that $X$-stabilizers $\in \cx_i$ in layers $i=1...d-1$ evolve under transversal $\cnotgate$s in layer $i$ in their local of frame of reference (i.e. with respect to qubits from $\qx_i, \qx_{i-1}, \qz_i$ in their support prior to any permutations) in a manner similar to $X$ stabilizers in $\cx_0$ (Eq.~\ref{eqn-X-stab-layer-0}),
\begin{align}
    \begin{pmatrix} I & H_C & I \end{pmatrix} \; \xrightarrow{ \cnotgate(\qz_{i},\qx_{i})}
\; &\begin{pmatrix} I & H_C & I+H_C \end{pmatrix}  \nonumber \\
= &\begin{pmatrix} I & H_C & C \end{pmatrix}
\end{align}
This implies that the net permutation $\sigma_i$ required on qubits in $\qx_{i}$ would need to undo the shift $C$ introduced by the transversal $\cnotgate$ gates in the $i$th layer, in addition to any permutation $\pi_{i}^{\prime}$ on checks in $\cx_{i}$. The desired net transformation of edge $I \rightarrow (\pi_{i}^{\prime})^{\top}C\sigma_{i} = I$ which implies $\sigma_{i} = C^{-1}\pi_{i}^{\prime} = C^{-1}\sigma_{i-1} \; (\because \text{ Eq.~}\ref{eqn-perm-diff-layers})$. This decouples the recurrence relation in Eq.~\ref{eqn-perm-diff-layers} to give a simple recurrence relation for $\sigma_{i}$ for which we have already seen the base case, $\sigma_{0} = C^{-1}$ for layer $0$. Solving, we get the closed form $\sigma_i = (C^{-1})^{i}C^{-1} = C^{-i-1}$. Thus the net permutation $\sigma_i$ on qubits $\qx_i$ in each layer is given by $C^{-i-1}$. The permutations $\pi_i$ on checks in $\cz_{i}$ follow by tracking evolution of $Z$ stabilizers $\in \cz_i$. This completes the set of permutations in the dual layers. Applying Eq.~\ref{eqn-perm-diff-layers} we also obtain closed forms for permutations $\sigma_i^{\prime}$ and $\pi_i^{\prime}$ in the primal layers. In summary,
\begin{align}\label{eqn-perm-layeri}
    \sigma_{i} = C^{\, -i-1} \; &\text{ and } \; \pi_{i} = C^{\, -i-1} \\
    \sigma_{i}^{\prime} = C^{\, -i} \quad \, &\text{ and } \; \pi_{i}^{\prime} = C^{\, -i}
\end{align}

The required permutations on all sets of qubits and checks in each layer to remain within the codespace are shown in Figure~\ref{fig:cyclic_permutations_codespace}. Note as a result of these permutations $\sigma_i$ on qubits in $\qx_i$, the targets for transversal $\cnotgate$s are shifted by one index in each successive layer, shifted in total by $i$ for layer $i$. Instead of tranversal $\cnotgate$s being implemented between qubit $j$ as control and qubit $j$ as target for all $j$, now the transversal gates act between qubit $j$ control and qubit $j\!+\!1$ as target in layer $1$,  qubit $j\!+\!2$ as target in layer $2$ and so on till qubit $j\!+\!i \pmod{d}$ as target in layer $i$.

\tocless\subsection{Logical action of the Dehn twist}

\tocless\subsubsection{$\overline{\cnotgate}$ via a Dehn twist on the toric code}
\label{subsec:dehntwist_toric}

We first state and verify the logical action for Dehn twists on the toric code in Lemma~\ref{thm:dehntwist_toric_action}, through the lens of our notation.

\begin{replemma}{thm:dehntwist_toric_action}
Consider the toric code, a $[[2d^2,2,d]]$ CSS code with Tanner graph described as in Definition \ref{defn-toric-ancilla}. A $\overline{\cnotgate}_{12}$ gate can be performed between its two logical qubits (w.l.o.g.$\!$ $1$ is chosen to be the control qubit) using the following circuit
\begin{align}
    \overline{\cnotgate}_{12} =  \; \prod_{i=0}^{d-1} \, 
    \mathsf{\Gamma}_i^{\prime}(\qz_{i}) \;
    \mathsf{\Gamma}_i(\qx_{i}) \; {\cnotgate}(\mathcal{Q}^{Z}_{i},\mathcal{Q}^{X}_{i})
\end{align}
where $\mathsf{\Gamma}_i(\mathcal{A})$ is a unitary corresponding to permutations by $C^{-i-1}$ on qubits $e_j \in \mathcal{A}$, and $\mathsf{\Gamma}_i^{\prime}(\mathcal{A})$ is a unitary corresponding to permutations by $C^{-i}$ on qubits $e_j \in \mathcal{A}$, as derived in sec~\ref{subsec:dehn_twist_permutations} and ${\cnotgate}(\mathcal{A},\mathcal{B})$ is shorthand for transversal $\cnotgate$ gates between qubits in ordered sets $\mathcal{A}$ and $\mathcal{B}$, $|\mathcal{A}|=|\mathcal{B}|$, defined as in Eq.~\ref{eqn:shorthand-cnot-notation}.
\end{replemma}

\begin{proof}
 The desired action of $\overline{\cnotgate}_{12}$ on the Pauli logicals of the toric code is given by
\begin{subequations}
\begin{align}
    \bar{X}^{(1)} &\rightarrow \bar{X}^{(1)} \bar{X}^{(2)} \label{eqn:map_x1_toric} \\
    \bar{Z}^{(1)} &\rightarrow \bar{Z}^{(1)} \label{eqn:map_z1_toric}
     \\
    \bar{X}^{(2)} &\rightarrow \bar{X}^{(2)} \label{eqn:map_x2_toric}\\
    \bar{Z}^{(2)} &\rightarrow \bar{Z}^{(1)} \bar{Z}^{(2)} \label{eqn:map_z2_toric}
\end{align}
\end{subequations}

Consider the initial logical operators of the attached toric code ancilla system, before the circuit described by Eq.~(\ref{eqn:toric_dehntwist_cnot}) is implemented. In the standard form, the logical representatives of the toric code are tensored $X$ and $Z$ operators supported on qubits along each topologically nontrivial loop in the original lattice or equivalently in the compact description Tanner graph as
\begin{subequations}
\begin{align}
    \bar{X}^{(1)} &= \prod_{i=0}^{d-1} X ({e}_0 \in \mathcal{Q}_i^{Z}) \label{eqn:x1_toric}\\
    \bar{X}^{(2)} &= X (\vec{1} \in \qx_0) \\
    \bar{Z}^{(1)} &= Z (\vec{1} \in \qz_0) \\
    \bar{Z}^{(2)} &= \prod_{i=0}^{d-1} Z ({e}_0 \in \qx_{i}) 
\end{align}
\end{subequations}

In order to see why Eq.~\ref{eqn:x1_toric} holds, note that $|\bar{X}^{(1)} \cap \bar{Z}^{(1)}| \equiv \phi_{1}$ is odd, and is at least 1. This implies $\phi_1$ qubits in each set $\qz_i$ fully supporting a $\bar{Z}^{(1)}$ representative also supports part of the $\bar{X}^{(1)}$ representative. We can multiply the weight-2 cyclic code $X$ checks in $\cx_i$ to clean even-sized support from $\phi_1$ to obtain a form of the logical where each set $\qz_i$ supports exactly \textit{one} qubit in $\bar{X}^{(1)}$. The cyclic code stabilizers in $\cx_i$ can also move this single-qubit support within $\qz_i$ such that $\bar{X}^{(1)}$ representative has support on the first qubit $e_0$ in every layer.

It is trivial to see that logical $\bar{X}^{(2)}$ is unchanged under Dehn twist circuit in Eq.~\ref{eqn:toric_dehntwist_cnot}. Recall $\bar{X}^{(2)}$, supported on all the qubits in set $\qx_0$, is entirely contained within the first ($i\!=\!0$) layer and so it is sufficient to consider $\cnotgate(\qz_{0},\qx_{0})$ implemented in layer $0$ alone. These gates are controlled on qubits  $e_j \in \qz_{0}$ for all $j \in [|\qz_{0}|]$, and are targeted on $e_j \in \qx_{0}$. Each physical $\cnotgate$ gate leaves an $X$-type operator on its target qubit unchanged, $X(e_j \in \qx_0) \rightarrow X(e_j \in \qx_0)$. The transversal $\cnotgate$s leave any $X$-type operator residing on qubits in set $\qx_0$ unchanged. 

Similarly, $\bar{Z}^{(1)}$ resides on all the qubits in set $\mathcal{Q}_0^{Z}$, and is unchanged under action of transversal physical $\cnotgate$s, as the $Z$ operators supported on control qubits of a physical $\cnotgate$ remain invariant.

Next we look at $\bar{X}^{(1)}$, specifically considering the representative in Eq.~\ref{eqn:x1_toric}. This representative resides on one qubit from each of the $d$ layers, and thus will be affected by the first physical $\cnotgate$ gate in each layer. Each physical $\cnotgate$ maps 
\begin{equation}
    X(e_0 \in \qz_i) \rightarrow X(e_0 \in \qz_i)X(e_0 \in \qx_i) \quad \textrm{ for all } i \in [d]
\end{equation}

Next, we apply permutation $C^{-i}$ on $\qx_i$ in each layer $i$. Under permutation by $C^{-i}$, each qubit $e_0^{\top}$ in $\qx_i$ is mapped to $e_{i}^{\top}$, leaving the resulting operator $X(e_0 \in \qz_i)X(e_i \in \qx_i)$. Recall, each $X$ check in any set $\cx_i$ can clean exactly one qubit from layer $i$ and add the exactly the same support to layer $i\!-\!1$. All the qubits $e^{\top}_{i}$ can be moved by stabilizer multiplication $\prod_{i=0}^{d} \prod_{i^{\prime}=0}^{i} (e_{i^{\prime}} \in \cx_i)$ to create $\sum_{i=0}^{d-1} e_i = \vec{1}$ support on $\qx_0$. The resulting operator is $X(e_0 \in \qz_i)X(\vec{1} \in \qx_0) = \bar{X}^{(1)} \bar{X}^{(2)}$. This verifies the logical map $\bar{X}^{(1)} \rightarrow \bar{X}^{(1)} \bar{X}^{(2)}$.

The last logical to verify the action of the circuit in Eq.~\ref{eqn:toric_dehntwist_cnot} is the $\bar{Z}^{(2)}$ logical. We pick a representative of $\bar{Z}^{(2)}$ which has exactly single qubit support in each of the $d$ layers, which we choose without loss of generality to be the first qubit in all qubit sets $\mathcal{Q}_i^{X}$. Since $Z$ operators flow from the target to the control of physical $\cnotgate$s, each physical $\cnotgate$ gate maps

\begin{equation}
    Z(e_0 \in \qx_i) \rightarrow Z(e_0 \in \qz_i)Z(e_0 \in \qx_i) \quad \textrm{ for all } i \in [d]
\end{equation}

After permutations $C^{-i-1}$ to qubits in sets $\qz_i$ in each layer, each qubit $e_0^{\top}$ in $\qz_i$ is mapped to $e_{i+1}^{\top}$, leaving the resulting operator $Z(e_{i+1} \in \qz_i)Z(e_0 \in \qx_i)$. Note that the cyclic shift $i+1$ is modulo $d$. For any pair of qubits in different layers, $i_1$ and $i_2$, for $i_1 \neq i_2$, $e_{i_1+1} \neq e_{i_2+1}$, and further, for any layer $i \in [d]$, the qubit indexed ${i+1}\pmod{d}$ in $\qz_i$ supports the resulting operator. Once again it is possible to multiply $Z$ stabilizers from $\cz_i$ to move the logical support qubit by qubit from all layers $i$ to layer $i\!=\!0$, to obtain resulting operator $Z(\vec{1} \in \qz_i)Z(e_0 \in \qx_i) = \bar{Z}^{(1)}\bar{Z}^{(2)}$. This verifies the logical map $\bar{Z}^{(2)} \rightarrow \bar{Z}^{(1)} \bar{Z}^{(2)}$.

The last thing to note for correctness of logical action is that the Dehn twist circuit described preserves the stabilizer space. Permutations $C^{-i}$ on qubits and checks in the primal layers and permutations $C^{-i-1}$ on qubits and checks in the dual layers restore the original Tanner graph edges of merged code after physical $\cnotgate$s, as shown in Appendix~\ref{subsec:dehn_twist_permutations}.
\end{proof}

\tocless\subsubsection{Addressable logical action of $\cnotgate$ using toric code adapter}
\label{app:toric_adapter_cnot_exsitu}

The following theorem shows how this can be useful for performing addressable gates on the LDPC code logicals.

\begin{theorem}\label{thm:ldpc-dehntwist-cnot}
Consider an arbitrary $[[n,k,d]]$ CSS LDPC code. Let $c$ and $t$ be any two logical qubits in this code. $\bar{Z}^{(c)}$ and $\bar{X}^{(t)}$ are assumed w.l.o.g.\ to be pairwise disjoint (see Lemma \ref{lem:supportlemma} in Appendix). Consider the merged code obtained after merging with the toric code adapter using auxiliary graphs $\mathcal{G}_{Z}$ and $\mathcal{G}_{X}$. The circuit defined below implements the desired targeted logical $\overline{\cnotgate}$ between the quantum LDPC logical qubits $c$ and $t$. 

\begin{align}\label{eqn:ldpc_dehntwist_cnot}
    \overline{\cnotgate}_{c\, t} = \, \; \; \prod_{i=0}^{d-1} \, 
    \mathsf{\Gamma}_i^{\prime}(\qz_{i}) \;
    \mathsf{\Gamma}_i(\qx_{i}) \; {\cnotgate}(\mathcal{Q}^{Z}_{i},\mathcal{Q}^{X}_{i})
\end{align}
where $\mathsf{\Gamma}_i(\mathcal{A})$ is a unitary corresponding to permutations by $C^{-i-1}$ on qubits $e_j \in \mathcal{A}$, and $\mathsf{\Gamma}_i^{\prime}(\mathcal{A})$ is a unitary corresponding to permutations by $C^{-i}$ on qubits $e_j \in \mathcal{A}$, and ${\cnotgate}(\mathcal{A},\mathcal{B})$ is shorthand for transversal gates between two ordered sets of qubits $\mathcal{A}$ and $\mathcal{B}$ defined in Eq.~\ref{eqn:shorthand-cnot-notation}.
\end{theorem}

\begin{proof}
    The reasoning follows from a straightforward reduction to implementing the Dehn twist gate on the toric code ancilla, i.e. $\overline{\cnotgate}_{c\, t} = \, \overline{\cnotgate}_{12}$. By Lemma~\ref{thm:dehntwist_toric_action}, $\cnotgate$ within the toric code is correctly implemented by the same circuit.

    Full proof below.
The desired logical action for the $\overline{\cnotgate}_{ct}$ gate is to conjugate the Pauli logicals as follows:
\begin{subequations}
\begin{align}
    \bar{X}^{(c)} &\rightarrow \bar{X}^{(c)} \bar{X}^{(t)} \\
    \bar{Z}^{(c)} &\rightarrow \bar{Z}^{(c)}  \label{eqn:action_ldpc_z2}\\
    \bar{X}^{(t)} &\rightarrow \bar{X}^{(t)} \\
    \bar{Z}^{(t)} &\rightarrow \bar{Z}^{(c)} \bar{Z}^{(t)}
\end{align}
\end{subequations}

$\bar{X}^{(c)}$ anti-commutes with $\bar{Z}^{(c)}$, and $\bar{Z}^{(c)}$ has a representation entirely contained in every layer of the new code merged code ($\because$ Eq.~\ref{eqn:zrep_on_qzi}). So $\bar{X}^{(c)}$ definitely has to pick up support on the toric code qubits in merged code, of size $|\bar{X}^{(c)} \cap \bar{Z}^{(c)}| \equiv \phi_{c} = 1 {\pmod 2}$ in any layer. In order to prove the logical action under circuit given in Eq.~\ref{eqn:ldpc_dehntwist_cnot}, it is apt to find a representative for $\bar{X}^{(c)}$ on merged code.

First, consider the case where one multiplies all $X$ checks from all sets $\cx_0$ to $\cx_{d-1}$ to get a new representative $\bar{X}^{(c)} \prod_{i=0}^{d-1} \mathcal{H}_{X}(\vec{1} \in \cx_i)$. The product $\prod_{i=0}^{d-1} \mathcal{H}_{X}(\vec{1} \in \cx_i)$ has zero support on qubits $\qx_{i}$ for all $i\in[d]$, since exactly two $X$ checks overlap on each qubit, from $\cx_i$ and $\cx_{i+1} \pmod{d}$. The product also does not gain any support on qubits in $\qz_i$, since all $\cx_i$ have even weight on $\qz_i$, since the Tanner edges between all these pairs of check and qubit sets have the structure of the cyclic code $H_C$, and $\vec{1}H_C = 0$. Hence the net effect of aforementioned product of all $X$ checks is to leave $\bar{X}^{(c)}$ unchanged.

Next consider multiplying a subset $u$ of checks from set $\cx_{i}$ from each layer instead of the entire set of checks. This adds support equal to $uH_c$, which is guaranteed to be even. Any product of these checks does not change the odd parity of the overlap between $\bar{X}^{(c)}$ and $\bar{Z}^{(c)}$.
Say for some choice of checks $u$ from $\cx_0$,
\begin{equation}
    \phi_c + u H_c = e_0
\end{equation}
The product of choice $u$ of checks from $\cx_i$ across all layers, described by $\prod_{i=0}^{d-1} \mathcal{H}_{X}(u \in \cx_i)$, is supported on the first qubit in every layer. Further, $\mathcal{H}_{X}(u \in \cx_0)$ adds support $s$ to $\mathcal{E}_{Z}$ where $s = uT_z$.
\begin{align}
    \prod_{i=0}^{d-1} \mathcal{H}_{X}(u \in \cx_i) &= \prod_{i=0}^{d-1} X(e_0 \in \qx_i) X(s \in \mathcal{E}_{Z}) \\
    &= \bar{X}^{(1)} X(s \in \mathcal{E}_{Z})  \\& \quad (\because \prod_{i=0}^{d-1} X(e_0 \in \qx_i) = \bar{X}^{(1)}) \nonumber
\end{align}

\noindent If we multiply $\bar{X}^{(c)}$ with $\prod_{i=0}^{d-1} \mathcal{H}_{X}(u \in \cx_i)$ we see the new representative $\bar{X}^{(c) \prime}$ on merged code has some support $s$ on the interface and additional support exactly described by $\bar{X}^{(1)}$ (toric code logical) on the toric code qubits $\qz_i$.
\begin{equation}\label{eqn:x_c_merged_rep}
    \bar{X}^{(c)} \xrightarrow{\; \; \prod_{i=0}^{d-1} \mathcal{H}_{X}(u \in \cx_i)} \, \bar{X}^{(c)} \, \bar{X}^{(1)} \, X^{\otimes s} \, = \, \bar{X}^{(c)\prime}
\end{equation}

Also note that $\bar{X}^{(t)}$ can be moved from $\lx$ in the original LDPC code to $\qx_{d-1}$ (Fig~\ref{fig:g_int_z}). As per the toric code adapter set up, $\qx_{d-1}$ also support $\bar{X}^{(2)}$. Thus 
\begin{equation}\label{eqn:x_t_merged_x2}
    \bar{X}^{(t)} = X(\vec{1} \in \qx_i) = \bar{X}^{(2)}
\end{equation}

It follows that the map $\bar{X}^{(1)} \rightarrow \bar{X}^{(1)}\bar{X}^{(2)}$ implies the map $\bar{X}^{(c)} (\bar{X}^{(1)}) X^{\otimes s} \rightarrow \bar{X}^{(c)} (\bar{X}^{(1)} \bar{X}^{(2)}) X^{\otimes s}$, which due to eqns \ref{eqn:x_c_merged_rep} and \ref{eqn:x_t_merged_x2} is the same as $\bar{X}^{(c) \prime} \rightarrow \bar{X}^{(c) \prime}\bar{X}^{(t)}$, which we set out to prove. The former is simply a logical map for the toric code logical $\bar{X}^{(1)}$ under transversal $\cnotgate$s and permutations, which we have already verified in Lemma \ref{thm:dehntwist_toric_action}.
From Eq.~\ref{eqn:x_t_merged_x2}, we also know $\bar{X}^{(t)}$ will be mapped identically as $\bar{X}^{(2)}$ under action of transversal $\cnotgate$, which from Lemma~\ref{thm:dehntwist_toric_action} is trivial.

Reasoning along the same lines, if we multiply $\bar{Z}^{(t)}$ with a suitable choice $u^{\prime}$ of toric code $Z$ checks from $\cz_i$, we obtain a $\bar{Z}^{(t)}$ representative that has some support $s^{\prime}$ on the interface, and support described exactly by the $\bar{Z}^{(2)}$ (toric code logical) on qubits $\qx_i$.
\begin{equation}
    \bar{Z}^{(t)} \xrightarrow{} \bar{Z}^{(t)} \, \bar{Z}^{(2)} \, Z^{\otimes s'}
\end{equation}
From Fig~\ref{fig:g_int_Z} we know there is a representative of $\bar{Z}^{(c)}$ which is entirely supported on $\qx_i$ for $i \in [d]$, which is also support of the toric code logical $\bar{Z}^{(2)}$. Thus we have exactly equal operators representing both
$\bar{Z}^{(c)} = Z(\vec{1} \in \qz_i) = \bar{Z}^{(2)}$ for any $i$. Since $\bar{Z}^{(1)}$ remains unchanged as the control of transversal $\cnotgate$s, it is evident that $\bar{Z}^{(c)}$ also is unchanged. Original code logicals $\bar{Z}^{(t)} \rightarrow \bar{Z}^{(c)} \bar{Z}^{(t)}$ straightforwardly Lemma~\ref{thm:dehntwist_toric_action}.
\end{proof}

\tocless\subsubsection{Code Distance of the merged LDPC-toric code
}\label{app:toric_adapter_distance}

\begin{reptheorem}{thm:dist-preserve-static-merge}
    The merged code has distance $d$ if the original code had distance $d$, provided auxiliary graphs $\mathcal{G}_Z$ and $\mathcal{G}_X$ satisfy the graph desiderata 0-3 in Theorem~\ref{thm:graph_desiderata}.   
\end{reptheorem}

\begin{proof}
The logical operators of the toric code could potentially reduce in distance when moved to the original code using vertex checks $\vz$. To deal with this and the remaining logicals in the LDPC code we invoke the alternate sufficient condition in Theorem~\ref{thm:expansionless_joint_toric} (derived in Section~\ref{subsec:replace_desideratum_4}), which allows for distance-preserving merging without requiring expansion in auxiliary graphs. Recall the sufficient condition for a valid (therein referred to as `left') code, in this case the toric code, is that 
weight of $\overline{Z}_l\overline{X}_l$ cannot be reduced to less than $2d_l$ by multiplying by stabilizers and logical operators of the left code other than $\overline{Z}_l$, $\overline{X}_l$, and $\overline{Z}_l\overline{X}_l$. The non-overlapping operators considered here will be $\bar{Z}^{(1)} = Z(\qz_i)$ and $\bar{X}^{(2)} = X(\qx_i)$, which are both weight $d$. As is evident, the supports of the minimum weight representatives of this operator are disjoint in which case the product $\overline{Z}^{(1)}\overline{X}^{(2)}$ exactly equal to $|\qz_i|+|\qx_i|$, which is $2d$. When considering the product $|\overline{Z}^{(1)}\overline{X}^{(2)}\xi|$ (where $\xi$ represents the stabilizer group), notice that any representative of $\overline{Z}^{(1)}$ would overlap $\overline{X}^{(1)}$ on at least one qubit, since they anti-commute. This is true for any representative of $\bar{X}^{(1)}$. In the canonical form, $\bar{X}^{(1)} = \prod_{i=0}^{d-1} X(e_0 \in \qz_i)$. That is, $\bar{X}^{(1)}$ is supported on the same index $e_0$ in each of the $d$ primal layers. Equivalent representatives can be obtained by choosing any index $e_j$ for $j \in [d]$. That is, $\bar{X}^{(1)} = \prod_{i=0}^{d-1} X(e_j \in \qz_i)$. This means $\bar{Z}^{(1)}$ needs to have support on at least one qubit at each index $e_j$ for $j\in [d]$ which can be distributed over any of the $d$ primal layers of the toric code, in order to anticommute with every $\bar{X}^{(1)}$ logical representative. Similarly, any representative of $\bar{X}^{(2)}$ needs to have a single qubit support in each of the $d$ dual layers. In total, this guarantees $|\overline{Z}^{(1)}\overline{X}^{(2)}\xi| \geq 2d$. This completes the sufficient condition to invoke Thm~\ref{thm:expansionless_joint_toric}, thereby proving the toric code distance is preserved in the deformed code without requiring any expansion in the auxiliary graphs of the LDPC code logicals. Ideally a toric code of distance $\max(\textrm{wt}(\bar{Z}^{(c)}),\textrm{wt}(\bar{X}^{(t)}))$ is chosen for the protocol.  
\end{proof}




\tocless\subsubsection{Distance preservation during the logical gate}

The next point to consider is the fault-tolerance property of the unitary circuit. In order to prove fault-tolerance of the circuit, one needs to also consider the scenario where the transversal $\cnotgate$s can be faulty. The notion of \textit{faults} and \textit{fault-distance} are defined as follows. A location in a circuit refers to any physical unitary gate, state preparation or measurement operation. A \textit{fault location} is a location in the circuit which performs a random Pauli operation following the desired operation. Each fault location in the circuit can introduce a Pauli error on qubits in its support according to a probability distribution. Each physical $\cnotgate$ gate in our circuit is a fault location which could potentially introduce correlated $2$-qubit Pauli errors. Since $\cnotgate$ gates do not mix $X$ and $Z$-type Pauli operators, it is sufficient to deal with only $X$-type errors. The \textit{fault-distance} of a circuit is $d$ if up to $d-1$ faults are guaranteed to lead to detectable errors. It turns out that the fault distance of the circuit is also $d$, formally stated below.

\begin{lemma}
\label{lem:dist-preserve-cnot-unitary}
    The fault distance of the Dehn twist circuit for $l$ layer toric code adapter is equal to the code distance $d$ of the original LDPC code if the toric code has $l\!=\!d$ layers.
\end{lemma}

\begin{proof} 

By definition, the fault-distance is upper bounded by the code distance $d$. A simple and sufficient case to illustrate this in our context is where exactly the first $\cnotgate$ in each of the $d$ layers fails, i.e. the physical gate controlled on ${e}_0 \in \qz_i$ (and targeted on ${e}_0 \in \qx_i$) fails for all $i \in [d]$, and introduces a Pauli $X$ error on the control qubits. The total error is described by the Pauli operator $\prod_{i=0}^{d-1} X(e_0 \in \qx_i)$. We know $\prod_{i=0}^{d-1} X(e_0 \in \qz_i)$ is a representative for the logical $\bar{X}_{1}$. Thus the error has mapped our quantum state to another state within the codespace, and would show trivial syndrome during error-correction despite a nontrivial error, i.e.\ not identity or a product of stabilizers. In other words, the error is \textit{undetectable}. Since $d$ faults mapped to this undetectable error, the fault-distance of the Dehn twist circuit is upper bounded at $d$.

Let us examine cases where any $\leq d-1$ physical $\cnotgate$ gates fail in this circuit consisting of $d^2$ physical $\cnotgate$s, implemented in a specific way as batches of $d$ transversal gates. These $d-1$ faults could be distributed in any manner over the $d$ layers. Failure of $d-1$ $\cnotgate$s could lead to upto $2d-2$ errors, although some of these are correlated and can be identified. Not all of the errors are harmful, since some of them are detectable or are equivalent to stabilizers. We want to check that if at most $d-1$ faults occur, any nontrivial error created is detectable. Notice each physical $\cnotgate$ acts between qubits of opposite types, i.e. one from set $\qz_i$ in the primal layer and one from set $\qx_i$ in the dual layer (alternatively viewed as a physical $\cnotgate$ acting between a qubit on a horizontal edge and a qubit on a vertical edge of a lattice in the topological picture). Each $\cnotgate$ that fails causes at most one error in $\qz_i$, and at most one error in $\qx_i$. Since the total $d-1$ $\cnotgate$s fail, there can be at most $d-1$ errors in the primal layers and at most $d-1$ errors in the dual layers. Next we claim that the two $\bar{X}$ logicals in the code have support on either at least $d$ qubits in primal layers or $d$ qubits in dual layers. This is because $\bar{X}_{1}$ logical operators need to have non-zero support on each of the $d$ primal layers, in order to anti-commute with every representative of $\bar{Z}^{(1)}$ in different layers, and similarly $\bar{X}^{(2)}$ needs non-zero support on $d$ qubits in the dual layers in order to anti-commute with $\bar{Z}^{(2)}$.

First consider the primal layer. Since errors within the primal layer are bounded to be weight $\leq d-1$, which is less than the total number of such layers, $d$, there is at least one layer which does not host an error. $d-1$ errors within the primal layer due to $d-1$ $\cnotgate$s thus cannot be equivalent to the $\bar{X}^{(1)}$ weight-$d$ operator and are either stabilizers or errors which can be detected by $Z$ stabilizers connected to this layer that anti-commute with the error operator. A similar applies for the $X$ error weight in the dual layers. Since there were at most $d-1$ errors in total across all dual layers to begin with, the error cannot be equivalent to the weight-$d$ $\bar{X}^{(2)}$ logical, and is thus either a product of stabilizers or detectable. Consequently, for the specific choice of controls and targets for transversal gates in the Dehn twist circuit, any $d-1$ faults or less cannot lead to errors that are equivalent to logical operators, and are thus detectable.
\end{proof}


\newpage

\end{document}